\documentclass[12pt,a4paper]{book}
\input epsf

\usepackage[top=1.25in, bottom=1.25in, left=1in, right=1in]{geometry}

\usepackage{graphicx,hyperref} 
\hypersetup{
     colorlinks   = true,
     citecolor    = red,
     linkcolor= blue
}
\usepackage{amssymb,amsmath,amsthm,amscd,amsfonts,mathrsfs,setspace}
\usepackage{float}
\usepackage{cite} 
\usepackage{algorithm}

\usepackage{algorithmic}

\usepackage{logicpuzzle}

\usepackage{caption}
\usepackage{color}
\usepackage{kantlipsum}
\usepackage{mathtools}

\allowdisplaybreaks

\usepackage[caption = true]{subfig}
\def\syz{\operatorname{syz}}

\def\Res{\operatorname{Res}}
\def\LT{\operatorname{LT}}

\def\tr{\operatorname{tr}}
\def\trfive{\tr_5}

\def\A#1#2{\la#1#2\ra}
\def\B#1#2{[#1#2]}
\newcommand{\feyn}[1]{#1\kern-0.45em/}

\def\Z{\mathbb Z}
\def\C{\mathbb C}
\def\CP{\mathbb{CP}}
\def\F{\mathbb F}
\def\Q{\mathbb Q}
\def\nn{\nonumber}
\def\half{\frac{1}{2}}
\def\K{\mathbb K}
\def\Fpoly{\F[z_1,\ldots z_n]}
\def\la{\langle}
\def\ra{\rangle}
\def\sp{\operatorname{span}}
\def\lex {{\it lex}}
\def\grlex {{\it grlex}}
\def\grevlex {{\it grevlex}}
\def\c{\hspace{-0.8mm}}

\def\e {\mathbf{e}}

\def\dnz{dz_1\wedge \ldots \wedge dz_n}

\def\GB{Gr\"obner basis}
\def\Buch{Buchberger's Algorithm}
\def\mm {{\sc Mathematica}}
 \def\Singular{{\sc Singular}}
\def\Macaulay {{\sc Macaulay2}}

\def\WLOG{without loss of generality}
\def\hw[#1]{{\color{blue} \it #1}}

\def\bz{\bar z}
\def\zs{(z_1,\ldots z_n)}
\def\xis{(\xi_1,\ldots \xi_n)}

\def\TF{{\mathbf T}}

\renewcommand{\algorithmicrequire}{\textbf{Input:}}

\newcommand{\BREAK}{\textbf{break}}

\newtheorem{thm}{Theorem}[chapter]

\newtheorem{definition}[thm]{Definition}
\newtheorem{proposition}[thm]{Proposition}
\newtheorem{example}[thm]{Example}

\newtheorem*{remark}{Remark}
\newtheorem{ex}{Exercise}[chapter]

\newcommand{\mathsym}[1]{{}}
\newcommand{\unicode}[1]{{}}

\newcommand{\avg}[1]{\langle#1\rangle}

\begin{document}

\title{Lecture Notes on \\ Multi-loop Integral Reduction \\and Applied
  Algebraic Geometry}
\author{Yang Zhang \thanks{Department of physics, ETH Z\"urich, Wolfgang-Pauli-Strasse 27, 8093 Z\"urich,
Switzerland}\thanks{Max Planck Institut f\"ur Physik, Werner Heisenberg
Institut, 80805 M\"unchen, Germany \href{mailto:yang.ithaca@gmail.com}{yang.ithaca@gmail.com}}}

\date{}
\maketitle
These notes are for the author's lectures, ``Integral Reduction and Applied
  Algebraic Geometry Techniques'' in School and Workshop on
  Amplitudes in Beijing 2016.

I would like to acknowledge Simon Badger, Niklas Beisert, Zvi Bern,
Emil J. Bjerrum-Bohr, 
Jorrit Bosma, Poul Damgaard, Simon
Caron-Huot, Lance Dixon, Claude Duhr, Bo Feng, Hjalte Frellesvig,
Alessandro Georgoudis, Johannes Henn, Rijun Huang, Harald Ita, David Kosower, Kasper
Larsen, Jianxin Lu, Qing Lu, Pierpaolo Mastrolia, Dhagash Mehta, Mads Sogaard, Martin Sprenger, Mingmin Shen,
Henry Tye, Gang Yang, Huafeng Zhang and Huaxing Zhu for useful
discussions in related directions. In particular, I express my
gratitude to Michael
Stillman for his excellent lectures on algebraic geometry at Cornell
University.

I also acknowledge the support provided by Swiss National Science Foundation,
via Ambizione grant (PZ00P2\_161341).
\tableofcontents{}


\chapter{Introduction}
From childhood, we know that integral calculus is more difficult
than differential calculus, moreover a multiple integral can be a
hard nut to crack. Multifold integrals appear ubiquitously in science
and technology, for example, to understand high order quantum
interactions, we have to deal with {\it multi-loop Feynman
  integrals}. For precision LHC physics, next-to-next-to-leading order
(NNLO) and even next-to-next-to-leading order (N3LO) contribution
should be calculated, to be compared with the experimental data. This
implies we have to compute two-loop or three-loop, 
  non-supersymmetric, frequently massive Feynman integrals.  It is a
tough task.

Recall that in college, when we get a complicated integral, usually we
do not compute it directly by brute force. Instead, we may first: 
\begin{itemize}
\item Reduce the integrand. For example, given a
  univariate rational function integral, we use partial fraction to
  split the integrand into a sum of fractions, each of which contains
  only one pole. 
\item Convert the integral to residue computations. For an analytic
  univariate integrand, sometimes we can deform the contour of
  integral and make it a residue computation. The latter is often much
  easier than the original integral. 
\item Rewrite the integral by integration-by-parts (IBP). 
\end{itemize}
All these basic techniques are used every day in high energy physics,
and in all other branches of physics. For example, Ossola, Papadopoulos
and Pittau (OPP) \cite{Ossola:2006us,Ossola:2007ax}  developed a systematic one-loop integrand reduction
method, in the fashion of partial fraction. This method reduces 
one-loop Feynman integrals to one-loop master integrals, whose
coefficients can be automatically extracted from tree diagrams by
unitarity analysis. Nowadays, OPP method becomes a standard programmable algorithm for
computing next-to-leading order (NLO) contributions. 

However, for two-loop and higher-loop Feynman integrals, these basic
techniques for simplifying integrals often become insufficient. For
example, 
\begin{itemize}
\item For multi-loop orders, a Feynman integrand is still a rational
  function, however, in {\it multiple variables}. In this case, it is
  not easy to carry out partial fraction or general integrand
  reduction. This new issue is the {\it monomial order}, the order of
  variables. Naive reduction results may be
  too complicated for next steps, integral computation or unitarity
  analysis. 
\item For multi-loop generalized unitarity, sometimes we have residues not from one
  complex variable, but from multiple complex variables. It is
  well-known that the analysis of several complex variables is much harder than
  univariate complex analysis. For example, 
  \begin{quote}
    (Hartog) Let $f(z_1,\ldots, z_n)$ be an analytic function in 
    $U \backslash \{P\}$, where $U$ is an open set of
    $\C^n$ ($n>1$) and $P$ is a
    point in $U$. Then $f(z_1,\ldots, z_n)$ is analytic in $U$. 
\end{quote} 
Hartog's theorem implies that any isolated singular point of a
multivariate analytic function is removable. Hence, non-trivial singular
points of multivariate analytic function have a much more
complicated geometric structure than those in univariate cases. Besides, multivariate Cauchy's theorem does not apply for the
case when analytic functions have zero Jacobian at the
pole. That makes residue computation difficult. For instance,
\begin{equation}
  \label{residue_question}
  \oint \oint_{\text{around (0,0)}} \frac{dz_1 dz_2}{(a z_1^3+z_1^2+z_2^2)(z_1^3+z_1
    z_2-z_2^2)} =~?
\end{equation}

\item For multi-loop integrals, the number of IBP relations becomes
  huge. We may need to list a large set of IBP relations, and then
  use linear algebra to eliminate unwanted terms to get {\it useful} 
  IBPs. However, the linear system can be very large and Gauss
  elimination (especially in analytic computations) may exhaust
  computer RAM.

 Is there a way to list only useful IBPs, by
  adding constraints on differential forms? The answer is ``yes'', but
  these constraints are subtle. These are linear equations which only allow
  polynomial solutions \cite{Gluza:2010ws}. \footnote{As an analogy, consider the equation $6 x+ 9 y=15$
  in $x$, $y$. If $x$, $y$ are allowed to be rational numbers, it is a
  simple linear equation. However, if only integer values for $x$,
  $y$ are allowed, it is a less-trivial Diophantine equation in number
  theory. Here we have polynomial-valued Diophantine equations.}
How do we solve them efficiently?
\end{itemize}
Most Feynman integral simplification procedures in multi-loop orders,
suffer from the complicated structure of multiple variables. Note that,
usually our
targets are just polynomials or rational functions. However, multivariate
polynomial problems can be extremely difficult. (One famous example is
Jacobian conjecture, which stands unsolved today.) 

The modern branch of mathematics dealing with multivariate polynomials and
rational functions is {\it algebraic geometry}. Classically, algebraic
geometry studies the geometric sets defined by zeros of polynomials. Polynomial
problems are translated to geometry problems, and vice versa. Note that since
only polynomials are allowed, algebraic geometry is more
``rigid'' than {\it differential geometry}. Classical algebraic
geometry culminates at the classification theorem of
algebraic surfaces by the {\it Italian school} in 19th century. 

Modern
algebraic geometry is rigorous, much more general and abstract. The
classical geometric objects are replaced by the abstract concept {\it
  scheme}, and powerful techniques like {\it homological algebra} and {\it
  cohomology} are introduced in algebraic geometry thanks to Alexander
Grothendieck and contemporary mathematicians \cite{MR3075000,MR0163909,MR0163911,MR0173675,MR0199181,MR0217086,MR0238860,MR0463157}. Modern algebraic
geometry shows its power in the proof of {\it Fermat's last theorem} by
Andrew Wiles. Now algebraic geometry applies on number theory,
representation theory, complex geometry and theoretical physics.

Back to our cases, there are numerous polynomial/rational function
problems. Clearly, they are not as sophisticated as {\it Fermat's last
  theorem} or {\it Riemann hypothesis}. Apparently they resemble
classical algebraic geometry problems. However, beyond the
classification of curves or surfaces, we need computational
power to solve polynomial-form equations, to compute multivariate residues in the real
world. The computational aspect of algebraic geometry, was neglected for a long time. 

When I was a graduate student, I was lucky taking a class by Professor
Michael
Stillman. One fascinating thing in the class was that many times after
learning an important theorem, Michael turned on the computer and ran
a program called ``Macaulay2'' \cite{M2}. He typed in number fields, 
polynomials, and geometric objects in the study. Then various commands in the program
can automatically generate the dimension, the genus and various maps between
objects. He taught us one essential tool behind the program was the
so-called {\it
  Gr\"obner basis}, which is the crucial concept in the new subject
{\it computational algebraic geometry} (CAG)
\cite{MR3330490,opac-b1094391}. It was my first time hearing
about CAG and soon found it useful.

CAG aims at multivariate polynomial and rational function problems in
the real world. It began with {\it Buchberger's algorithm} in 1970s,
which obtained the Gr\"obner basis for a {\it polynomial
  ideal}. Buchberger's algorithm for polynomials is similar to Gaussian
Elimination for linear algebra: the latter finds a linear
basis of a subspace while the former finds a ``good''
generating set for an ideal. With Gr\"obner basis, one can carry out
multivariate polynomial division and simplify rational functions; one
can eliminate variables from a polynomial system; one can apply polynomial
constraints without solving them... Then CAG developed quickly
and now it
is so all-purpose that people use it outside mathematics, like in
robotics, cryptography and game theory. I believe that CAG is crucial for the
deep understanding of multi-loop scattering amplitudes. 

Hence, the purpose of these lecture notes is to introduce a
fast-developing research field: applied algebraic geometry in multi-loop
scattering amplitudes. I would like to show CAG methods by
examples,
\begin{itemize}
\item Multi-loop integrand reduction via  Gr\"obner basis. This
  generalizes one-loop OPP integrand reduction method to all loop
  orders. In this section, I will introduce basic notations of
  polynomial ring, rudiments of algebraic geometry and the
  Gr\"obner basis method.
\item Multivariate residue computation, in {\it generalized} unitarity analysis. 
A flavor of several complex variables will be provided in the
section. Then I present the definition of multivariate residues and
CAG based algorithms for computing multivariate
residues. Finally I show that they are very useful in high-loop unitarity
analysis.
\item Multi-loop IBP with polynomial constraints. These constraints
  form a {\it syzygy} system, which can be solved by \GB\
  \cite{Gluza:2010ws} techniques. We show that we can combine this 
  with unitarity cuts and the Baikov representation
  \cite{Baikov:1996rk} to further improve the
  efficiency. 
\end{itemize}
I will illustrate mathematical concepts and methods by practical
examples and exercises, even beyond mathematics/physics, like the game {\it
  Sudoku}. The proof of many mathematical theorems will be skipped or
just roughly sketched. These notes do not cover other important topics
in amplitudes studies, like {\it Symbol}, {\it differential equation},
{\it Grassmannian}
or {\it bootstrap}. We refer to the beautiful online articles, for
instance, \cite{ Dixon:2011nj , Duhr:2011zq, ArkaniHamed:2012nw , Henn:2014qga 
 } for these topics. I will not cover all the technical details of the
research frontier from integral reduction, since I believe it is more
important for readers to get the idea of basic algebraic geometry and find its
applications in their own research fields.

\chapter{Integrand reduction and Gr\"obner basis}

\section{Basic physical objects}
In these notes we mainly focus on scattering amplitudes in perturbative
quantum field theory and (super-)gravity. To make the reduction
methods general, we aim at non-supersymmetric amplitudes. These
methods definitely work with supersymmetric theories, however, it is
more efficient to combine them with specific shortcuts in
supersymmetric theories.

Referring to an $L$-loop Feynman diagram, we mean a {\it connected} diagram with $n$ external 
lines, $P$ propagators, and $L$ {\it fundamental cycles} \footnote{We
  need some graph theory concepts here: for a graph $G$, a {\it
    spanning tree} $T$ is a tree subgraph which contains all vertices of
  $G$. Given any edge $e$ in $G$ which is not in $T$, we define a
  {\it fundamental cycle} $C_e$
as the
  simple cycle which consists of $e$ and a subset of $T$. The number of
  fundamental cycles is independent of the choice of $T$.}. We further require that
each external line is connected to some fundamental cycle. Define $V$ as the number
of vertices in this diagram, then the graph theory relation holds,
\begin{equation}
  \label{Graph_theory}
  L=P-V+1.
\end{equation}
Note that this relation is not Euler's famous formula, since this
relation holds for both planar and nonplanar graphs in graph theory,
but Euler characteristic does not enter this relation.  (Of course, for a
planar graph, by embedding fundamental cycles into a
plane as face boundaries, it becomes Euler's formula for planar
graphs.)

For gauge theories, we have {\it color-ordered} Feynman diagrams such
that the external color particles must be drawn from infinity in a given cyclic
order, and the Feynman rules would differ from the unordered
ones. Sometimes, with these constraints, we cannot draw a Feynman
diagram on a plane without crossing lines. We call such
a Feynman diagram a {\it nonplanar diagram in the sense of color
  ordering.} Note that this definition is different from {\it
  nonplanar diagram in the sense of graph theory}, since by lifting the color
order constraint, a colored-ordered nonplanar diagram may be embedded
into a plane without crossing lines. See an example in  Fig.\ref{xbox_planar}.
\begin{figure}[h]
\subfloat[a nonplanar diagram with color ordering]{\includegraphics[width = 3in]{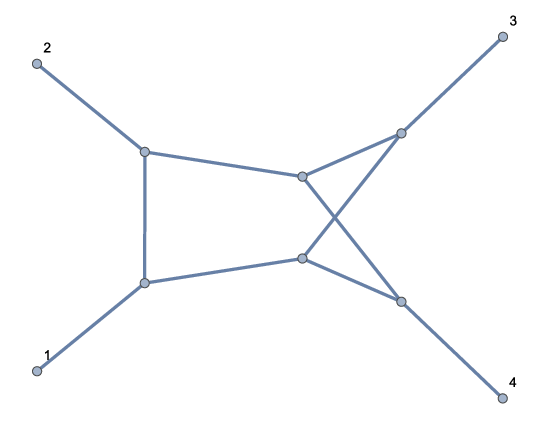}} 
\subfloat[Redraw Fig.\ref{xbox_planar}a by neglecting external line
color ordering]{\includegraphics[width = 3in]{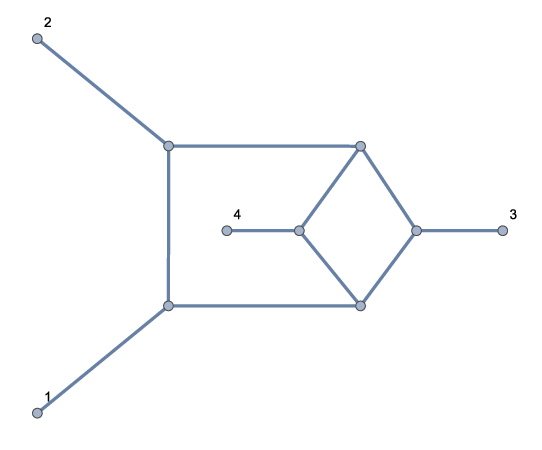}}
\caption{A nonplanar diagram in color
  ordering, may be a planar diagram in the sense of graph theory. }
\label{xbox_planar}
\end{figure}

Sometimes, for an $L$-loop diagram with $L>1$, two fundamental cycles do not
share a common edge. In this case the diagram is {\it factorable},
i.e., factorized into two diagram. We consider a factorable diagram as
two lower loop-order diagrams, instead of an ``authentic'' $L$-loop
diagram. See an example in Figure \ref{reducible_diagrams}a. For a $n$-point $L$-loop
diagram, if two external lines attach to one vertex, we consider this diagram as an
$n-1$ point diagram. See an example
in Figure \ref{reducible_diagrams}b.
\begin{figure}[h]
\subfloat[a factorable diagram.]{\includegraphics[width = 3in]{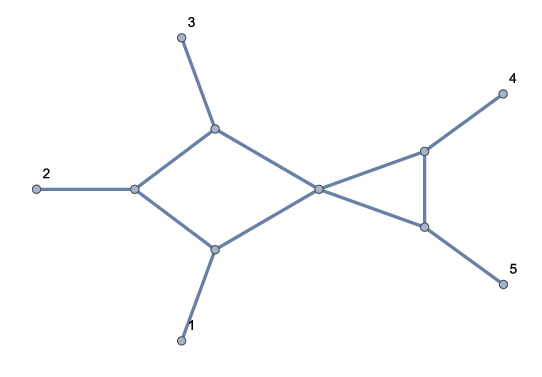}} 
\subfloat[This diagram
is considered as a $4$-point diagram instead of a $5$-point diagram.]{\includegraphics[width = 3in]{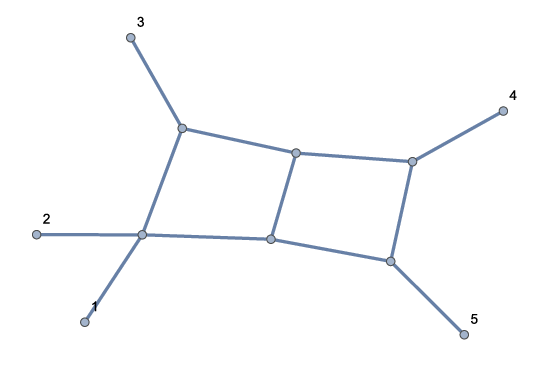}}
\caption{Diagrams to be simplified}
\label{reducible_diagrams}
\end{figure}

A Feynman diagram has the associated Feynman integral,
\begin{equation}
  I=\int \frac{d^D l_1}{i \pi^{D/2}} \ldots \frac{d^D l_L}{i \pi^{D/2}}
  \frac{N(l_1,\ldots l_L)}{D_1 \ldots D_P}\,.
\label{Feynman integral}
\end{equation}
For each fundamental cycle, we assign an internal momenta $l_i$.
Here the denominators of Feynman propagator have the form, $D_i=(\alpha_1 l_1 + \dots \alpha_l l_L+\beta_1 k_1 +\ldots
 \beta_n k_n)^2-m_i^2$. $k_1 \ldots k_n$ are the external momenta. $\alpha$'s
 must be $\pm 1$. For fermion propagators, we complete the denominator
 squares to get this form. $N(l_1,\ldots l_L)$ is the numerator, which
 depends on Feynman rules and the symmetry factor. Here we hide the dependence of
 external momenta/polarizations in $N(l_1,\ldots l_L)$. The spacetime dimension 
 $D$ may take the value $4-2\epsilon$ in the
 {\it dimensional regularization scheme} (DimReg).  Sometimes we also
 discuss the case $D=4$ or some other fixed integer, for studying leading
 singularity and maximal unitarity cut.

\section{Integrand reduction at one loop}
Consider the problem of reducing the integrand in (\ref{Feynman
  integral}) before integration. Schematically integrand reduction, as
a generalization of partial fractions, is to express the numerator $N$ as, 
\begin{equation}
  N= \Delta + \sum_{j=1}^P h_j D_j,
\label{IntRed}
\end{equation}
where $\Delta$ and $h_j$'s are polynomials in loop momenta components.
The term $h_j D_j$ cancels a denominator $D_i$ and provides a Feynman
integral with fewer propagators. Then this term merges with other
Feynman integrals in the scattering amplitude. $\Delta$ remains for 
this diagram. If $\Delta$ is ``significantly simpler'' than $N$,
this integrand reduction is useful.    
\subsection{Box diagram}

To make our discussion solid, we first introduce
the classical OPP reduction method \cite{Ossola:2006us,Ossola:2007ax} at
one loop order. It is well known that if $D=4$, all one-loop Feynman integrals
with more than $4$ distinct propagators can be reduced to Feynman
integrals with at most $4$ distinct propagators, while if $D=4-2\epsilon$, one-loop Feynman integrals
with more than $5$ distinct propagators are reduced to Feynman
integrals with at most $5$ distinct propagators, at the integrand level.
These statements can be proven by tensor calculations
\cite{Melrose:1965kb}. Later in this section, we re-prove these by
a straightforward algebraic geometry argument. 

For simple $D=4$ cases, we only need to start from the box
diagram. For instance, consider $D=4$ four-point massless box,
\begin{figure}
\centering
\includegraphics[scale=0.8]{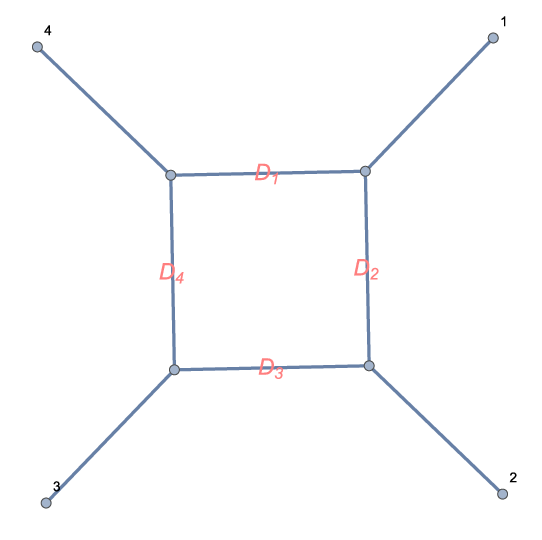}
\caption{One-loop massless box diagram}
\label{box}
\end{figure}
with denominators in propagators,
\begin{gather}
  D_1=l^2,\quad D_2=(l-k_1)^2, \quad D_3=(l-k_1-k_2)^2,\quad D_4=(l+k_4)^2.
\end{gather}
The Mandelstam variables are $s=(k_1+k_2)^2$ and $t=(k_1+k_4)^2$. It is useful to re-parameterize the loop momentum $l$ instead of using
its Lorentz components. There are several parametrization
methods: (1) van
Neerven-Vermaseren parameterization \cite{vanNeerven:1983vr} (2)
spinor-helicity parameterization (3) Baikov
parametrization. Here we use the straightforward van
Neerven-Vermaseren parameterization, and postpone applications of
other parameterizations later.

Note that by energy-momentum conservation, only external momenta
$k_1$, $k_2$ and $k_4$ are independent. To make a $4D$ basis, we
introduce an auxiliary vector $\omega_\mu\equiv\frac{2 i}{s}\epsilon_{\mu \nu \rho \sigma} k_1^\nu
k_2^\rho k_4^\sigma$ \footnote{The normalization is from the
  convention of spinor helicity formalism. So $\omega$ is a pure
  imaginary vector, and later on the unitarity solutions appear to be
  real in van
Neerven-Vermaseren variables. }
\begin{equation}
  \label{eq:19}
  \omega^2 = -\frac{t (s+t)}{s} \,.
\end{equation}
 Then the basis $\{e_1,e_2 \ldots e_4\}\equiv
\{k_1,k_2,k_4,\omega\}$. The {\it Gram matrix} of this basis is,
\begin{equation}
  \label{eq:3}
G=\left(
  \begin{array}{cccc}
    0 & \frac{s}{2} & \frac{t}{2} & 0 \\
\frac{s}{2} & 0 & \frac{-s-t}{2} & 0 \\
\frac{t}{2} & \frac{-s-t}{2} & 0 & 0 \\
0 & 0 & 0 & -\frac{t (s+t)}{s} \\
  \end{array}
\right),\quad G_{ij}=e_i\cdot e_j\,.
\end{equation}
Note that for any well-defined basis, Gram matrix should be
non-degenerate. For any $4D$ momentum $p$, define van Neerven-Vermaseren variables as,
\begin{gather}
  \label{box_vNV}
  x_i(p)\equiv p \cdot e_i,\quad   i=1,\ldots ,4\,.
\end{gather}
Then for any two $4D$ momenta, a scalar product translates to van
Neerven-Vermaseren form, by linear algebra
\begin{equation}
  \label{eq:4}
  p_1 \cdot p_2=\mathbf x(p_1)^T (G^{-1}) \mathbf x(p_2) \, ,
\end{equation}
where the bold $\mathbf x(p)$ denotes the column $4$-vector, $(x_1,x_2,x_3,x_4)^T$. Back to
our one-loop box, define $x_i\equiv
x_i(l)$. Hence a Lorentz-invariant numerator $N_\text{box}$ in \eqref{Feynman integral}
has the form,
\begin{equation}
  \label{eq:5}
  N_\text{box}=\sum_{m_1} \sum_{m_2} \sum_{m_3} \sum_{m_4} c_{m_1m_2m_3m_4}x_1^{m_1} x_2^{m_2} x_3^{m_3} x_4^{m_4}\,,
\end{equation}
For a renormalizable theory, there is a bound on the sum,
$m_1+m_2+m_3+m_4\leq 4$. The goal in integrand reduction is to expand
\begin{equation}
  N_\text{box}=\Delta_\text{box}+h_1 D_1 +\ldots h_4 D_4\,, 
\label{box_IR}
\end{equation}
 such that the remainder polynomial
$\Delta_\text{box}$ is as simple as possible.

Following \cite{Ossola:2006us}, the simplest $\Delta_\text{box}$ can be obtained
by a direct analysis. Note that
\begin{align}
  x_1 &= l \cdot k_1 =\half (D_1-D_2), \nn \\
  x_2 &= l \cdot k_2 =\half (D_2-D_3)+\frac{s}{2}, \nn \\
 x_3 &= l \cdot k_4 =\half (D_4-D_1),
\label{RSP_box}
\end{align}
hence $x_1$ and $x_3$ can be written as combinations of $D_i$'s,
while  $x_2$ is equivalent to the constant $s/2$ up to combinations of
 $D_i$'s. A scalar product which equals combinations of denominators
and constants is called a {\it reducible scalar product} (RSP). In this
cases, $x_1,x_2,x_3$ are RSPs. The
remainder $\Delta_\text{box}$ shall not depend on RSPs, hence, 
\begin{equation}
  \label{eq:7}
  \Delta_\text{box}= \sum_{m_4} c_{m_4} x_4^{m_4}.
\end{equation}
$x_4$ is called a  {\it irreducible scalar product} (ISP). Furthermore, using the expansion of $l^2$ and \eqref{RSP_box},
\begin{align}
 D_1=l_1^2&= \frac{1}{{4 s t (s+t)}}\big(-4 s^2 x_4^2 +s^2 t^2 +4 D_1 s^2 t-2 D_2
  s^2 t-2 D_4 s^2 t+D_2^2 s^2+D_4^2 s^2\nn\\
&-2 D_2 D_4 s^2
+2 D_1 s t^2-2
   D_3 s t^2+2 D_1 D_2 s t
-4 D_1 D_3 s t+2 D_2 D_3 s t+2 D_1 D_4 s t\nn\\
&-4 D_2 D_4 s t
+2 D_3
   D_4 s t+D_1^2 t^2+D_3^2 t^2-2 D_1 D_3 t^2 \big),
\label{box_quadratic_eqn}
\end{align}
which means
\begin{equation}
\label{box_quadratic_eqn_cut}
  x_4^2=\frac{t^2}{4} +\mathcal O (D_i).
\end{equation}
Hence quadratic and higher-degree monomials in $x_4$ should be removed
from the box integrand, and
\begin{equation}
  \Delta_\text{box}=c_0 + c_1 (l \cdot \omega).
\label{box_integrand_basis}
\end{equation}
This is the {\it integrand basis} for the $4D$ box, which contains only
$2$ terms. Note that by Lorentz symmetry,
\begin{equation}
  \label{eq:11}
  \int d^D l \frac{l\cdot \omega}{D_1 D_2 D_3 D_4} =0,
\end{equation}
for any value of $D$. So $c_1$ should not appear in the final
expression of scattering amplitude. We call such a term a {\it spurious term}.
But it is important for integrand
reduction, as we will see soon.

There are two ways of using the integrand basis
\eqref{box_integrand_basis},
\begin{enumerate}
\item Direct integrand reduction (IR-D). If the numerator $N$ is known, for instance from Feynman rules, we can use \eqref{RSP_box} and
  \eqref{box_quadratic_eqn} to reduce $N$ explicitly to get $c_0$ and
  $c_1$. $h_1D_1+\ldots h_4D_4$ is kept for further triangle,
  bubble ... computations.

\item Integrand reduction with unitarity (IR-U). Sometimes, it is
  more efficient to fit the coefficients $c_0$ and $c_1$ from tree
  amplitudes, by unitarity. Here $c_0$ and $c_1$ correspond to the
  remaining information at the quadruple cut,
  \begin{equation}
    \label{eq:12}
    D_1 = D_2 =D_3 =D_4=0.
  \end{equation}
From \eqref{RSP_box} and \eqref{box_quadratic_eqn_cut}, there are two
solutions for $l$, namely $l^{(1)}$ and $l^{(2)}$, characterized by,
\begin{align}
  \label{box_solution}
  \text{(1)}\quad x_1&=0,\quad x_2=\frac{s}{2},\quad x_3=0,\quad x_4=\frac{t}{2}\,,\\
 \text{(2)}\quad x_1&=0,\quad x_2=\frac{s}{2},\quad x_3=0,\quad
                        x_4=-\frac{t}{2}\,.
\end{align}
On this cut, the box diagram becomes four tree diagrams,
summed over different {\it on-shell} massless internal states. 
\begin{gather}
  \label{tree_product}
  S^{(i)}_\text{box}=\sum_{h_1} \sum_{h_2} \sum_{h_3}\sum_{h_4} A(k_1,
 l^{(i)}-k_1, -l^{(i)};s_1, h_2, -h_1)\times \nn\\
A(k_2, l^{(i)}-k_1-k_2,k_1-l^{(i)};s_2, h_3, -h_2) A(k_3,
l^{(i)}+k_4,k_1+k_2-l^{(i)};s_3, h_4, -h_3)\nn\\
\times A(k_4,l^{(i)},-k_4-l^{(i)};s_4,h_1,-h_4)\,,
\end{gather}
where $s_i$'s stand for external particles helicities, while $h_i$'s
stand for internal particles helicities and should be
summed. Unitarity implies that,
\begin{equation}
  \label{box_unitarity}\left\{
  \begin{array}{cc}
    c_0+\frac{t}{2}c_1=S^{(1)}_\text{box} & \\
    c_0-\frac{t}{2}c_1=S^{(2)}_\text{box} & \\
  \end{array}
\right. .
 \end{equation}
Generically, there is a unique solution for $(c_0,c_1)$. Here we see
the importance of the box integrand basis
\eqref{box_integrand_basis}. If there are fewer than $2$ terms in the
basis (oversimplified), then the integrand cannot be fitted from
unitarity.  If there are more than $2$ terms in the
basis (redundant), then the integrand will contain free
parameters which mess up the amplitude computation for following
steps.
\end{enumerate}

\subsection{Triangle diagram}
After the box integrand reduction is done, we proceed to the triangle
cases. Note that there are more than one triangle diagrams, in a
$4$-point scattering process, by pinching one internal line. Consider
this one,
\begin{equation}
  \label{eq:16}
  I=\int \frac{d^4l }{i \pi^2} \frac{N_\text{tri}}{D_1D_2D_3}\,,
\end{equation}
where external lines $3$ and $4$ are combined. 
\begin{figure}
\centering
\includegraphics[scale=0.8]{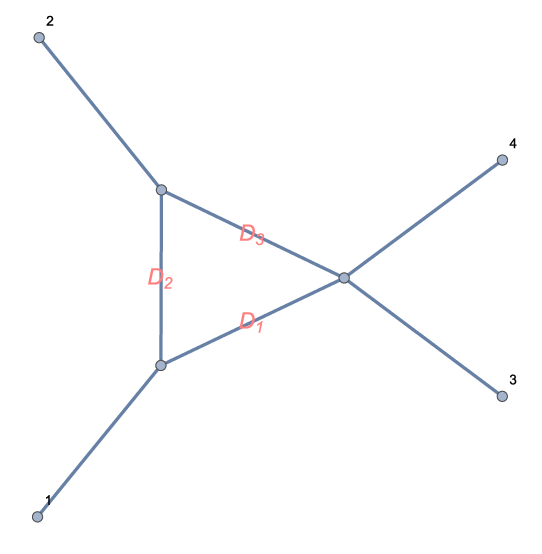}
\caption{One-loop triangle diagram}
\label{box}
\end{figure}
The kinematics is much
simpler than that of the box case. Besides $\omega$, we introduce another
imaginary auxiliary vector,
\begin{equation}
  \label{eq:17}
  \tilde \omega=i\bigg(-\frac{s+t}{t} k_1+\frac{t}{s} k_2-k_4\bigg)\,.
\end{equation}
Then,
\begin{gather}
  \tilde \omega \cdot  k_1=0,\quad \tilde \omega\cdot  k_2=0,\quad
  \omega\cdot \tilde \omega=0, \quad (\tilde \omega)^2=\omega^2=-\frac{t(s+t)}{s}.
\end{gather}
Note that the momentum $k_4$ does not appear in propagators of
this triangle diagram, so we would better replace the variable
$x_3=l\cdot k_4$ by a new variable $y_3\equiv l\cdot \tilde \omega$,
\begin{equation}
  \label{eq:20}
  x_3=-\frac{s+t}{s} x_1+\frac{t}{s} x_2+i y_3 \,.
\end{equation}

The integrand reduction for triangle reads
$N_\text{tri}=\Delta_\text{tri}+h_1D_1+h_2D_2+h_3D_3$. Generically,
\begin{equation}
  \label{eq:21}
  N_\text{tri} =\sum_{m_1} \sum_{m_2} \sum_{m_3} \sum_{m_4} d_{m_1m_2m_3m_4}x_1^{m_1} x_2^{m_2} y_3^{m_3} x_4^{m_4}\,,
\end{equation}
with the renormalization constraint that $m_1+m_2+m_3+m_4 \leq
3$ \cite{Ossola:2006us}. Here we already replaced $x_3$. Again,
\begin{align}
  x_1 &= l \cdot p_1 =\half (D_1-D_2)\,, \nn \\
  x_2 &= l \cdot p_2 =\half (D_2-D_3)+\frac{s}{2}\,. 
\label{RSP_tri}
\end{align}
we have $2$ RSPs, $x_1$, $x_2$ and $2$ ISPs, $y_3$, $x_4$. Again, from
$D_1=l^2$, we have
\begin{equation}
  \label{eq:22}
  y_3^2+x_4^2=\mathcal O (D_i) \,,
\end{equation}
which means we can trade $y_3^2$ for $x_4^2$. Hence with the
renormalization condition,
\begin{equation}
  \Delta_\text{tri}=d_0' +d_1' y_3+ d_2' x_4 +d_3' y_3 x_4 + d_4'
  x_4^2+d_5' y_3 x_4^2+d_6' x_4^3 \,.
\end{equation}
which contains $7$ terms. By Lorentz symmetry,
\begin{equation}
  \int d^Dl\frac{y_3^m x_4^n}{D_1D_2D_3}=0 \,,
\end{equation}
as long as $m$ is odd or $n$ is odd. It seems that $x_4^2$ term 
survives the integration. To further simplify the integral,
we redefine the integrand basis,
\begin{equation}
  \Delta_\text{tri}=d_0+d_1 y_3+ d_2 x_4 +d_3 y_3 x_4 + d_4
  (x_4^2-y_3^2)+d_5 y_3 x_4^2+d_6 x_4^3 \,.
\label{integrand_basis_tri}
\end{equation}
By the symmetry between $\tilde \omega$ and $\omega$, the term
proportional to $d_4$ integrates to zero. Hence, the integrand basis
of triangle contains $1$ scalar integral and $6$ spurious
terms.\footnote{We use the massless case as an illustrative
  example. Actually for a triangle diagram with two massless external
  lines, the scalar integral itself can be further reduced to bubble
  integrals, via IBPs.}

To use this basis, again, there are two manners as in the previous
section. 
\begin{enumerate}
\item (IR-D). Suppose that the box integrand reduction is finished and
  the triangle diagram integrand is obtained, say from Feynman
  rules. We combine the triangle integrand and the term proportional to
  $D_4$ in \eqref{box_IR}, and carry out the reduction process in this
  section explicitly. Finally, we get coefficients $d_0,\ldots,d_6$.
\item (IR-U). The goal is to determine $d_0,\ldots,d_6$ from
  unitarity. We need the triple cut,
  \begin{equation}
    \label{eq:10}
    D_0=D_1=D_3=0 \,,
  \end{equation}
\end{enumerate}
There are two branches of solutions,
\begin{align}
  \label{tri_cut_solution}
  \text{(1)}\quad x_1&=0,\quad x_2=\frac{s}{2},\quad y_3=i z,\quad x_4=z \,,\\
 \text{(2)}\quad x_1&=0,\quad x_2=\frac{s}{2},\quad y_3=-i z,\quad x_4=z \,,
\end{align}
where for each branch $z$ is a free parameter. On this cut, the
numerator becomes a sum of products of tree amplitudes,
\begin{gather}
  \label{eq:14}
  S^{(i)}_\text{tri}(z)=\sum_{h_1} \sum_{h_2} \sum_{h_3} A(k_1,
 l^{(i)}-k_1, \c -l^{(i)};s_1, h_2, \c -h_1)(z)\times \nn\\
A(k_2, l^{(i)}\c -k_1\c -k_2,k_1-l^{(i)};s_2, h_3, -h_2)(z) A(k_3,k_4,
l^{(i)},k_1+k_2\c -l^{(i)};s_3,s_4, h_1, \c -h_3)(z)\nn \,.\\
\end{gather}
for $i=1,2$. We try to fit coefficients in $\Delta_\text{tri}$ with
$S^{(i)}_\text{tri}(z)$. However, the new issue is that $\Delta_\text{tri}$
on either branch, is a polynomial of $z$. $S^{(i)}_\text{tri}(z)$ in
general is not a polynomial of $z$, since the last tree amplitude may
have a pole when $(l+p_4)^2=0$. On the cut,
\begin{equation}
  \label{eq:1}
  \frac{1}{(l+p_4)^2}=\frac{1}{t+2 i y_3}\,,
\end{equation}
which becomes a fraction in $z$ for each branch. Note that this pole
is from quadruple cut, hence we have to subtract the box integrand basis to
avoid the double counting. The correct unitarity relation is,
\begin{equation}
  \Delta_\text{tri}\big(l^{(i)}(z)\big)=S^{(i)}_\text{tri}(z)
  -\frac{c_0+c_1 \big(l^{(i)}(z)\cdot
    \omega\big)}{\big(l^{(i)}(z)+p_4\big)^2}\,, \quad i=1,2  \,.
\label{OPP_subtraction}
\end{equation}
If $c_0$ and $c_1$ are known from box integrand reduction, then
both sides of the equation are polynomials in $z$ and Tylor
expansions determine coefficients $d_0,\ldots,d_6$. \footnote{Note that in
  general, for a massive triangle diagram, the two cut branches may
  merge into one. In this case, a Laurent expansion over $z$ is needed
  and \eqref{OPP_subtraction} again remove the redundant pole.}

The further reduction for bubbles is similar.

\subsection{D-dimensional one-loop integrand reduction}
Dimensional regularization is a standard way for QFT
renormalization. Here we briefly introduce OPP integrand reduction
\cite{Ossola:2007ax,Giele:2008ve,Ellis:2011cr} in D-dimension for
one-loop diagrams.

Again, consider the four-point massless box integral in $D=4-2\epsilon$, 
\begin{equation}
  \label{box_D}
  I_\text{box}^D[N]=\int \frac{d^D l}{i \pi^{D/2}} \frac{N^D_\text{box}}{D_1
    D_2 D_3 D_4}\,,
\end{equation}
with the same definition of $D_i$'s. The loop momentum $l$ contains
two parts $l=l^{[4]}+l^\perp$, where $l^{[4]}$ is the four-dimensional
part and $l^\perp$ is the component in the extra dimension. 
\begin{equation}
  \label{eq:24}
  l^2=(l^{[4]})^2+(l^\perp)^2=(l^{[4]})^2 -\mu_{11}\,.
\end{equation}
Here we introduce a variable $\mu_{11}=-(l^\perp)^2$. We use the
scheme such that all external particles are in $4D$, hence,
\begin{equation}
  \label{eq:25}
  (l^\perp) \cdot k_i=0,\quad i=1,\ldots, 4
\end{equation}
and similar orthogonal conditions hold between $l^\perp$ and external
polarization vectors hold. This implies $l^\perp$ appears in the
integrand only in the form of $\mu_{11}$. $l^{[4]}$ is
parameterized by the same van Neerven-Vermaseren variables $x_1,\ldots
x_4$, as before. Therefore,
\begin{equation}
  \label{eq:26}
  N^D_\text{box}=\sum_{m_1} \sum_{m_2} \sum_{m_3} \sum_{m_4} \sum_m
  c_{m_1 m_2 m_3 m_4 m}x_1^{m_1} x_2^{m_2} x_3^{m_3} x_4^{m_4} \mu_{11}^m\,,
\end{equation}
with the renormalization condition $m_1+m_2+m_3+m_4+2m\leq
4$. ($\mu_{11}$ contains $2$ powers of $l$.) Again, as in the 4D case,
\begin{eqnarray}
\label{box_D_RSP}
x_1 = \half (D_1-D_2),\quad x_2 =\half (D_2-D_3)+\frac{s}{2},\quad 
 x_3 = \half (D_4-D_1),
\end{eqnarray}
so $x_1$, $x_2$ and $x_3$ are RSPs which do not appear in the
integrand basis. The ISPs are $x_4$ and $\mu_{11}$. From the relation
$D_1=(l^{[4]})^2 -\mu_{11}$, we get,
\begin{equation}
  \label{eq:27}
  x_4^2=\frac{t^2}{4}-\frac{(s+t)t}{s} \mu_{11} +\mathcal O(D_i)\,,
\end{equation}
Hence we can trade $x_4^2$ for $\mu_{11}$ in the integrand basis,
\begin{equation}
  \Delta_\text{box}^D=c_0+c_1 x_4+ c_2 \mu_{11}+ c_3 \mu_{11} x_4+c_4
  \mu_{11}^2\,,
\label{box_integrand_basis_D}
\end{equation}
which contains $5$ terms. The terms proportional to $x_4$ are again
spurious, i.e., integrated to zero. 

The coefficients $c_0,\ldots c_4$ can either be calculated from
explicit reduction (IR-D) or unitarity (IR-U). For the latter, the
quadruple cut $D_1=D_2=D_3=D_4=0$ is applied. There is one family of solutions
which is
one-dimensional,
\begin{equation}
  \label{eq:29}
  x_1=0,\quad x_2=\frac{s}{2},\quad x_3=0,\quad x_4=z, \quad \mu_{11}=\frac{s(t^2-4z^2)}{4t(s+t)}.
\end{equation}
Amazingly, the $4D$ quadruple cut contains two zero-dimensional solutions
while $D$-dim quadruple cut has only one family of solution. The two
roots in $4D$ are connected by a cut-solution curve, in DimReg. The Taylor series in $z$
fits coefficients $c_0,\ldots c_4$.

If only $\epsilon\to 0$ limit of the amplitudes is needed,
\eqref{box_integrand_basis_D} can be further simplified by
{\it dimension shift} identities,
\begin{align}
  \label{box_dimension_shift}
  \int\frac{d^D}{i\pi^{D/2}}
  \frac{\mu_{11}}{D_1D_2D_3D_4}&=\frac{D-4}{2} I_\text{box}^{D+2}[1]\\
\int\frac{d^D}{i\pi^{D/2}}
  \frac{\mu_{11}^2}{D_1D_2D_3D_4}&=\frac{(D-4)(D-2)}{4} I_\text{box}^{D+4}[1]
\end{align}
These identities can be proven via Baikov parameterization (Chapter
\ref{cha:integr-parts-reduct}) . It is well
known that the $6D$ scalar box integral is finite and the $8D$ scalar box is
UV divergent such that,
\begin{align}
  \label{eq:31}
\lim_{D\to 4}\frac{D-4}{2} I_\text{box}^{D+2}[1] &=0,\\
  \lim_{D\to 4}\frac{(D-4)(D-2)}{4} I_\text{box}^{D+4}[1] &=-\frac{1}{3}.
\end{align}
Hence the integrand basis after integration becomes,
\begin{eqnarray}
  \label{eq:32}
  \lim_{D\to 4}\int\frac{d^Dl}{i\pi^{D/2}}
  \frac{\Delta^D_\text{box}}{D_1D_2D_3D_4} &=& c_0
                                               I_\text{box}^{D}[1]-\frac{1}{3} c_4
\end{eqnarray}
in the $\epsilon\to 0$ limit. 
The second term is called a {\it rational term}, which cannot be
obtained from the $4D$ quadruple cut. 

It seems that $D$-dimensional integrand reduction is more complicated
than the $4D$ case, with more variables and more integrals in the basis. However, it provides the complete amplitude for a
general renormalizable QFT, and mathematically, its cut solution has
simpler structure.

OPP method is programmable and highly efficient for automatic one-loop
amplitude computation \cite{Ossola:2007ax,Badger:2010nx,Cullen:2011xs,Hirschi:2011pa}. 

\section{Issues at higher loop orders}
Since OPP method is very convenient for one-loop cases, the natural question
is: is it possible to generalize OPP method for higher loop orders?

Of course, higher loop diagrams contain more loop momenta and usually
more propagators. Is it a straightforward generalization? The answer
is ``no''. For example, consider the $4D$
$4$-point massless double box diagram (see Fig. \ref{graph_dbox}),
\begin{figure}
\centering
\includegraphics[scale=0.9]{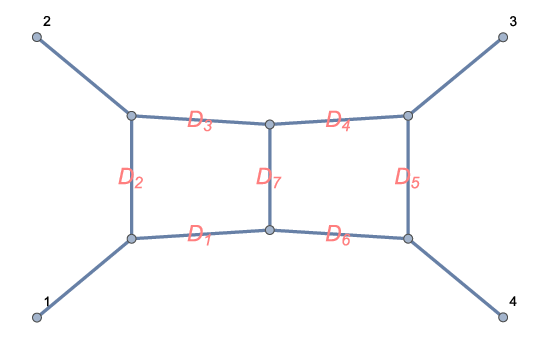}
\caption{two-loop double box diagram}
\label{graph_dbox}
\end{figure}
associated with the integral,
\begin{eqnarray}
  \label{eq:33}
  I_\text{dbox}[N]=\int \frac{d^4 l_1}{i \pi^2}\frac{d^4 l_2}{i \pi^2}\frac{N}{D_1 D_2 D_3 D_4 D_5 D_6 D_7}.
\end{eqnarray}
The denominators of propagators are,
\begin{gather}
  \label{dbox_propagators}
  D_1=l_1^2,\quad D_2=(l_1-k_1)^2,\quad D_3=(l_1-k_1-k_2)^2,\quad D_4=(l_2+k_1+k_2)^2, \nn\\
D_5=(l_2-k_4)^2,\quad
D_6=l_2^2,\quad D_7=(l_1+l_2)^2\,.
\end{gather}
The goal of reduction is to express,
\begin{gather}
  \label{dbox_IR}
  N_\text{dbox}=\Delta_\text{dbox} +h_1 D_1 +\ldots +h_7 D_7
\end{gather}
such that $\Delta_\text{dbox}$ is the ``simplest''. (In the sense that
all its coefficients in $\Delta_\text{dbox}$ can be uniquely fixed
from unitarity, as in the box case.) 

We use van Neerven-Vermaseren basis as before,
$\{e_1,e_2,e_3,e_4\}=\{k_1,k_2,k_4,\omega\}$. Define
\begin{equation}
  \label{dbox_vNV}
  x_i=l_1\cdot e_i, \quad y_i=l_2\cdot e_i,\quad i=1,\ldots 4.
\end{equation}
Then we try to determine $\Delta_\text{dbox}$ in these variables like
one-loop OPP method. 
\begin{align}
  x_1 &= \half (D_1-D_2)\,, \nn \\
  x_2 &=\half (D_2-D_3)+\frac{s}{2}\,, \nn \\
y_2 &= \half (D_4-D_6)-y_1-\frac{s}{2}\,, \nn \\
  y_3 &= \half (D_6-D_5) \,,
 \label{dbox_RSP}
\end{align}
Hence we can remove RSPs: $x_1$, $x_2$, $y_2$ and $y_3$ in
$\Delta_\text{dbox}$. (We trade $y_2$ for $y_1$, by 
symmetry consideration: under the left-right flip symmetry of double
box, $x_3 \leftrightarrow y_1$. ) There are $4$ ISPs, $x_3$, $y_1$,
$x_4$ and $y_4$.

Then following the one-loop OPP approach, the quadratic terms in $(l_i
\cdot \omega)$ can be removed from the integrand basis, since,
\begin{align}
  x_4^2&=x_3^2-t x_3+\frac{t^2}{4} +\mathcal O(D_i)\,,\nn\\
 y_4^2&=y_1^2-t y_1+\frac{t^2}{4}+\mathcal O(D_i) \,,\nn\\
 x_4 y_4 &=\frac{s+2t}{s} x_3
             y_1+\frac{t}{2}x_3+\frac{t}{2}y_1-\frac{t^2}{4}+\mathcal
             O(D_i) \,.
 \label{dbox_quadratic}
\end{align}
Then the trial version of integrand basis has the form,
\begin{eqnarray}
  \label{eq:28}
  \Delta_\text{dbox}=\sum_m \sum_n \sum_\alpha \sum_\beta
  c_{m,n,\alpha,\beta} x_3^{m} y_1^{n} x_4^{\alpha} y_4^{\beta} \,,
\end{eqnarray}
where $(\alpha,\beta)\in\{(0,0),(1,0),(0,1)\}$. The renormalization
condition is,
\begin{eqnarray}
  \label{eq:34}
  m+\alpha\leq 4,\quad n+\beta\leq 4,\quad m+n+\alpha+\beta\leq 6\, . 
\end{eqnarray}
By counting, there are $56$ terms in the basis. Is this basis correct?

Have a look at the unitarity solution. The heptacut $D_1=\ldots D_7=0$
has a complicated solution structure \cite{Kosower:2011ty}. (See table. \ref{dbox_sol}).
\begin{table}
    \centering
\begin{tabular}{|c|c|c|c|c|c|c|c|c|}
\hline
& $x_1$ & $x_2$ & $x_3$ & $x_4$ & $y_1$ & $y_2$ & $y_3$ & $y_4$ \\
\hline
(1)& $0$ & $\frac{s}{2}$ & $z_1$ & $z_1-\frac{t}{2} $ & $0$ & $-\frac{s}{2}$
                                                & 0 & $\frac{t}{2}$\\
\hline
(2)& $0$ & $\frac{s}{2}$ & $z_2$ & $-z_2+\frac{t}{2} $ & $0$ & $-\frac{s}{2}$
                                                & 0 & $-\frac{t}{2}$\\
\hline
(3)& $0$ & $\frac{s}{2}$ & $0$ & $\frac{t}{2}$ & $z_3$ & $-z_3-\frac{s}{2}$ &
                                                                       $0$
                                                        &
                                                          $z_3-\frac{t}{2}$
  \\
\hline
(4)& $0$ & $\frac{s}{2}$ & $0$ & -$\frac{t}{2}$ & $z_4$ & $-z_4-\frac{s}{2}$ &
                                                                       $0$
                                                        &
                                                          $-z_4+\frac{t}{2}$
  \\
\hline
(5) & $0$ & $\frac{s}{2}$ & $\frac{z_5-s}{2}$ & $\frac{z_5-s-t}{2}$ &
                                                              $\frac{s(s+t-z_5)}{2z_5}$
                                        & $-\frac{s(s+t)}{2z_5}$ & $0$ &
                                                                       $\frac{(s+t)(s-z_5)}{2z_5}$
  \\
\hline
(6)& $0$ & $\frac{s}{2}$ & $\frac{z_6-s}{2}$ & $\frac{-z_6+s+t}{2}$ &
                                                              $\frac{s(s+t-z_6)}{2z_6}$
                                        & $-\frac{s(s+t)}{2z_6}$ & $0$ &
                                                                       $-\frac{(s+t)(s-z_6)}{2z_6}$\\
\hline
\end{tabular}
    \caption{solutions of the $4D$ double box heptacut.}
\label{dbox_sol}
  \end{table}
There
are $6$ branches of solutions, each of which is parameterized by a
free parameter $z_i$. Solutions (5) and (6) contain poles in $z_i$,
hence we need Laurent series for tree products,
\begin{equation}
  \label{eq:35}
  S^{(i)}=\sum_{k=-4}^4 d_k^{(i)} z_i^k,\quad i=5,6\,.
\end{equation}
The bounds are from renormalization conditions, so there are $9$ nonzero coefficients for
each case. Solutions (1), (2), (3), (4) are relatively simpler, 
\begin{equation}
  \label{eq:35}
  S^{(i)}=\sum_{k=0}^4 d_k^{(i)} z_i^k,\quad i=1,2,3,4\,. 
\end{equation}
So there are $5$ nonzero coefficients for
each case. These solutions are not completely indenpendent, for example, solution
(1) at $z_1=s$
and solution (6) at $z_6=t/2$ correspond to the same loop momenta. Therefore, 
\begin{equation}
  \label{eq:36}
  S^{(1)}(z_1\to s)=S^{(6)}(z_6 \to t/2)\,.
\end{equation}
There are $6$ such intersections, namely between solutions (1) and (6), (1)
and (4), (2) and (3), (2) and (5), (3) and (6), (4) and (5). Hence,
there are $9\times 2+5\times 4 -6=32$ independent $d_k^{(i)}$'s. 

Now the big problem emerges,
\begin{equation}
  \label{eq:38}
  56>32\,.
\end{equation}
There are more terms in the integrand basis than those 
determined from unitarity cut. That means this integrand basis is 
redundant. However, it seems that we already used all algebraic
constraints in \eqref{dbox_RSP} and \eqref{dbox_quadratic}. Which
constraint is missing?

We need to reconsider \eqref{dbox_IR}, especially the meaning of
``simplest'' integrand basis. For simple example like massless double
box diagram, it is possible to use the detailed structures like
symmetries and Gram determinant constraints, to get a proper integrand basis \cite{Mastrolia:2011pr, Badger:2012dp}.
However, in general, we need an automatic reduction
method, without looking at the details. So we refer to
a new mathematical approach, {\it computational algebraic geometry}.

\section{Elementary computational algebraic geometry methods}
\subsection{Basic facts of algebraic geometry in affine space I}
In order to apply the new method, we need to list some basic concepts
and facts on algebraic geometry \cite{MR0463157}.

We start from a polynomial ring $R=\F[z_1,\ldots z_n]$ which is the
collection of all polynomials in $n$ variables
$z_1,\ldots z_n$ with coefficients in the {\it field} $\F$. For
example, $\F$ can
be $\Q$, the rational numbers, $\C$, the complex numbers, $\Z/p\Z$, the {\it finite
field} of integers modulo a prime number $p$, or $\C(c_1,c_2,\ldots
c_k)$, the complex rational functions of parameters $c_1,\ldots ,c_k$.

Recall that the right hand side of \eqref{dbox_IR} contains the sum
$h_1 D_1 +\ldots + h_7 D_7$ where $D_i$'s are known polynomials and
$h_i$'s are arbitrary polynomials. What are general properties of such
a sum? That leads to the concept of {\it ideal}.

\begin{definition}
  An ideal $I$ in the polynomial ring $R=\F[z_1,\ldots z_n]$ is a
  subset of $R$ such that,
  \begin{itemize}
\item $0\in I$. For any two $f_1,f_2\in I$, $f_1+f_2 \in I$. For any
  $f\in I$, $-f\in I$.
\item For  $\forall f \in I$ and $\forall h \in R$, $h f\in I$.
  \end{itemize}
\end{definition}

The ideal in the polynomial ring $R=\F[z_1,\ldots z_n]$ generated
  by a subset $S$ of $R$ is the collection of all such polynomials,
  \begin{equation}
    \sum_i h_i f_i, \quad h_i\in R, \quad f_i\in S.
\label{Ideal_generator}  
\end{equation}
 This ideal is denoted as $\langle S \rangle$. In particular, $\langle
 1 \rangle=R$, which is an ideal which contains all polynomials. Note that even if $S$ is an
 infinite set, the sum in \eqref{Ideal_generator}  is always restricted to 
 a sum of a finite number of terms. $S$ is called the generating set of
 this ideal.
 \begin{example}
\label{example_ideal}
  Let $I=\langle x^2+y^2+z^2-1, z\rangle$ in $\Q[x,y,z]$. By definition,
  \begin{equation}
    I=\{h_1 (x^2+y^2+z^2-1)+h_2\cdot z,\ \forall h_1,h_2\in R\}\,,
  \end{equation}
Pick up $h_1=1$, $h_2=-z$, and we see $x^2+y^2-1\in I$. Furthermore, 
\begin{equation}
  x^2+y^2+z^2-1 = (x^2+y^2-1)+z\cdot z\,.
\end{equation}
Hence $I=\langle  x^2+y^2-1, z\rangle$. We see that, in general, the
generating set of an ideal is not unique.
 \end{example}
Our integrand reduction problem can be rephrased as: given $N$ and the
ideal $I=\langle D_1 ,\ldots, D_7 \rangle$, how many terms in $N$ are
in $I$? To answer this, we need to study properties of ideals.
\begin{thm}[Noether]
\label{thm_Noether}
   The generating set of an ideal $I$ of $R=\F[z_1,\ldots z_n]$ can
   always be
   chosen to be finite.
 \end{thm}
 \begin{proof}
   See Zariski, Samuel \cite{MR0384768}.
 \end{proof}
This theorem implies that we only need to consider ideals generated by
finite sets in the polynomial ring $R$. 
\begin{definition}
\label{quotient_ring}
  Let $I$ be an ideal of $R$, we define an equivalence relation,
  \begin{equation}
    \label{eq:59}
    f\sim g,\quad \text{if and only if } f-g\in I\,.
  \end{equation}
We define an equivalence class, $[f]$ as the set of all $g\in R$ such that
$g\sim f$. The {\it quotient ring} $R/I$ is set of
  equivalence classes,
  \begin{equation}
    R/I=\{[f]| f\in R\}\,.
  \end{equation}
with multiplication $[f_1][f_2]\equiv [f_1f_2]$. (Check this
multiplication is well-defined.)
\end{definition}
To study the structure of
an ideal, it is very useful to consider the algebra-geometry relation.
\begin{definition}
 Let $\mathbb K$ be a field, $\F\subset \K$. The $n$-dimensional $\K$-affine space $\mathbf A^n_\K$ is the set of all $n$-tuple
 of $\K$. Given a subset $S$ of the polynomial ring $\F[z_1,\ldots,z_n]$, its {\it algebraic
 set} over $\K$ is,
 \begin{equation}
   \label{eq:8}
   \mathcal Z_\K(S)=\{p\in \mathbf A^n_\K | f(p)=0,\ \text{for every } f \in S\}.
 \end{equation}
If $\K=\F$, we drop the subscript $\K$ in $\mathbf
A^n_\K$ and $\mathcal Z_\K(S)$.
\end{definition}
So the algebraic set $\mathcal Z(S)$ consists of all {\it common solutions}
of polynomials in $S$. Note that to solve polynomials in $S$ is
equivalent 
to solve all polynomials simultaneously in the ideal generated by $S$,
\begin{equation}
  \label{eq:18}
  \mathcal Z(S)=\mathcal Z(\langle S \rangle ),
\end{equation}
since if $p\in \mathcal Z(S)$, then $f(p)=0$, $\forall f\in S$. Hence, 
\begin{equation}
  \label{eq:37}
  h_1(p) f_1(p) + \ldots + h_k(p) f_k(p)=0,\quad \forall h_i \in R,\  \forall f_i \in S.
\end{equation}
So we always consider the algebraic set of an ideal.

For example, $\mathcal Z(\langle 1 \rangle )=\emptyset$ (empty set) since $1\not =0$. For the ideal $I=\langle x^2+y^2+z^2-1, z\rangle$ in
example \ref{example_ideal}, $\mathcal Z(I)$ is the unit circle on the plane
$z=0$. 

We want to learn the structure of an ideal from its algebraic
set. First, for the empty algebraic set,
\begin{thm}[Hilbert's weak Nullstellensatz]
\label{weak_Nullstellensatz}
 Let $I$ be an ideal of $\Fpoly$ and $\K$ be an algebraically closed
 field \footnote{A field $\K$ is algebraically closed, if any
   non-constant polynomial in $\K[x]$ has a solution in $\K$. $\Q$ is
   not algebraically closed, the set of all algebraic numbers $\bar Q$ and
 $\C$ are algebraically closed.} ,
 $\F\subset \K$. If $\mathcal Z_\K(I)=\emptyset$, then $I=\langle 1 \rangle$.
\end{thm}
\begin{proof}
  See Zariski and Samuel, \cite[Chapter 7]{MR0389876}.
\end{proof}

\begin{remark}
The field extension $\K$ must be algebraically closed.  Otherwise,
say, $\K=\F=\Q$, the ideal $\la x^2-2 \ra$ has empty algebraic set in
$\Q$. (The solutions are not rational). However, $\la x^2-2\ra\not=\la
1\ra$. On the other hand, $\F$ need not be algebraically
closed. $I=\langle 1 \rangle$ means,
\begin{equation}
  1=h_1 f_1 +\ldots + h_k f_k,\quad f_i \in I, \ h_i\in \Fpoly\,.
\end{equation}
where $h_i$'s coefficients are in $\F$, instead of an algebraic extension
of $\F$. 
\end{remark}
\begin{example}
  We prove that, generally, the $4D$ pentagon diagrams are reduced to
  diagrams with fewer than $5$ propagators, $D$-dimensional hexagon
  diagram are  reduced to
  diagrams with fewer than $6$ propagators, in the integrand level.

For the $4D$ pentagon case, there are $5$ denominators from
propagators, namely $D_1,\ldots D_5$. There are $4$ Van Neerven-Vermaseren variables for the loop
momenta, namely $x_1$, $x_2$, $x_3$ and $x_4$. So $D_i$'s are
polynomials in $x_1,\ldots,x_4$ with coefficients in
$\F=\Q(s_{12},s_{23},s_{34},s_{45},s_{15})$. Define $I=\la D_1,\ldots
D_5,\ra$. Generally $5$ equations in $4$ variables,
\begin{equation}
  \label{eq:40}
  D_1=D_2=D_3=D_4=D_5=0\,,
\end{equation}
have no solution (even with algebraic extensions). Hence by Hilbert's weak
Nullstellensatz, $I=\la 1\ra$. Explicitly, there exist $5$ polynomials
$f_i$'s in $\F[x_1,x_2,x_3,x_4]$ such that
\begin{equation}
  \label{eq:41}
 f_1D_1+f_2 D_2 +f_3 D_3 +f_4 D_4 +f_5 D_5=1\,.
\end{equation}
Therefore,
\begin{gather}
  \int d^4l \frac{1}{D_1 D_2 D_3 D_4 D_5}=\int d^4l \frac{f_1}{D_2 D_3
    D_4 D_5} +\int d^4l\frac{f_2}{D_1 D_3 D_4 D_5}+\int d^4l\frac{f_3}{D_1 D_2 D_4 D_5}\nn\\
\int d^4l\frac{f_4} {D_1 D_2 D_3 D_5}+\int d^4l\frac{f_5} {D_1 D_2 D_3 D_4}\,,
\end{gather}
where each term in the r.h.s is a box integral (or simpler). Note
that $f_i$'s are in $\F[x_1,x_2,x_3,x_4]$, so the coefficients
of these polynomials are rational functions of Mandelstam
variables $s_{12},s_{23},s_{34},s_{45},s_{15}$. Weak
Nullstellensatz theorem does not provide an algorithm for finding such
$f_i$'s. The algorithm will be given by the Gr\"obner
basis method in next subsection, or by the resultant method \cite{opac-b1094391}. 

Notice that in the DimReg case, we have one more variable
$\mu_{11}=-(l^\perp)^2$. The same argument using  Weak
Nullstellensatz leads to the result.
\end{example}
For a general algebraic set, we have the important theorem:
\begin{thm}[Hilbert's Nullstellensatz]
  Let $\F$ be an algebraically closed field and $R=\F[z_1,\ldots
  z_n]$. Let $I$ be an ideal of $R$. If $f\in R$ and,
  \begin{equation}
    \label{eq:44}
    f(p)=0,\quad \forall p\in \mathcal Z(I),
  \end{equation}
then there exists a positive integer $k$ such that $f^k\in I$. 
\end{thm}
\begin{proof} 
  See Zariski and Samuel, \cite[Chapter 7]{MR0389876}.
\end{proof}
Hilbert's Nullstellensatz characterizes all polynomials vanishing on
$\mathcal Z(I)$, they are ``not far away'' from elements in $I$. For example, $I=\la
(x-1)^2 \ra$ and $\mathcal Z(I)=\{1\}$. The polynomial $f(x)=(x-1)$ does not
belong to $I$ but $f^2\in I$. 

\begin{definition}
  Let $I$ be an ideal in $R$, define the {\it radical ideal} of $I$ as,
  \begin{equation}
    \label{eq:45}
    \sqrt I=\{f\in R| \exists k\in\Z^+, f^k\in I\}\,.
  \end{equation}
For any subset $V$ of $\mathbf A^n$, define the ideal of $V$ as 
\begin{equation}
  \label{eq:46}
  \mathcal I(V)=\{f\in R| f(p)=0, \ \forall p\in V\}\,.
\end{equation}
Then Hilbert's Nullstellensatz reads, over an algebraically closed
field,
\begin{equation}
  \label{eq:47}
   \mathcal I (\mathcal Z(I)) = \sqrt I\,. 
\end{equation}
An ideal $I$ is called {\it radical}, if $\sqrt I = I$. 
\end{definition}
If two ideals $I_1$ and $I_2$ have the same algebraic set $\mathcal
Z(I_1)=\mathcal Z(I_2)$, then they have the same radical ideals $\sqrt
I_1=\sqrt I_2$. On the other hand, if two sets in $\mathbb A^n$ have
the same ideal, what could we say about them? To answer this question,
we need to define topology of $\mathbb A^n$:
\begin{definition}[Zariski topology] Define Zariski topology of
  $\mathbf A^n_\F$ by setting all algebraic set to be topologically closed.
  (Here $\F$ need not be algebraic closed.)
  \end{definition}
\begin{remark}
  The intersection of any number of Zariski closed sets is closed since,
  \begin{equation}
    \bigcap_i \mathcal Z(I_i) =\mathcal Z(\bigcup_i I_i ).
\label{algebraic_set_intersection}
  \end{equation}
 The union of two closed sets is closed since,
\begin{equation}
    \mathcal Z(I_1)\bigcup \mathcal Z(I_2) =\mathcal Z( I_1 I_2
    )=\mathcal Z( I_1 \cap I_2 ).
\label{algebraic_set_union}
  \end{equation}
$\mathbf A^n_\F$ and $\emptyset$ are both closed because $\mathbf
A^n_\F=\mathcal Z(\{0\})$, $\emptyset=\mathcal Z(\la 1\ra)$. That means
Zariski topology is  well-defined. We leave the proof of 
\eqref{algebraic_set_intersection} and \eqref{algebraic_set_union} as
an exercise. 

Note that Zariski topology is different from the usual topology
defined by Euclidean distance, for $\F=\Q,\mathbb R,\C$. For example,
over $\C$, the ``open'' unit disc defined by $D=\{z||z|<1\}$ is not Zariski
open in $\mathbf A^1_\C$. The reason is that $\C-D=\{z||z|\geq 1\}$ is
not Zariski closed, i.e. $\C-D$ cannot be the solution set of one or
several complex polynomials in $z$. 
\end{remark}
Zariski topology is the foundation of affine algebraic geometry. With
this topology, the dictionary between algebra and geometry can be
established.
\begin{proposition}  (Here $\F$ need not be algebraic closed.)
  \begin{enumerate}
  \item If $I_1 \subset I_2$ are ideals of $\Fpoly$, $\mathcal
    Z(I_1) \supset \mathcal Z(I_2)$ 
    \item If $V_1 \subset V_2$ are subsets of $\mathbf A^n_\F$, $\mathcal
    I(V_1) \supset \mathcal I(V_2)$ 
    \item For any subset $V$ in $\mathbf A^n_\F$, $\mathcal Z (\mathcal
  I(V))=\overline V$, the Zariksi closure of $V$. 
  \end{enumerate}
\end{proposition}
\begin{proof}
  The first two statements follow directly from the definitions. For the
  third one,  $V\subset \mathcal Z (\mathcal
  I(V))$. Since the latter is Zariski closed, $\overline V \subset  \mathcal Z (\mathcal
  I(V))$. On the other hand, for any Zariski closed set $X$ containing
  $V$, $X=\mathcal Z(I)$. $I\subset \mathcal I(V)$. From statement 1,
  $X=\mathcal Z(I) \supset \mathcal Z (\mathcal
  I(V))$. As a closed set, $\mathcal Z (\mathcal
  I(V))$ is contained in any closed set which contains $V$, hence $\mathcal Z (\mathcal
  I(V))=\overline V$.
\end{proof}
In the case $\F$ is algebraic closed, the above proposition and
Hilbert's Nullstellensatz established the one-to-one correspondence
between radical ideals in $\Fpoly$ and closed sets in $\mathbf
A^n_\F$. We will study geometric properties like reducibility,
dimension, singularity later in these lecture notes. Before this, we turn to the computational aspect of affine
algebraic geometry, to see how to explicitly compute objects like
$I_1\cap I_2$ and $\mathcal Z(I)$.

\subsection{Gr\"obner basis}
\subsubsection{One-variable case}
We see that ideal is the central concept for the algebraic side of
classical algebraic geometry. An ideal can be generated by different
generating sets, some may be redundant or complicated. In
linear algebra, given a linear subspace $V=\sp\{v_1\ldots v_k\}$ we
may use Gaussian elimination to find the linearly-independent basis of
$V$ or Gram-Schmidt process to find an orthonormal basis. For ideals, a ``good basis'' can also dramatically simplify algebraic geometry
problems. 
\begin{example}
 \label{GB_one_variable}
As a toy model, consider some univariate cases. 
\begin{itemize}
\item 
For example, $I=\la x^3-x-1 \ra$ in
$R=\Q[x]$. Clearly, $I$ consists of all polynomials in $x$ proportional to
$x^3-x-1$, and every nonzero element in $I$ has the degree higher or
equal than $3$. So we say $B(I)=\{x^3-x-1\}$ is a ``good basis''
for $I$. $B(I)$ is useful:
for any polynomial $F(x)$ in $\Q[x]$,  polynomial division
determines,
\begin{equation}
  F(x) =q(x) (x^3-x-1)+r(x) , \quad q(x),r(x)\in \Q[x],\ \deg r(x)<3
\end{equation}
Hence $F(x)$ is in $I$ if and only if the remainder $r$ is zero. It
also implies that $R/I=\sp_\Q\{[1],[x],[x^2]\}$.

\item Consider $J=\la x^3-x^2+3x-3,x^2-3x+2\ra$. Is the naive choice
$B(J)=\{f_1,f_2\}=\{x^3-x^2+3x-3,x^2-3x+2\}$ a good basis? For
instance, $f=f_1-x f_2=2x^2+x-3$ is in $I$ but it is proportional to
neither $f_1$ nor $f_2$. Polynomial division over this basis is not
useful, since $f$'s degree is lower than $f_1$, the only division
reads,
\begin{equation}
  \label{eq:48}
  f=2 f_2 + (7 x-7) \,.
\end{equation}
The remainder does not tell us the membership of $f$ in
$I$. Hence $B(J)$ does not characterize $I$ or $R/I$, and it is not
``good''. Note that $\Q[x]$ is a principal ideal domain (PID), any
ideal can be generated by one polynomial. Therefore, use Euclidean
algorithm (Algorithm \ref{Euclid}) to find the greatest common factor of $f_1$ and $f_2$,
\begin{equation}
  (x-1)=\frac{1}{7}f_1(x)-\frac{x+2}{7}f_2(x),\quad (x-1)|f_1(x),\ (x-1)|f_2(x)
\end{equation}
Hence $J=\la x-1 \ra$. We can check that $\tilde B(J)=\{x-1\}$ is a ``good'' basis in the sense that
Euclidean division over $\tilde B(J)$ solves membership questions of
$J$ and determined $R/J=\sp_\Q\{[1]\}$.
\end{itemize}
\end{example}

\begin{algorithm}
\caption{Euclidean division for greatest common divisor}  
\label{Euclid}
\algsetup{indent=4em}
\begin{algorithmic}[1]
\STATE \algorithmicrequire\  $f_1, f_2$, $\deg f_1\geq \deg f_2$
\WHILE{$f_2\not |f_1$} 
                 \STATE polynomial division $f_1=q f_2+r$
                 \STATE  $f_1:=f_2$
                 \STATE $ f_2:=r$
\ENDWHILE
\RETURN  $f_2$ (gcd)
\end{algorithmic}
\end{algorithm}

Recall that in \eqref{dbox_IR}, given inverse propagators $D_1, \ldots, D_7$, we
need to solve the membership problem of $I=\la D_1\ldots D_7\ra$ and
compute $R/I$. However, in general, a set like $\{ D_1
\ldots D_7\}$ is not a ``good basis'', in the sense that the polynomial
division over this basis does not solve the membership problem or give
a correct integrand basis (as we see previously). Since it is a
multivariate problem, the polynomial ring $R$ is not a PID and we
cannot use Euclidean algorithm to find a ``good basis''.

Look at Example \ref{GB_one_variable} again. 
For the univariate
case, there is a natural monomial order $\prec$ from the degree,
\begin{equation}
  \label{eq:49}
  1 \prec x \prec x^2 \prec x^3 \prec x^4 \prec \ldots\,,
\end{equation}
and all monomials are sorted. For any polynomial $F$, define the
{\it leading term}, $\LT(F)$ to be the highest monomial in $F$ by this
order (with the
coefficient). For multivariate cases, the degree criterion is not fine enough to
sort all monomials, so we need more general monomial orders.
\begin{definition}
  Let $M$ be the set of all monomials with coefficients $1$, in the ring
  $R=\Fpoly$. A monomial order $\prec$ of $R$ is an ordering on $M$
  such that,
  \begin{enumerate}
  \item $\prec$ is a total ordering, which means any two different
    monomials are sorted by $\prec$.
\item $\prec$ respects monomial products, i.e., if $u\prec v$ then for
  any $w\in M$, $uw\prec vw$.
\item $1\prec u$, if $u\in M$ and $u$ is not constant. 
  \end{enumerate}
\end{definition}
There are several important monomial orders. For the ring $\Fpoly$,
we use the convention $1\prec z_n\prec z_{n-1}\prec\ldots \prec z_1$ for all
monomial orders. Given
two monomials, 
$g_1=z_1^{\alpha_1}\ldots z_n^{\alpha_n}$ and
 $g_2=z_1^{\beta_1}\ldots z_n^{\beta_n}$, consider the following orders:
\begin{itemize}
\item Lexicographic order (\lex). First compare $\alpha_1$ and
  $\beta_1$. If $\alpha_1<\beta_1$, then $g_1\prec g_2$. If $\alpha_1=\alpha_2$, we compare $\alpha_2$ and
  $\beta_2$. Repeat this process until for certain $\alpha_i$ and $\beta_i$
  the tie is broken. 
\item Degree lexicographic order (\grlex). First compare the total
  degrees. If $\sum_{i=1}^n\alpha_i<\sum_{i=1}^n\beta_i$, then
  $g_1\prec g_2$.
If total degrees are equal, we compare $(\alpha_1, \beta_1)$,
  $(\alpha_2, \beta_2)$ ... until the tie is broken, like \lex.
\item Degree reversed lexicographic order (\grevlex). First compare the total
  degrees. If $\sum_{i=1}^n\alpha_i<\sum_{i=1}^n\beta_i$, then
  $g_1\prec g_2$. If total degrees are equal, we compare $\alpha_n$ and
  $\beta_n$. If $\alpha_n<\beta_n$, then $g_1 \succ g_2$
  (reversed!). If $\alpha_n=\beta_n$, then we further compare $(\alpha_{n-1}$,
  $\beta_{n-1})$,  $(\alpha_{n-2}$,
  $\beta_{n-2})$ ... until the tie is broken, and use the reversed result. 
\item Block order. This is the combination of \lex \ and other orders. We
  separate the variables into $k$ blocks, say,
  \begin{equation}
    \label{eq:51}
    \{z_1,z_2,\ldots z_n\}=\{z_1,\ldots z_{s_1}\} \cup
    \{z_{s_1+1},\ldots z_{s_2}\} \ldots \cup \{z_{s_{k-1}+1},\ldots z_n\}\,.
  \end{equation}
Furthermore, define the monomial order for variables in each block. To compare
$g_1$ and $g_2$, first we compare the first block by the given
monomial order. If it is a tie, we compare the second
block... until the tie is broken.
\end{itemize}
\begin{example}
  Consider $\Q[x,y,z]$, $z\prec y \prec x$. We sort all monomials up to
  degree $2$ in \lex, \grlex, \grevlex \ and the block order $[x]\succ
  [y,z]$ with \grevlex \ in each block. This can be done be the
  following \mm\ code:

\begin{doublespace}
\noindent\(\pmb{F=1+x+x^2+y+x y+y^2+z+x z+y z+z^2;}\\
\pmb{\text{MonomialList}[F,\{x,y,z\},\text{Lexicographic}]}\\
\pmb{\text{MonomialList}[F,\{x,y,z\},\text{DegreeLexicographic}]}\\
\pmb{\text{MonomialList}[F,\{x,y,z\},\text{DegreeReverseLexicographic}]}\\
\pmb{\text{MonomialList}[F,\{x,y,z\},\{\{1,0,0\},\{0,1,1\},\{0,0,-1\}\}]\text{     }}\)
\end{doublespace}
and the output is,

\begin{doublespace}
\noindent\(\left\{x^2,x y,x z,x,y^2,y z,y,z^2,z,1\right\}\)
\end{doublespace}
\begin{doublespace}
\noindent\(\left\{x^2,x y,x z,y^2,y z,z^2,x,y,z,1\right\}\)
\end{doublespace}
\begin{doublespace}
\noindent\(\left\{x^2,x y,y^2,x z,y z,z^2,x,y,z,1\right\}\)
\end{doublespace}
\begin{doublespace}
\noindent\(\left\{x^2,x y,x z,x,y^2,y z,z^2,y,z,1\right\}\)
\end{doublespace}
Note that for \lex, $x\succ y^2$, $y\succ z^2$ since we first compare the power of
$x$ and the $y$. The total degree is not respected in this
order. On the other hand, \grlex\ and \grevlex\ both consider the
total degree first. The difference between \grlex\ and \grevlex\ is
that, $ x z\succ_\text{\grlex} y^2$ while $ x z\prec_\text{\grevlex}
y^2$. So $\grevlex$ tends to set monomials with more variables, lower,
in the list of monomials with a fixed degree. This property is useful for
computational algebraic geometry. Finally, for this block order,
$x\succ y^2$ since $x$'s degrees are compared first. But $y\prec z^2$,
since $[y,z]$ block is in \grevlex. 
\end{example}
With a monomial order, we define the leading term as
the highest monomial (with coefficient) of a polynomial in this
order. 
Back to the second part of Example \ref{GB_one_variable},
\begin{equation}
  \label{eq:50}
  \LT(f_1)=x^3\quad \LT(f_2)=x^2, \quad  \LT(x-1)=x
\end{equation}
The key observation is that although $x-1\in J$, its leading term
is not divisible by the leading term of either $f_1$ or $f_2$. This
makes polynomial division unusable and $\{f_1,f_2\}$ is not a ``
  good basis''. This leads to the concept of \GB.

\subsubsection{\GB}
\begin{definition}
  For an ideal $I$ in $\Fpoly$ with a monomial order, a \GB\ $G(I)=\{g_1,\ldots g_m\}$ is a generating set for
  $I$ such that for each $f\in I$, there always exists  $g_i\in G(I)$ such
  that,
  \begin{equation}
    \label{eq:52}
    \LT(g_i) | \LT(f) \,.
  \end{equation}
\end{definition}
We can check that for the ideal $J$ in Example \ref{GB_one_variable},
$\{f_1,f_2\}$ is not a Gr\"obner basis with respect to the natural
order, while $\{x-1\}$ is. 

\subsubsection{Multivariate polynomial division}

To harness the power of \GB\, we need the multivariate division
algorithm, which is a generalization of univariate Euclidean
algorithm (Algorithm \ref{multivariate_polynomial_division}). The
basic procedure is that: given a polynomial $F$ and a list of
$k$ polynomials $f_i$'s, if $\LT(F)$ is
divisible by some $\LT(f_i)$, then remove $\LT(F)$ by subtracting a multiplier of
$f_i$. Otherwise move $\LT(F)$ to the remainder $r$. The output will
be 
\begin{equation}
  \label{eq:53}
  F=q_1 f_1 + \ldots q_k f_k +r\,,
\end{equation}
where $r$ consists of monomials cannot be divided by any $LT(f_i)$.  Let
$B=\{f_1,\ldots f_k\}$, we
 denote $\overline{F}^B$  as the remainder $r$.
\begin{algorithm}
\caption{Multivariate division algorithm}  
\label{multivariate_polynomial_division}
\algsetup{indent=4em}
\begin{algorithmic}[1]
\STATE \algorithmicrequire\ $F$, $f_1\ldots f_k$, $\succ$
\STATE $q_1:=\ldots :=q_k=0$, $r:=0$
\WHILE{$F\not =0$} 
\STATE $reductionstatus:=0$
\FOR{$i=1$ to $k$}
          \IF{$\LT(f_i)|\LT(F)$} 
                     \STATE $q_i:=q_i+\frac{\LT(F)}{\LT(f_i)}$
                     \STATE $F:=F-\frac{\LT(F)}{\LT(f_i)} f_i$
                     \STATE $reductionstatus:=1$
                     \STATE \BREAK
           \ENDIF
           \ENDFOR
           \IF{$reductionstatus=0$}
                      \STATE $r:=r+\LT(F)$
                     \STATE $F:=F-\LT(F)$
           \ENDIF

\ENDWHILE
\RETURN  $q_1\ldots q_k$, $r$
\end{algorithmic}
\end{algorithm}
Recall that the one-loop OPP integrand
reduction and the naive trial of two-loop integrand reduction are very
similar to this algorithm.

Note that for a general list of polynomials, the algorithm has two
drawbacks: (1) the remainder $r$ depends on the order of the list,
$\{f_1,\ldots f_n\}$ (2) if $F\in \la f_1\ldots f_n\ra$, the algorithm
may not give a zero remainder $r$. These made the previous
two-loop integrand reduction unsuccessful. \GB\ eliminates these
problems.  

\begin{proposition}
\label{GB_division}
  Let $G=\{g_1,\ldots g_m\}$ be a \GB\ in $\Fpoly$ with the monomial
order $\succ$. Let $r$ be the remainder of the division of $F$ by $G$,
from Algorithm \ref{multivariate_polynomial_division}.
\begin{enumerate}
\item $r$ does not depend on the order of $g_1,\ldots g_m$.
\item If $F\in I=\la g_1,\ldots g_m\ra$, then $r=0$.
\end{enumerate}
\end{proposition}
\begin{proof}
  If the division with different orders of $g_1,\ldots g_n$ provides
  two remainder $r_1$ and $r_2$. If $r_1\not =r_2$, then $r_1-r_2$ contains monomials
  which are not divisible by any $\LT(g_i)$. But $r_1-r_2\in I$, this
  is a contradiction to the definition of \GB. 

If $F\in I$, then $r\in I$. Again by the definition of \GB, if
$r\not=0$, $\LT(r)$ is
divisible by some $\LT(g_i)$. This is a contradiction to multivariate
division algorithm. 
\end{proof}

Then the question is: given an ideal $I=\la f_1\ldots f_k\ra$ in
$\Fpoly$ and a monomial order $\succ$, does the \GB\ exist and how do
we find it? This is answered by \Buch, which was presented in 1970s
and marked the beginning of computational algebraic geometry.

\subsubsection{\Buch}
Recall that for one-variable case, Euclidean algorithm (Algorithm
\ref{Euclid}) computes the
gcd of two polynomials hence the \GB\ is given. The key step is to cancel
leading terms of two polynomials. That inspires the concept of
S-polynomial in multivariate cases.
\begin{definition}
\label{S-polynomial}
  Given a monomial order $\succ$ in $R=\Fpoly$, the S-polynomial of
  two polynomials $f_i$ and $f_j$ in $R$ is,
  \begin{equation}
    S(f_i,f_j)=\frac{\LT(f_j)}{\gcd\big(\LT(f_i),\LT(f_j)\big)} f_i
    -\frac{\LT(f_i)}{\gcd\big(\LT(f_i),\LT(f_j)\big)} f_j.
  \end{equation}
\end{definition}
Note that the leading terms of the two terms on the r.h.s
cancel. 

\begin{thm}[Buchberger]
   Given a monomial order $\succ$ in $R=\Fpoly$, \GB\ with respect to
   $\succ$ exists and can
   be found by \Buch\ (Algorithm \ref{Buchberger}).
\end{thm}
\begin{proof}
  See Cox, Little, O'Shea \cite{MR3330490}. 
\end{proof}

\begin{algorithm}
\caption{Buchberger algorithm}  
\label{Buchberger}
\algsetup{indent=4em}
\begin{algorithmic}[1]
\STATE \algorithmicrequire\ $B=\{f_1\ldots f_n\}$ and a monomial order $\succ$
\STATE $queue:=\text{all subsets of B with exactly two elements}$

\WHILE{$queue!=\emptyset$}
                     \STATE $\{f,g\}:=\text{head of } queue$
                     \STATE $r:=\overline{S(f,g)}^B$
                     \IF{$r\not=0$}
                                           \STATE $B:=B\cup{r}$
                                           \STATE queue $<<$
                                           $\{\{B_1,r\},\ldots \{{\text{last\ of}\ }B,r\}\}$
                     \ENDIF
                     \STATE {\bf delete} head of $queue$
\ENDWHILE

\RETURN $B$ (\GB)
\end{algorithmic}
\end{algorithm}

The uniqueness of \GB\ is given via {\it reduced \GB}.
\begin{definition}
  For $R=\Fpoly$ with a monomial order $\succ$, a reduced \GB\ is a 
  \GB\ $G=\{g_1,\ldots g_k\}$ with respect to $\succ$, such that
  \begin{enumerate}
  \item Every $\LT(g_i)$ has the coefficient $1$, $i=1,\ldots,k$. 
   \item Every monomial in $g_i$ is not divisible by  $\LT(g_j)$, if $j\not =i$.
  \end{enumerate}
\end{definition}

\begin{proposition}
  For $R=\Fpoly$ with a monomial order $\succ$, $I$ is an ideal. The
  reduced \GB\ of $I$ with respect to $\succ$, 
    $G=\{g_1,\ldots g_m\}$, is
    unique up to the order of the list $\{g_1,\ldots g_m\}$. It is 
    independent of the choice of the generating set of $I$.  
\end{proposition}
\begin{proof}
  See Cox, Little, O'Shea \cite[Chapter 2]{MR3330490}. Note that given a \GB\ $B=\{h_1\ldots h_m\}$,
  the reduced \GB\ $G$ can be obtained as follows,

  \begin{enumerate}
  \item For any $h_i\in B$, if $\LT(h_j)|\LT(h_i)$, $j\not=i$, then
    remove $h_i$. Repeat this process, and finally we
    get the {\it minimal basis} $G'\subset B$.
 \item For every $f\in G'$, divide $f$ towards $G'-\{f\}$. Then replace
   $f$ by the remainder of the division. Finally, normalize the
   resulting set such that every polynomial has leading coefficient
   $1$, and we get the reduced \GB\ $G$.
  \end{enumerate}
\end{proof}

Note that \Buch\  reduces only one polynomial pair every time, more recent
algorithms attempt to (1) reduce many polynomial pairs at once (2)
identify the ``unless'' polynomial pairs {\it a priori}. Currently,
the most efficient algorithms are Faugere's F4
and F5 algorithms \cite{Faugere199961,
  Faugere:2002:NEA:780506.780516}. 

Usually we compute Gr\"obner basis
by programs, for example,
\begin{itemize}
\item \mm\ The embedded \pmb{GroebnerBasis} computes Gr\"obner basis by
  \Buch. The relation between \GB\ and the original generating set is
  not given. Usually, \GB\ computation in \mm\  is not very fast.
\item {\sc Maple } Maple computes \GB\ by either \Buch\ or highly
  efficient F4
  algorithm. 
\item {\sc Singular } is a powerful computer algebraic system \cite{DGPS}
  developed in University of Kaiserslautern. \Singular\ uses either \Buch\ or F4
  algorithm to computer \GB.
\item {\sc Macaulay2} is a sophisticated algebraic geometry program
  \cite{M2}, which orients to research mathematical problems in algebraic
  geometry. It contains \Buch\ and experimental codes of F4
  algorithm. 
\item Fgb package \cite{FGb}. This is a highly efficient package of F4 and
  F5 algorithms by Jean-Charles Faug\'{e}re. It has both {\sc Maple
  } and {\sc C++ } interfaces. Usually, it is faster than the F4
  implement in {\sc Maple}. Currently, coefficients of
  polynomials are restricted to $\Q$ or $\Z/p$, in this package.
\end{itemize}

\begin{example}
\label{example_Buchberger}
  Consider $f_1=x^3 - 2 x y$, $f_2=x^2 y - 2 y^2 + x$. Compute the \GB\ of
  $I=\la f_1,f_2\ra$ with \grevlex\ and $x\succ y$ \cite{MR3330490}. 

We use \Buch. 
\begin{enumerate}
\item In the beginning, the list is $B:=\{h_1,h_2\}$ and the pair set
  $P:=\{(h_1,h_2)\}$, where
$h_1=f_1$, $h_2=f_2$,
\begin{equation}
  S(h_1,h_2)=-x^2,\quad h_3:=\overline{S(h_1,h_2)}^B=-x^2\,,
\end{equation}
with the relation $h_3=y h_1-x h_2$.
\item Now $B:=\{h_1,h_2,h_3\}$ and
$P:=\{(h_1,h_3),(h_2,h_3)\}$. Consider the pair $(h_1,h_3)$,
\begin{equation}
  S(h_1,h_3)=2xy,\quad h_4:=\overline{S(h_1,h_3)}^B=2 x y\,,
\end{equation}
with the relation $h_4=- h_1-x h_3$.
\item $B:=\{h_1,h_2,h_3,h_4\}$ and $ P:=\{(h_2,h_3),(h_1,h_4),(h_2,h_4),(h_3,h_4)\}$.
For the pair $(h_2,h_3)$,
\begin{equation}
  S(h_2,h_3)=-x+2y^2,\quad h_5:=\overline{S(h_2,h_3)}^B=-x+2y^2\,,
\end{equation}
The new relation is $h_5=-h_2-y h_3$. 
\item $B:=\{h_1,h_2,h_3,h_4,h_5\}$ and
\begin{equation}
  P:=\{(h_1,h_4),(h_2,h_4),(h_3,h_4),(h_1,h_5),(h_2,h_5),(h_3,h_5),(h_4,h_5)\}.
\end{equation}
For the pair $(h_1,h_4)$, 
\begin{equation}
  S(h_1,h_4)=-4 x y^2,\quad \overline{S(h_1,h_4)}^B=0
\end{equation}
Hence this pair does not add information to \GB. Similarly, all the
rests pairs are useless. 
\end{enumerate}
Hence the Groebner basis is
\begin{equation}
  \label{eq:42}
  B=\{h_1,\ldots h_5\}=\{x^3-2 x y,x^2 y+x-2 y^2,-x^2,2 x y,2 y^2-x\}.
\end{equation}
Consider all the relations  in intermediate steps, we determine
the conversion between the old basis $\{f_1, f_2\}$ and $B$,
\begin{gather}
  \label{eq:43}
  h_1= f_1,\quad h_2=  f_2,\quad h_3= f_1 y-f_2 x\nn\\
h_4= -f_1 (1+x y)+f_2 x^2,\quad h_5= -f_1 y^2 +(x y-1) f_2 
\end{gather}
Then we determine the reduced \GB. Note that $\LT(h_3)|\LT(h_1)$,
$\LT(h_4)|\LT(h_2)$, so $h_1$ and $h_2$ are removed. The minimal
\GB\ is $G'=\{h_3,h_4,h_5\}$. Furthermore,
\begin{equation}
  \label{eq:54}
  \overline{h_3}^{\{h_4,h_5\}} =h_3,\quad \overline{h_4}^{\{h_3,h_5\}}
  =h_4, \quad \overline{h_5}^{\{h_3,h_4\}} =h_5\quad
\end{equation}
so $\{h_3,h_4,h_5\}$ cannot be reduced further. The reduced \GB\ is
\begin{equation}
  \label{eq:56}
  G=\{g_1,g_2,g_3\}=\{-h_3, \half h_4, \half h_5\}=\{x^2,x y,
  y^2-\half x\}.
\end{equation}
The conversion relation is,
\begin{equation}
  \label{GB_conversion}
  g_1= -y f_1 +x f_2 ,\quad g_2= -\frac{ (1+x y)}{2}f_1+\half x^2 f_2
 ,\quad g_3= -\half y^2 f_1+\half (x y-1) f_2.
\end{equation}
\mm\ finds $G$ directly via $\pmb{
\text{GroebnerBasis}[\{x^3-2 x y,x^2 y-2 y^2+x\},\{x,y\},}$\\* 
$\pmb{\text{MonomialOrder}\to \text{DegreeReverseLexicographic}]}
$. However, it does not provide the conversion
\eqref{GB_conversion}. This can be found by {\sc Maple} or \Macaulay. 

As a first application of \GB\ , we can see some fractions can be
easily simplified (like integrand reduction),
\begin{eqnarray}
  \label{example_reduction}
  \frac{x^2}{(x^3-2 x y)(x^2 y-2 y^2+x)}&=&\frac{-y f_1+x f_2}{f_1 f_2}
  =-\frac{y}{f_2}+\frac{x}{f_1}\nn \\
 \frac{x y}{(x^3-2 x y)(x^2 y-2 y^2+x)}&=&\frac{- (1+x y)f_1/2+ x^2 f_2/2}{f_1 f_2}
  =-\frac{1+x y}{2f_2}+ \frac{x^2}{2f_1}\nn \\
 \frac{y^2}{(x^3-2 x y)(x^2 y-2 y^2+x)}&=&\frac{h_5+x/2}{f_1f_2}
  =\frac{x}{2f_1 f_2}-\frac{ y^2}{2f_2}+ \frac{x y-1}{2f_1}
\end{eqnarray}
In first two lines, we reduce a fraction with two denominators to
fractions with only one denominator. In the last line, a fraction with
two denominators is reduced to a fraction with two denominators but
lower numerator degree ($y^2\to x$). Higher-degree numerators can be
reduced in the same way. Hence we conclude that all fractions 
$N(x,y)/(f_1 f_2)$ can be reduced to, 
\begin{equation}
  \label{eq:55}
  \frac{1}{f_1 f_2}, \quad \frac{x}{f_1 f_2},\quad \frac{y}{f_1 f_2}
\end{equation}
and fractions with fewer denominators. Note that even with this
simple example,  one-variable partial fraction method does not
help the reduction. 
\end{example}
We have some comments on \GB:
\begin{enumerate}
  \item For $\Fpoly$, the computation of polynomial division and
    \Buch\ only used addition, multiplication and division in $\F$. No
    algebraic extension is needed. Let $\F\subset \K$ be a field
    extension. If $B=\{f_1,\ldots, f_k\}\subset \Fpoly$, then the \GB\
    computation of $B$ in $\K[x_1,\ldots,x_n]$ produces a \GB\,
    which is still in $\Fpoly$, irrelevant of the algebraic
    extension.
\item The form of a \GB\ and computation time dramatically depend on
  the monomial order. Usually, \grevlex\ is the fastest
  choice while \lex\ is the slowest. However, in some cases, \GB\
  with \lex\ is preferred. In these cases, we may instead consider some
  ``midway'' monomial order the like block order, or convert a
  known \grevlex\ basis to \lex\ basis \cite{FAUGERE1993329}. 
\item If all input polynomials are linear, then the reduced \GB\ is
  the {\it echelon form} in linear algebra. 
\end{enumerate}

\subsection{Application of \GB}
\GB\ is such a powerful tool that once it is computed, most
computational problems
on ideals are solved.
\subsubsection{Ideal membership and fraction reduction}
A \GB\ immediately solves the ideal membership problem. Given an $F\in
R=\Fpoly$, and $I=\la f_1,\ldots f_k\ra$. Let $G$ be a \GB\ of $I$
with a monomial order $\succ$. $F\in I$ if and only if
$\overline{F}^G=0$, i.e., the division of $F$ towards $G$ generates
zero remainder (Proposition \ref{GB_division}).  

$G$ also determined the structure of the quotient ring $R/I$
(Definition \ref{quotient_ring}). $f \sim g$ if and only if $f-g\in
I$. The division of $f_1-f_2$ towards $G$ detects equivalent
relations. In particular,
\begin{proposition}
  Let $M$ be the set of all monic monomials in $R$ which are not
  divisible by any leading term in $G$. Then the set,
  \begin{equation}
   V= \{[p]|p\in M\} 
  \end{equation}
is an $\F$-linear basis of $R/I$. 
\label{remainders}
\end{proposition}
\begin{proof}
  For any $F\in R$, $\overline{F}^G$ consists of monomials which are not
  divisible by any leading term in $G$. Hence $[F]$ is a linear
  combination of finite elements in $V$. 

  Suppose that $\sum_j c_j [p_j]=0$ and each $p_j$'s are  monic monomials which are not
  divisible by leading terms of $G$ . Then $\sum_jc_jp_j\in I$, but by
  the Algorithm \ref{multivariate_polynomial_division}. $\overline{\sum_jc_j
    p_j}^G=\sum_jc_j
    p_j$. So $\sum_jc_jp_j=0$ in $R$ and $c_j$'s are all zero.
\end{proof}

As an application, consider fraction reduction for $N/(f_1 \ldots
f_k)$, where $N$ is polynomial in $R$,
\begin{equation}
  \label{eq:60}
  \frac{N}{f_1 \ldots f_k}=\frac{r}{f_1 \ldots f_k} + \sum_{j=1}^k
  \frac{s_i}{f_1 \ldots \hat{f_j}\ldots f_k}.
\end{equation}
The goal is to make $r$ simplest, i.e., $r$ should not contain any
term which belongs to $I=\la f_1,\ldots f_k \ra$. We compute the \GB
\ of $I$, $G=\{g_1,\ldots g_l\}$ and record the conversion relations
$g_i=\sum_{j=1}^k f_j a_{ji}$ from the computation. 

Polynomial division of $N$ towards
$G$ gives,
\begin{equation}
  \label{eq:57}
  N=r+\sum_{i=1}^l q_i g_i
\end{equation}
where $r$ is the remainder. The result,
\begin{equation}
  \label{fraction_reduction}
   \frac{N}{f_1 \ldots f_k}=\frac{r}{f_1 \ldots f_k} + \sum_{j=1}^k
  \frac{\big(\sum_{i=1}^l a_{ji}q_i\big)}{f_1 \ldots \hat{f_j}\ldots f_k},
\end{equation}
gives the complete reduction since by the properties of $G$, no term in
$r$ belongs to $I$. \eqref{fraction_reduction} solves integrand
reduction problem for multi-loop diagrams. In practice, there are
shortcuts to compute numerators like $\big(\sum_{i=1}^l a_{ji}q_i\big)$.

\subsubsection{Solve polynomial equations with \GB}
In general, it is very difficult to solve multivariate polynomial
equations since variables are entangled. \GB\ characterizes the
solution set and can also remove variable entanglements. 

\begin{thm}
  Let $f_1\ldots f_k$ be polynomials in $R=\F[x_1,\ldots x_n]$ and
  $I=\la f_1\ldots f_k\ra$. Let $\bar \F$ be the algebraic closure of $\F$. The solution set
  in $\bar \F$, 
  $\mathcal Z_{\bar\F}(I)$ is finite, if and only if 
  $R/I$ is a finite dimensional $\F$-linear space. In this case,  the
  number of solutions in $\bar \F$, counted with multiplicity, equals
  $\dim_\F (R/I)$. 
\end{thm}
\begin{proof}
  See Cox, Little, O'Shea \cite{opac-b1094391}. The rigorous
  definition of multiplicity is given in the next chapter, Definition \ref{local_ring}. 
\end{proof}
 Note that again, we
  distinguish $\F$ and its algebraic closure $\bar \F$, since we do not
  need computations in $\bar \F$ to count total number of solutions in
  $\bar \F$. $\dim_\F (R/I)$ can be obtained by counting all monomials not
  divisible by $\LT(G(I))$, leading terms of the \GB. Explicitly,
  $\dim_\F (R/I)$ is computed by \pmb{vdim} of
  \Singular. 

\begin{example}
  \label{eq:61}
 Consider $f_1=-x^2+x+y^2+2, f_2=x^3-x y^2-1$. Determine the number of
 solutions $f_1=f_2=0$ in $\C^2$. 

Compute the \GB\ for $\{f_1,f_2\}$ in \grevlex with $x\succ y$, we get,
\begin{equation}
  \label{eq:62}
  G=\{y^2+3x+1,x^2+2x-1\}.
\end{equation}
Then $\LT(G)$=$\{y^2,x^2\}$. Then $M$ in Proposition \ref{remainders}
is clearly $\{1, x ,y, x y\}$. The linear basis for $\Q[x,y]/\la f_1,
f_2\ra$ is $\{[1],[x],[y],[x y]\}$. Therefore there are $4$ solutions
in $\C^2$. Note that B\'ezout's theorem would give the number $2\times
3=6$. However, we are considering the solutions in affine space, so
there are $6-4=2$ solutions at infinity. Another observation is that the
second polynomial in $G$ contains only $x$, so the variable
entanglement disappears and we can first solve for $x$ and then us
$x$-solutions to solve $y$. This idea will be developed in the next
topic, elimination theory.
\end{example}

\begin{example}[Sudoku]
  Sudoku is a popular puzzle with $9\times 9$ spaces. The
  goal is to fill in digits from ${1,2,\ldots, 9}$, such
  that each row, each column and each $3\times 3$ sub-box contain digits
  $1$ to $9$. See two Sudoku problems in Figure \ref{sudoku_1}.
\begin{figure}
\centering
\subfloat{
  \begin{lpsudoku}[scale=0.7,title=,titleindent=-3cm]
  \setrow{9}{{}, {}, {}, {}, {}, 8, {}, {}, 4}\setrow{8}{{}, 9, {}, {}, {}, 6, {}, 5, {}}\setrow{7}{6, 5, 2, {}, {}, {}, 1, {}, {}}\setrow{6}{7, {}, {}, {}, {}, 4, {}, 3, {}}\setrow{5}{{}, {}, {}, 3, 1, 9, {}, {}, {}}\setrow{4}{{}, 1, {}, 2, {}, {}, {}, {}, 9}\setrow{3}{{}, {}, 1, {}, {}, {}, 9, 6, 7}\setrow{2}{{}, 7, {}, 6, {}, {}, {}, 8, {}}\setrow{1}{4, {}, {}, 9, {}, {}, {}, {}, {}}
    \end{lpsudoku}}\subfloat{
  \begin{lpsudoku}[scale=0.7,title=,titleindent=-3cm]
    \setrow{9}{5, 3, {}, {}, 7, {}, {}, {}, {}}\setrow{8}{6, {}, {}, 1, {}, 5, {}, {}, {}}\setrow{7}{{}, 9, 8, {}, {}, {}, {}, 6, {}}\setrow{6}{8, {}, {}, {}, {}, {}, {}, {}, 3}\setrow{5}{4, {}, {}, 8, {}, 3, {}, {}, 1}\setrow{4}{7, {}, {}, {}, 2, {}, {}, {}, 6}\setrow{3}{{}, 6, {}, {}, {}, {}, 2, 8, {}}\setrow{2}{{}, {}, {}, 4, 1, 9, {}, {}, 5}\setrow{1}{{}, {}, {}, {}, 8, {}, {}, 7, 9}
    \end{lpsudoku}}
\caption{two Sudoku puzzles}
\label{sudoku_1}
\end{figure}

 Typically people solve Sudoku with backtracking algorithm: try to fill
 in as many digits as possible, and if there is no way to proceed then
 go
 one step back. It can be easily implemented in computer codes, and
 usually it is very efficient. Here we introduce solving Sudoku by
 \GB. This method is not the most efficient way, however, besides
 finding a solution, it illustrates the global structure of
 solutions.

We convert this puzzle to an algebraic problem. Name the digit on
$i$-th row and $j$-th column as $x_{ij}$. $x_{ij}$ must be in
$\{1,\ldots 9\}$. Let,
\begin{equation}
  \label{eq:63}
  F(x)=(x-1)(x-2) \ldots (x-9).
\end{equation}
So there are $81$ equations, $F(x_{ij})=0$. Two spaces in the same
row, or in the same column, or in the same sub-box, cannot contain the
same digit. For example,
$x_{11}\not=x_{12}$. Note this is not an equality, how do we write an
algebraic equation to describe this constraint?

The standard trick to ``differentiate'' polynomials. Consider
$F(y)-F(x)$, where $x$ and $y$ refer to two boxes that cannot contain
the same digit. $F(y)-F(x)$ must be proportional to $y-x$. 
\begin{equation}
  \label{eq:58}
  \frac{F(y)-F(x)}{y-x}=g(x,y).
\end{equation}
where $g(x,y)$ is a polynomial. It is clearly that when $y\not=x$,
$g(x,y)=0$. On the other hand, from the Taylor
series,
\begin{equation}
  \label{eq:64}
  F(y)-F(x)=(y-x)\bigg(F'(x)+\half (y-x) F''(x) + \ldots
  \bigg)=(y-x)g(x,y).
\end{equation}
If $g(x,y)=0$ but $y=x$, then $F'(x)=0$. However $F(x)$ has no multiple
root, that means $F(x)$ and $F'(x)$ cannot be both zero. So  if
$g(x,y)=0$ then $y\not=x$. There are $810$ such equations like
$g(x_{11},x_{12})=0$. Then with the known input information in Sukodu, we have a
polynomial equation system. 

For the first Sudoku, there are $81+810+27=918$ equations. It is really a
large system with high degree polynomials. Amazingly, we can still
solve it by \GB. Using \pmb{slimgb} command in \Singular, and the
number field $Z/11$, this sudoku is solved on a laptop computer with in about
$4.9$ seconds. The output \GB\ is linear and gives the unique solution
of the Sudoku (Figure \ref{sudoku_solution_1}).
\begin{figure}
\centering
\begin{lpsudoku}[scale=0.7,title=,titleindent=-1.5cm]
\setrow{9}{\hw[1], \hw[3], \hw[7], \hw[5], \hw[9], \bf{8}, \hw[6],
  \hw[2], \bf{4}}\setrow{8}{\hw[8], \bf{9}, \hw[4], \hw[1], \hw[2],
  \bf{6}, \hw[7], \bf{5}, \hw[3]}\setrow{7}{\bf{6}, \bf{5}, \bf{2},
  \hw[7], \hw[4], \hw[3], \bf{1}, \hw[9], \hw[8]}\setrow{6}{\bf{7},
  \hw[6], \hw[9], \hw[8], \hw[5], \bf{4}, \hw[2], \bf{3},
  \hw[1]}\setrow{5}{\hw[2], \hw[4], \hw[8], \bf{3}, \bf{1}, \bf{9},
  \hw[5], \hw[7], \hw[6]}\setrow{4}{\hw[5], \bf{1}, \hw[3], \bf{2},
  \hw[6], \hw[7], \hw[8], \hw[4], \bf{9}}\setrow{3}{\hw[3], \hw[2],
  \bf{1}, \hw[4], \hw[8], \hw[5], \bf{9}, \bf{6},
  \bf{7}}\setrow{2}{\hw[9], \bf{7}, \hw[5], \bf{6}, \hw[3], \hw[1],
  \hw[4], \bf{8}, \hw[2]}\setrow{1}{\bf{4}, \hw[8], \hw[6], \bf{9},
  \hw[7], \hw[2], \hw[3], \hw[1], \hw[5]}
  \end{lpsudoku}
\caption{Sudoku with unique solution, which is determined by \GB.}
\label{sudoku_solution_1}
\end{figure}
For Sudoku 2, there
are $919$ equations. \Singular\ takes about $5.1$ seconds on a laptop to get
\GB\ ,
\begin{gather}
 G= \{x_{58}^2-4,x_{11}-5,x_{12}-3,x_{13}-4,x_{14}-6,x_{15}-7,x_{16}-8,x_{17}-9,x_{18}-1,x_{19}-2,\nn\\x_{21}-6,x_{22}-7,x_{23}-2,x_
   {24}-1,x_{25}-9,x_{26}-5,x_{27}-3,x_{28}-4,x_{29}-8,x_{31}-1,x_{32}-9,\nn\\x_{33}-8,x_{34}-3,x_{35}-4,x_{36}-2,x_{37}-5,x_{38}-6,x
   _{39}-7,x_{41}-8,x_{42}+9 x_{58}-9,\nn\\x_{43}+9 x_{58}-2,x_{44}-7,x_{45}+3
   x_{58}-11,x_{46}-1,x_{47}-4,x_{48}+x_{58}-11,x_{49}-3,x_{51}-4,\nn\\
x_{52}+2 x_{58}-9,x_{53}+2 x_{58}-2,x_{54}-8,x_{55}+8
   x_{58}-11,x_{56}-3,x_{57}-7,x_{59}-1,x_{61}-7,\nn\\x_{62}-1,x_{63}-3,x_{64}-9,x_{65}-2,x_{66}-4,x_{67}-8,x_{68}-5,x_{69}-6,x_{71}-
   9,x_{72}-6,x_{73}-1,\nn\\x_{74}-5,x_{75}-3,x_{76}-7,x_{77}-2,x_{78}-8,x_{79}-4,x_{81}-2,x_{82}-8,x_{83}-7,x_{84}-4,x_{85}-1,\nn\\x_{86}
   -9,x_{87}-6,x_{88}-3,x_{89}-5,x_{91}-3,x_{92}-4,x_{93}-5,x_{94}-2,\nn\\x_{95}-8,x_{96}-6,x_{97}-1,x_{98}-7,x_{99}-9\}\,.
\end{gather}
Note that the new feature is that $G$ contains a quadratic polynomial,
which means the solution for this sudoku is not unique. From leading
term counting, there are $2$ solutions. Explicitly, solve the first
equation 
\begin{gather}
  \label{eq:67}
  x_{58}^2=4 \mod 11\,,
\end{gather}
and we get two solutions, $x_{58}=2$ or $x_{58}=9$. Afterwards, we get two
complete  solutions (Figure \ref{sudoku_solution_2}).

\begin{figure}
\centering
\subfloat{
  \begin{lpsudoku}[scale=0.7,title=,titleindent=-1.5cm]
 \setrow{9}{\bf{5}, \bf{3}, \hw[4], \hw[6], \bf{7}, \hw[8], \hw[9], \hw[1], \hw[2]}\setrow{8}{\bf{6}, \hw[7], \hw[2], \bf{1}, \hw[9], \bf{5}, \hw[3], \hw[4], \hw[8]}\setrow{7}{\hw[1], \bf{9}, \bf{8}, \hw[3], \hw[4], \hw[2], \hw[5], \bf{6}, \hw[7]}\setrow{6}{\bf{8}, \hw[2], \hw[6], \hw[7], \hw[5], \hw[1], \hw[4], \hw[9], \bf{3}}\setrow{5}{\bf{4}, \hw[5], \hw[9], \bf{8}, \hw[6], \bf{3}, \hw[7], \hw[2], \bf{1}}\setrow{4}{\bf{7}, \hw[1], \hw[3], \hw[9], \bf{2}, \hw[4], \hw[8], \hw[5], \bf{6}}\setrow{3}{\hw[9], \bf{6}, \hw[1], \hw[5], \hw[3], \hw[7], \bf{2}, \bf{8}, \hw[4]}\setrow{2}{\hw[2], \hw[8], \hw[7], \bf{4}, \bf{1}, \bf{9}, \hw[6], \hw[3], \bf{5}}\setrow{1}{\hw[3], \hw[4], \hw[5], \hw[2], \bf{8}, \hw[6], \hw[1], \bf{7}, \bf{9}}
    \end{lpsudoku}}\subfloat{
  \begin{lpsudoku}[scale=0.7,title=,titleindent=-1.5cm]
 \setrow{9}{\bf{5}, \bf{3}, \hw[4], \hw[6], \bf{7}, \hw[8], \hw[9], \hw[1], \hw[2]}\setrow{8}{\bf{6}, \hw[7], \hw[2], \bf{1}, \hw[9], \bf{5}, \hw[3], \hw[4], \hw[8]}\setrow{7}{\hw[1], \bf{9}, \bf{8}, \hw[3], \hw[4], \hw[2], \hw[5], \bf{6}, \hw[7]}\setrow{6}{\bf{8}, \hw[5], \hw[9], \hw[7], \hw[6], \hw[1], \hw[4], \hw[2], \bf{3}}\setrow{5}{\bf{4}, \hw[2], \hw[6], \bf{8}, \hw[5], \bf{3}, \hw[7], \hw[9], \bf{1}}\setrow{4}{\bf{7}, \hw[1], \hw[3], \hw[9], \bf{2}, \hw[4], \hw[8], \hw[5], \bf{6}}\setrow{3}{\hw[9], \bf{6}, \hw[1], \hw[5], \hw[3], \hw[7], \bf{2}, \bf{8}, \hw[4]}\setrow{2}{\hw[2], \hw[8], \hw[7], \bf{4}, \bf{1}, \bf{9}, \hw[6], \hw[3], \bf{5}}\setrow{1}{\hw[3], \hw[4], \hw[5], \hw[2], \bf{8}, \hw[6], \hw[1], \bf{7}, \bf{9}}
    \end{lpsudoku}}
\caption{Sudoku with multiple solutions, determined by \GB.}
\label{sudoku_solution_2}
\end{figure}

\end{example}

\subsubsection{Elimination theory}
We already see that \GB\ can remove variable entanglement, here we
study this property via elimination theory,

\begin{thm}
\label{Elimination}
  Let $R=\F[y_1,\ldots y_m,z_1,\ldots z_n]$ be a polynomial ring and
  $I$ be an ideal in $R$. Then $J=I\cap \Fpoly$, the {\it elimination
    ideal}, is an ideal of $\Fpoly$. $J$ is generated by $G(I)\cap
  \Fpoly$, where $G(I)$ is the \GB\ of $I$ in \lex\ order with $y_1\succ
  y_2\ldots \succ y_m\succ z_1\succ z_2\ldots \succ z_n$.
\end{thm}
\begin{proof}
  See Cox, Little and O'Shea \cite{MR3330490}. 
\end{proof}
Note that elimination ideal $J$ tells the relations between $z_1\ldots
z_n$, without the interference with $y_i$'s. In this sense, $y_i$'s are ``eliminated''. It is very useful for
studying polynomial equation system. In practice, \GB\ in \lex\ may
involve heavy computations. So frequently, we use block order instead,
$[y_1,\ldots y_m]\succ [z_1,\ldots z_n]$ while in each block \grevlex\
can be applied. 

Eliminate theory applies in many scientific directions, for example, it
transfers tree-level scattering equations (CHY formalism)
\cite{Cachazo:2013iea,Cachazo:2013hca,Cachazo:2013gna,Cachazo:2014nsa,Cachazo:2014xea} with $n$ particles, in $(n-3)$
variables, to a univariate polynomial equation \cite{Dolan:2015iln}. Here we give a simple
example in IMO,
\begin{example}[International Mathematical Olympiad, 1961/1]\ 

\noindent {\bf Problem} Solve the system of equations: 
\begin{eqnarray}
  \label{eq:69}
  x + y + z &=& a\nn \\
x^2 + y^2 + z^2 &=& b^2\nn \\
xy &=& z^2
\end{eqnarray}
where $a$ and $b$ are constants. Give the conditions that $a$ and $b$ must satisfy
so that $x, y, z$ (the solutions of the system) are distinct positive
numbers.

\noindent {\bf Solution} The tricky part is the condition for positive
distinct $x,y,z$. Now with \GB\, this problem can
be solved automatically. 

First, eliminate $x,y$ by \GB\ in \lex\ with $x\succ y \succ z$. For
example, in \mm\ 
\begin{gather*}
 \pmb{ \text{GroebnerBasis}[\{-a+x+y+z,-b^2+x^2+y^2+z^2,x
   y-z^2\},\{x,y,z\},}\\
\pmb{
\text{MonomialOrder}\to \text{Lexicographic},\text{CoefficientDomain}\to \text{RationalFunctions}]}
\end{gather*}
and the resulting \GB\ is,
\begin{equation}
  \label{eq:71}
  G=\left\{a^2\c -2 a z\c -b^2,-a^4+y \left(2 a^3\c +2 a b^2\right)+2 a^2 b^2-4 a^2 y^2-b^4,a^2-2 a
   x\c-2 a y+b^2\right\}.
\end{equation}
The first element is in $\Q(a,b)[z]$, hence it generates the
elimination ideal. Solve this equation, we get,
\begin{equation}
  \label{IMO_z}
  z=\frac{a^2-b^2}{2 a}\,.
\end{equation}
Then eliminate $y,z$ by \GB\ in \lex\ with $z\succ y \succ x$. We get
the equation,
\begin{equation}
  \label{IMO_x}
  a^4+x (-2 a^3-2 a b^2)-2 a^2 b^2+4 a^2 x^2+b^4=0\,.
\end{equation}
To make sure $x$ is real we need the discriminant,
\begin{equation}
  \label{x_discriminant}
  -4 a^2 (a^2-3 b^2) (3 a^2-b^2)\geq 0\,.
\end{equation}
Similarly, to eliminate $x,z$, we use \lex\ with $z\succ x \succ y$
and get
\begin{equation}
  \label{eq:75}
  a^4+y (-2 a^3-2 a b^2)-2 a^2 b^2+4 a^2 y^2+b^4=0\,,
\end{equation}
and the same real condition as \eqref{x_discriminant}. Note that $x$
and $y$ are both positive, if and only if $x,y$ are real, $x+y>0$ and
$x y$.  Hence positivity for $x,y,z$ means,
\begin{align}
  \label{eq:66}
  z=\frac{a^2-b^2}{2 a}&>0\nn\\
  x+y=a-z=a- \frac{a^2-b^2}{2a}&>0\\
 -4 a^2 (a^2-3 b^2) (3 a^2-b^2)&\geq 0.
\end{align}
which implies that,
\begin{equation}
  \label{admissible_region_pre}
  a>0,\quad b^2<a^2\leq 3b^2.
\end{equation}
To ensure that $x$, $y$ and $z$ are distinct, we consider the ideal
in $\Q[a,b,x,y,z]$.
\begin{equation}
  \label{eq:70}
  J=\{-a+x+y+z,-b^2+x^2+y^2+z^2,x
   y-z^2,(x-y)(y-z)(z-x)\}.
\end{equation}
Note that to study the $a$, $b$ dependence, we consider $a$ and $b$
as variables. Eliminate ${x,y,z}$, we have,
\begin{equation}
  \label{eq:74}
  g(a,b)=(a-b) (a+b) (a^2-3 b^2)^2 (3 a^2-b^2) \in J.
\end{equation}
If all the four generators in $J$ are zero for some value of
$(a,b,x,y,z)$, then $g(a,b)=0$. Hence, if $g(a,b)\not=0$,
$x$, $y$ and $z$ are distinct in the solution. So it is clear that
inside the region defined by \eqref{admissible_region_pre}, the subset set
\begin{equation}
  \label{admissible_region}
  a>0,\quad b^2<a^2<3b^2.
\end{equation}
satisfies the requirement of the problem. On the other hand, if
$a^2=3b^2$, explicitly we can check that $x$, $y$ and $z$ are not
distinct in all solutions. Hence $x,y,z$ in a solution are positive
and distinct, if and only if $a>0$ and $b^2<a^2<3b^2$. With \eqref{IMO_z} and 
\eqref{IMO_x}, it is trivial to obtain the solutions. 
\end{example}

\subsubsection{Intersection of ideals}
In general, given two ideals $I_1$ and $I_2$ in $R=\Fpoly$, it is very
easy to get the generating sets for $I_1+I_2$ and $I_1 I_2$. However,
it is difficult to compute $I_1\cap I_2$. Hence again we refer to
\GB\, especially to elimination theory. 

\begin{proposition}
  Let $I_1$ and $I_2$ be two ideals in $R=\Fpoly$. Define 
  $J$ as the ideal generated by $\{t f|f\in I_1\}\cup \{(1-t) g|g\in
  I_2\}$ in $\F[t,z_1,\ldots z_n]$. Then $I_1\cap I_2=J\cap R$, and
  the latter can be computed by elimination theory.   
\end{proposition}
\begin{proof}
  If $f\in I_1$ and $f\in I_2$, then $f=t f+(1-t)f \in J$. So $I_1\cap
  I_2\in J\cap R$. On the other hand, if $F\in J\cap R$, then
  \begin{equation}
    \label{eq:77}
    F(t,z_1,\ldots,z_n)=a(t, z_1,\ldots,z_n) t f(z_1,\ldots,z_n)+b(t, z_1,\ldots,z_n) (1-t) g(z_1,\ldots,z_n)\,,
  \end{equation}
where $f\in I_1$, $g\in I_2$. Since $F\in R$, $F$ is $t$
independent. Plug in $t=1$ and $t=0$, we get,
\begin{equation}
  \label{eq:78}
  F=a(1, z_1,\ldots,z_n) f(z_1,\ldots,z_n),\quad F=b(0, z_1,\ldots,z_n)  g(z_1,\ldots,z_n)\,.
\end{equation}
Hence  $F\in I_1\cap I_2$, $ J\cap R\subset I_1\cap I_2$. 
\end{proof}
In practice, terms like $tf$ and $(1-t)g$ increase degrees by $1$,
hence this elimination method may not be efficient. More efficient
method is given by {\it syzygy} computation \cite[Chapter 5]{opac-b1094391}. 

\subsection{Basic facts of algebraic geometry in affine space II}
In this subsection, we look closer at properties of algebraic sets and
ideals. Consider $I=\{x^2 - y^2, x^3 + y^3 - z^2\}$ in
$\C[x,y,z]$. From naive counting, $\mathcal Z(I)$ is a curve since
there are $2$ equations in $3$ variables. However, the plot of
$\mathcal Z(I)$ (Figure \ref{fig:reducible_curve}) looks like a line
and a cusp curve. So $\mathcal Z(I)$ is {\it reducible}, in the sense
that it can be decomposed into smaller algebraic sets. So we need the
concept of {\it primary decomposition}.
\begin{figure}[ht]
  \centering
  \includegraphics[scale=0.6]{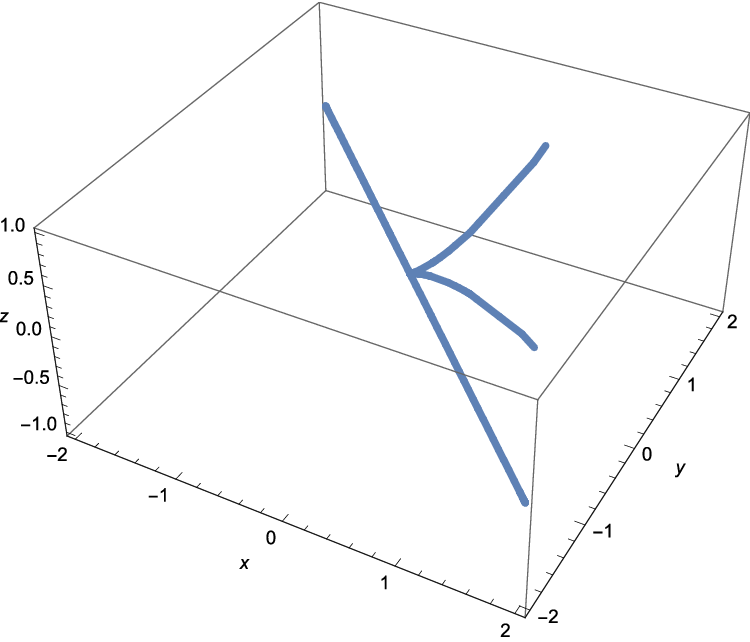}
  \caption{A reducible algebraic set (in blue), defined by $\mathcal Z(\{x^2 - y^2, x^3 + y^3 - z^2\})$.}
  \label{fig:reducible_curve}
\end{figure}
\begin{definition}
  An ideal $I$ in a ring $R$ is called prime, if $\forall a b\in I$ ($a$,
  $b\in R$) then $a\in I$ or $b\in I$. An ideal $I$ in $R$ is called
  primary is if $a b\in I$ ($a$,
  $b\in R$) then $a\in I$ or $b^n\in I$, for some positive integer $n$.
\end{definition}
A prime ideal must be a primary ideal. On the other hand, 
\begin{proposition}
  If $I$ is a primary ideal, then the radical of $I$, $\sqrt I$ is a
  prime ideal.
\end{proposition}
\begin{proof}
  See Zariski and Samuel \cite[Chapter 3]{MR0384768}.
\end{proof}

Note that $I=\{x^2 - y^2, x^3 + y^3 - z^2\}$ is not a prime ideal or
primary ideal. Define $a=x-y$, $b=x+y$, clearly $ab\in I$, but $a\not
\in I$ and $b^n\not \in I$ for any positive integer $n$. (The point
$P=(2,2,4)\in \mathcal Z(I)$. If $(x+y)^n\in I$ then $(x+y)^n|_P=0$. It
is a contradiction.) 

For another example, $J=\la (x-1)^2 \ra$ in $\C[x]$ is primary but
not prime. $\mathcal Z(J)$ contains only one point $\{1\}$ with the multiplicity $2$. $(x-1)(x-1) \in J$ but $(x-1)\not \in J$.  For there
examples, we see primary condition implies that the corresponding
algebraic set cannot be decomposed to smaller algebraic sets, while
prime condition further requires that the multiplicity is $1$.

\begin{thm}[Lasker-Noether]
  For an ideal $I$ in $\Fpoly$, $I$ has the primary decomposition,
  \begin{equation}
    \label{primary_decomposition}
    I=I_1 \cap \ldots \cap I_m\,,
  \end{equation}
such that,
\begin{itemize}
\item Each $I_i$ is a primary ideal in $\Fpoly$,
\item $I_i\not \supset \cap_{j\not=i} I_j$,
\item $\sqrt I_i \not = \sqrt I_j$, if $i\not =j$.
\end{itemize}
Although primary decomposition may not be unique, the radicals $\sqrt I_i$'s are uniquely determined by $I$ up to orders. 
\end{thm}
\begin{proof}
  See Zariski, Samuel \cite[Chapter 4]{MR0384768}.
\end{proof}
Note that unlike Gr\"obner basis, primary decomposition is very
sensitive to the number field. For an ideal $I\subset \Fpoly$, $\F
\subset \mathbb K$, the
primary decomposition results of $I$ in $\Fpoly$ and $\K[z_1,\ldots,z_n]$ can
be different. Primary decomposition can be computed by \Macaulay\ or
\Singular. However, the computation is heavy in general. 

Primary
decomposition was also used for studying string theory vacua \cite{Mehta:2012wk}.
\begin{example}
  Consider $I=\{x^2 - y^2, x^3 + y^3 - z^2\}$. Use \Macaulay\ or
\Singular, we find that, $I=I_1\cap I_2$, where,
\begin{equation}
  \label{eq:68}
  I_1=\la z^2,x+y\ra,\quad I_2=\la 2y^3-z^2,x-y \ra 
\end{equation}
Then $\sqrt I_1=\la z,x+y \ra$ is a prime ideal, where $I_2$ itself is
prime. 
\end{example}

When $I\subset \Fpoly$ has a primary decomposition $I=I_1\cap \ldots
\cap I_m$, $m>1$, then $\mathcal Z_\F(I)=\mathcal Z_\F(I_1) \cup \ldots \cup
\mathcal Z_\F(I_m) $. Then algebraic set decomposed to the union of
sub algebraic sets. We switch the study of reducibility to the geometric
side.  

\begin{definition}
  Let $V$ be a nonempty closed set in $\mathbf A_\F$ in Zariski
  topology, $V$ is irreducible, if $V$ cannot be a union of two
  closed proper subsets of $V$.   
\end{definition}

\begin{proposition}
  Let $\mathbb K$ be an algebraic closed field. There is a one-to-one
  correspondence:
  \begin{equation}
    \begin{array}{ccc}
      \text{prime ideals in }\mathbb K[z_1,\ldots z_n]
        &   &\text{irreducible
        algebraic sets in }\mathbf A_{\mathbb K}\\
I & \longrightarrow & \mathcal Z_\K(I)\\
\mathcal I(V) & \longleftarrow & V\\
\end{array}
  \end{equation}
  \begin{proof} (Sketch)
    This follows from Hilbert Nullstellensatz \eqref{eq:47}. 
  \end{proof}
\end{proposition}
We call an irreducible Zariski closed set ``affine variety''. Similar to primary decomposition of ideals, algebraic set has the
following decomposition,
\begin{thm}
\label{variety_decomposition}
  Let $V$ be an algebraic set. $V$ uniquely decomposes as the union of
  affine varieties, $V=V_1 \cup \ldots \cup V_m$, such that $V_i \not
  \supset V_j$ if $i\not =j$.
\end{thm}
\begin{proof}
Let $I=\mathcal I(V)$. The primary
  decomposition determines that $I=I_1 \cap\ldots \cap I_m$. Since $I$ is a
  radical ideal, all $I_i$'s are prime. Then $V=\mathcal
  Z(I)=\cap_{i=1}^m \mathcal Z(I_i)$. Each  $\mathcal Z(I_i)$ is an
  affine variety.  If $\mathcal Z(I_i) \supset \mathcal Z(I_j)$, then
  $I_i \subset I_j$ which is a violation of radical uniqueness of
  Lasker-Noether theorem. 

If there are two decompositions, $V=V_1 \cup \ldots \cup V_m=W_1\cup
\ldots \cup W_l$. $V_1=V_1\cap(W_1\cup
\ldots \cup W_l)=(V_1\cap W_1)\cup \ldots (V_1\cap W_l)$. Since $V_1$
is irreducible, $V_1$ equals some $V_1\cap W_j$, \WLOG , say
$j=1$. Then $V_1\subset W_1$. By the same analysis $W_1 \subset V_i$
for some $i$. Hence $V_1\subset V_i$ and so $i=1$. We proved
$W_1=V_1$. Repeat this process, we see that the two decompositions
are the same. 
\end{proof}

\begin{example}
\label{dbox_primary_decomposition}
  As an application, we use primary decomposition to find cut
  solutions of $4D$ double box in Table \ref{dbox_sol}. It is quite
  messy to derive all unitarity solutions by brute force computation. In this
  situation, primary decomposition is very helpful.

Use van
  Neerven-Vermaseren variables, the ideal $I=\la D_1,\ldots D_7\ra$
  decomposes as $I=I_1\cap I_2 \cap I_3 \cap I_4 \cap I_5 \cap I_6$. 
  \begin{eqnarray}
    \label{eq:76}
   I_1&=& \{2 y_4-t,s+2 y_2,-t+2 x_3-2 x_4,y_3,\frac{s}{2}+y_1+y_2,x_2-\frac{s}{2},x_1\}\,,\nn\\
   I_2&=& \{t+2 y_4,s+2 y_2,-t+2 x_3+2
   x_4,y_3,\frac{s}{2}+y_1+y_2,x_2-\frac{s}{2},x_1\}\,,\nn\\
   I_3&=& \{s+t+2 y_2+2 y_4,2
  x_4-t,x_3,y_3,\frac{s}{2}+y_1+y_2,x_2-\frac{s}{2},x_1\}\,,\nn\\
   I_4&=& \{s+t+2 y_2-2 y_4,t+2
  x_4,x_3,y_3,\frac{s}{2}+y_1+y_2,x_2-\frac{s}{2},x_1\}\,,\nn\\
   I_5&=& \{s+t+2 y_2+2 y_4,x_4 (2 s+2 t)+y_4 (2 s+2 t)+s t+t^2+4 x_4
  y_4,\nn\\ &&-t+2 x_3-2
  x_4,y_3,\frac{s}{2}+y_1+y_2,x_2-\frac{s}{2},x_1\}\,,\nn\\
 I_6&=& \{s+t+2 y_2-2 y_4,x_4 (-2 s-2 t)+y_4 (-2 s-2 t)+s t+t^2+4 x_4
  y_4,\nn\\ &&-t+2 x_3+2 x_4,y_3,
\frac{s}{2}+y_1+y_2,x_2-\frac{s}{2},x_1\}\,.
  \end{eqnarray}
Each $I_i$ is prime and corresponds to a solution in Table
\ref{dbox_sol}. \Singular\ computes this primary decomposition in
about $3.6$ seconds on a laptop. In practice, the computation can be
sped up if we first eliminate all RSPs.

Hence the unitarity solution set $\mathcal Z(I)$ consists of six
irreducible solution sets $\mathcal Z(I_i)$, $i=1\ldots 6$. Each one
can be parametrized by a free parameter. 
\end{example}

For a variety $V$, we want to define its dimension. Intuitively, we
may test if $V$ contains a point, a curve, a surface...? So the
dimension of $V$ is defined as the length of variety sequence in $V$, 
\begin{definition}
  The dimension of a variety $V$, $\dim V$, is the largest number $n$ in all
  sequences $\emptyset \not=W_0\subset W_1 \ldots \subset W_n\subset V$, where $W_i$'s are
  distinct varieties.
\end{definition}
On the algebraic side, let $V=\mathcal Z(I)$, where $I$ is an ideal in
$R=\Fpoly$. Consider the quotient ring
$R/I$. Roughly speaking, the remaining ``degree of freedom'' of $R/I$
should be the same as $\dim V$. Krull dimension counts ``the degree
of freedom'',
\begin{definition}[Krull dimension]
  The Krull dimension of a ring $S$, is the largest number $n$ in all
  sequences $p_0\subset p_1 \ldots \subset p_n$, where $p_i$'s are
  distinct prime ideals in $S$. 
\end{definition}
If for a prime ideal $I$, $R/I$ is has Krull dimension zero then $I$
is a {\it maximal ideal}. A maximal ideal $I$ in $R$ is an ideal which
such that for any proper ideal $J\supset I$, $J=I$. $I$ is a maximal
idea, if and only if $R/I$ is a field. ($R$ itself is not a maximal
idea of $R$). When $\F$ is
algebraically closed, then any maximal ideal $I$ in $R=\Fpoly$ has the
form \cite{MR3330490},
\begin{equation}
  \label{eq:79}
  I=\la z_1-c_1,\ldots z_n-c_n\ra,\quad c_i\in \F.
\end{equation}
Note that the point $(c_1,\ldots,c_n)$ is zero-dimensional, and
$R/I=\F$ has Krull dimension $0$. More generally,
\begin{proposition}
  If $\F$ is algebraically closed and $I$ a prime proper ideal of
  $R=\Fpoly$. Then the Krull dimension of $R/I$ equals $\dim \mathcal
  Z(I)$. 
\end{proposition}
\begin{proof}
  See Hartshorne \cite[Chapter 1]{MR0463157}. Note that Krull
  dimension of $R/I$ is different from the linear dimension $\dim_\F R/I$.
\end{proof}
In summary, we has the
algebra-geometry dictionary (Table \ref{AG_dictionary}), where the last two rows hold if $\F$ is algebraic closed.
\begin{table}
\centering
  \begin{tabular}{ccc}
    Algebra & & Geometry  \\
\hline
    Ideal $I$ in $\Fpoly$ & & algebraic set $\mathcal Z(I)$\\
      $I_1\cap I_2$   & & $\mathcal Z(I_1\cap I_2)=\mathcal Z(I_1)\cup
                          \mathcal Z(I_2)$ \\
      $I_1+ I_2$   & & $\mathcal Z(I_1+ I_2)=\mathcal Z(I_1)\cap 
                          \mathcal Z(I_2)$\\
                       $I_1\subset I_2$ &$\Rightarrow$ & $\mathcal Z(I_1)\supset
                                           \mathcal Z(I_2)$\\
prime ideal $I$ & $\Rightarrow$ & $\mathcal Z(I)$ (irreducible) variety\\
maximal ideal $I$ & $\Rightarrow$ & $\mathcal Z(I)$ is a point \\
Krull dimension of $\dim \Fpoly/I$ &$=$ & $\dim \mathcal Z(I) $
  \end{tabular}
\caption{algebraic geometry dictionary}
\label{AG_dictionary}
\end{table}

We conclude this section by an example which applies Gr\"obner basis,
primary decomposition and dimension theory.
\begin{example}
\label{Galois_group}
  (Galois theory) Galois theory studies the symmetry of a field
  extension, $\F \subset \K$ by the Galois group
  $\text{Aut}(\K/\F)$. Historically, Galois group of a polynomial is
  defined to be the permutation group of roots, such that algebraic
  relations are preserved. Galois completely determined
  if a polynomial equation can be solved by radicals. In practice,
  given a polynomial to find its Galois group may be difficult. Here
  we introduce an automatic method of computing Galois group. 

For example, consider the polynomial $f(x)=x^4+3 x+3$ in $\Q[x]$. It is irreducible
over $\Q[x]$ and contains no multiple root in $\mathbb C$. We denote
the four distant roots as $x_1$, $x_2$, $x_3$, $x_4$. To ensure that
these variables are distant, we use a classic trick in algebraic
geometry: auxiliary variable. 
Introduce a new variable $w$, define
that 
\begin{equation}
  \label{Galois_ideal}
  I=\la f(x_1), f(x_2), f(x_3), f(x_4),
  w(x_1-x_2)(x_1-x_3)(x_1-x_4)(x_2-x_3)(x_2-x_4)(x_3-x_4)-1 \ra\,.
\end{equation}
It is clear that in $\C[x_1,x_2,x_3,x_4,w]$, $\mathcal Z(I)$ is a finite
set (for example via Gr\"obner basis computation.) The four variables must be distinct on the solution set, because
of the last generator in \eqref{Galois_ideal}. Back to
$\Q[x_1,x_2,x_3,x_4,w]$,  we want to find more algebraic
relations over $\Q$ which are ``consistent'' with $I$. That is to find
a maximal ideal $J$ in $\Q[x_1,x_2,x_3,x_4,w]$, $I\subset J$. In
practice, we use primary decomposition and find that in
$\Q[x_1,x_2,x_3,x_4,w]$,
\begin{equation}
  \label{eq:72}
  I=I_1\cap I_2 \cap I_3\,,
\end{equation}
where explicitly each $I_i$ is prime. Since $dim_\Q
(\Q[x_1,x_2,x_3,x_4,w]/I)$ is finite,
\begin{equation}
  \label{eq:73}
  dim_\Q
(\Q[x_1,x_2,x_3,x_4,w]/I_1)<\infty\,.
\end{equation}
$I_1$ is prime hence $\Q[x_1,x_2,x_3,x_4,w]/I_1$ has no zero
divisor. A finite-dimensional $Q$-algebra with no zero divisor must be
a field. Hence $I_1$ is a maximal ideal of $\Q[x_1,x_2,x_3,x_4,w]$.

Compute the Groebner basis of $I_1$ with the block order $[w]\succ
[x_1,x_2,x_3,x_4]$, we have 
\begin{gather}
  \label{eq:80}
  G(I_1)=\{x_1+x_2+x_3+x_4,2 x_4^2+2 x_2 x_4+2 x_3 x_4+x_2+x_3-3,
2 x_3
   x_2+x_2+x_3+3,\nn\\ x_2^2\c -x_2+x_3^2\c-x_3,4 x_4^3\c -2 x_4^2+6 x_4+5 x_2+5 x_3+9,2 x_4
   x_3^2+x_3^2+2 x_4^2 x_3\c-3 x_3+x_4^2\c-3 x_4\c-3,\nn\\4 x_3^3\c-2 x_3^2+x_3\c-5 x_2+9,315 w\c-2
   x_3^2-4 x_4 x_3-2 x_4^2+3\}
\end{gather}
Except the last one, polynomials in $G(I_1)$ provides all the
algebraic relations over $\Q$ of the four roots. Note that some
relations are
trivial like $x_1+x_2+x_3+x_4=0$
which comes from coefficients of $f(x)$. Some relations like $2 x_3
x_2+x_2+x_3+3=0$, are
nontrivial.

Consider all $24$ permutations of $(x_1, x_2, x_3, x_4)$, we find
the $8$ of them preserves algebraic relations in $G(I_1)$, explicitly,
\begin{eqnarray}
  \label{eq:82}
 &&(x_1,x_2,x_3,x_4),(x_1,x_3,x_2,x_4),(x_2,x_1,x_4,x_3),(x_3,x_1,x_4,x_2),\nn\\
&&(x_2,x_4,x_1,x_3),(x_3,x_4,x_1,x_2),(x_4,x_2,x_3,x_1),(x_4,x_3,x_2,x_1)\,.
\end{eqnarray}
Hence Galois group of the $x^4+3 x+3$ is the dihedral group $D_4$. Clearly, this process
applies to all irreducible polynomials without multiple root.

Note that $\Q[x_1,x_2,x_3,x_4,w]/I_1$ actually is the splitting field
of this polynomial. 
\end{example}

\section{Multi-loop integrand reduction via Gr\"obner basis}
With the knowledge of basic algebraic geometry, now
multi-loop integrand reduction is almost a piece of cake. We apply \GB\ method \cite{Zhang:2012ce,Mastrolia:2012an}.

Consider the algorithm of direct integrand reduction (IR-D). Suppose that all terms with
denominator set $\mathcal D$, $\{D_1,\ldots D_k\}\subsetneqq \mathcal D$ are already
reduced, then,
\begin{enumerate}
\item Collect all integrand terms with inverse propagators $D_1,\ldots D_k$, which
  include terms from Feynman rules and also terms from the integrand
  reduction of parent diagrams. Denote the sum as $N/(D_1,\ldots D_k)$.
\item Define $I=\la D_1,\ldots ,D_k\ra$. Compute the \GB\ of $I$ in \grevlex,
  $G(I)=\{g_1,g_2,\ldots ,g_m\}$. 
\item Polynomial division $N=a_1 g_1+\ldots a_m g_m+\Delta$. Use
  Gr\"obner basis convention relation, rewrite the division as $N=q_1
  D_1+\ldots q_k D_k+\Delta$. 
\item Add $\Delta/(D_1 \ldots D_k)$ to the final result. Keep terms
  \begin{equation}
    \label{eq:86}
    \frac{q_1}{\hat{D_1}D_2 \ldots D_k} + \frac{q_2}{D_1\hat{D_2}
      \ldots D_k} +\ldots \frac{q_k}{D_1D_2\ldots \hat{D_k}}\,,
  \end{equation}
for child diagrams. 
\end{enumerate}
Repeat this process, until all terms left are integrated to zero (like
massless tadpoles, integral without loop momenta dependences). 

Integrand reduction (IR-U) is more subtle. Again, Suppose that all diagrams with
denominator set $\mathcal D$, $\{D_1,\ldots D_k\}\subsetneqq \mathcal D$ are 
reduced, then,
\begin{enumerate}
\item Define $I=\la D_1,\ldots D_k \ra$. Compute the \GB\ of $I$ in \grevlex\
  with numeric kinematics,
  $G(I)=\{g_1,g_2,\ldots ,g_m\}$.
\item Identify all degree-one polynomials in $G(I)$, and solve them
  linearly. The dependent variables are RSPs. Define $J$ as the
  ideal obtained by eliminate all RSPs in $I$.
\item Make a numerator ansatz $N$ in ISPs, with the power counting
  restriction from renormalization conditions. Divide $N$ toward
  $G(J)$, the remainder $\Delta$ is the integrand basis.
\item Cut all propagators by $D_1=\ldots =D_k=0$. Classify all solutions
  by the primary decomposition of $J$ and get $n$ irreducible solutions.
\item On the cut, compute the tree products summed over internal
  spins/helicities. Subtract all known parent diagrams on this cut. The
  result should be a list of $n$ functions $S_i$, defined on each cut
  solution.
\item Fit coefficients of $\Delta$ from $S_i$'s.
\end{enumerate}
We have some comments here: 
\begin{itemize}
\item To make an integrand basis with
undetermined coefficients, we only need Gr\"obner basis with numeric
kinematic conditions.
\item RSPs can be automatically found, because any
degree-one polynomial in $I$ should be a linear combination of
degree-one polynomials in $G(I)$, via Algorithm
\ref{multivariate_polynomial_division}. Hence linear algebra computation
determines RSPs.
\item Integrand basis should not contain RSPs. Furthermore, it is
  helpful to eliminate RSPs before the primary decomposition. 
\item If the cut solution is complicated, primary decomposition helps
  finding all of solutions. And in general, solution sets cannot be
  parameterized rationally before primary decomposition.
\end{itemize}
The key idea of these algorithms
is that polynomial division via \GB\ provides the simplest integrand,
in the sense that the resulting numerator does not contain any term
which are divisible by denominators.

Back to our double box examples, we use algebraic geometry methods to
automate most of the computations.
 Given $7$ propagators in Van
Neerven-Vermaseren variables, we use number field $\F=\Q(s,t)$, define
the ideal $I=\la D_1 , D_2, \ldots D_7\ra$.

First, we determine the RSPs. Compute $G(I)$ in
\grevlex, with numeric kinematics, $t\to -3, s\to 1$. We find that $G(I)$ contains $4$ linear polynomials,
\begin{gather}
  \label{eq:83}
  \{y_3,\frac{1}{2}+y_1+y_2,x_2-\frac{1}{2},x_1\}\subset G(I)\,.
\end{gather}
 This
allow us to define RSPs: we have $4$ linear polynomials and $5$
variables, pick up $y_1$ to be the free variable. And then we
determined $x_1,x_2,y_2,y_3$ are RSPs. (If needed, the full RSP relations can be
obtained from Groebner basis conversion.)
\begin{gather}
  \label{dbox_RSP_2}
  x_1=\frac{D_1-D_2}{2},\quad x_2=\frac{D_2-D_3}{2}+\frac{s}{2}\,,\nn\\
  y_2=\frac{D_4-D_6}{2}-\frac{s}{2}-y_1,\quad y_3=\frac{-D_6+D_7}{2}\,.
\end{gather}

Then, we consider to eliminate RSPs. Define $J$ to be an ideal in
$\F[x_3,y_1,x_4,y_4]$, which is the ideal after RSP elimination. With
numeric kinematics, the \GB\ of $J$ in \grevlex\ and $y_4\succ
x_4\succ y_1\succ x_3$ is,
\begin{gather}
  G(J)=\{-4 x_3^2-12 x_3+4 x_4^2-9,20 x_3 y_1+4 x_4 y_4+6 x_3+6 y_1+9,-4
  y_1^2-12 y_1+4 y_4^2-9,\nn\\
4 x_3^2 y_4+20 x_4 x_3 y_1+12 x_3 y_4+6 x_4 y_1+6 x_4 x_3+9 x_4+9
y_4,\nn\\
4 x_4 y_1^2+12 x_4 y_1+20 x_3 y_4 y_1+6 x_3 y_4+9 x_4+6 y_4 y_1+9 y_4,
4 x_3^2 y_1^2+2 x_3^2 y_1+2 x_3 y_1^2+3 x_3 y_1,\nn\\ 
80 x_3^2 y_1 y_4+16 x_3^2 y_4+40 x_3 y_1 y_4+18 x_3 y_4-6 x_4 y_1+24
x_4 x_3-9 x_4-9 y_4\}.
\label{dbox_GB}
\end{gather}
Note that the first $3$ polynomials are just  equations in
\eqref{dbox_quadratic}. However, the rest algebra relations in
\eqref{dbox_GB} are not obtained by the naive generalization of OPP
method. So previously we got a redundant basis. 

Consider the numerator in ISPs only,
\begin{eqnarray}
  \label{eq:28}
  N_\text{dbox}=\sum_m \sum_n \sum_\alpha \sum_\beta
  c_{mn\alpha\beta}' x_3^{m} y_1^{n} x_4^{\alpha} y_4^{\beta} ,
\end{eqnarray}
where $c_{mn\alpha\beta}'$ are indeterminate coefficients. By
renormalization condition, there $160$ such $c$'s. Divide $
N_\text{dbox}$ by $G(I)$, we get the remainder,
\begin{eqnarray}
  \label{eq:28}
  \Delta_\text{dbox}=\sum_{(m,n,\alpha,\beta)\in S} 
  c_{mn\alpha\beta} x_3^{m} y_1^{n} x_4^{\alpha} y_4^{\beta} ,
\end{eqnarray}
where the index set $S$ contains $32$ elements,
\begin{gather}
 (0, 0, 0, 0), (1, 0, 0, 0), (2, 0, 0, 0), (3, 0, 0, 0), (4, 0, 0, 0), 
(0, 1, 0, 0), (1, 1, 0, 0), (2, 1, 0, 0),\nn\\ (3, 1, 0, 0), (4, 1, 0, 0), 
(0, 2, 0, 0), (1, 2, 0, 0), (0, 3, 0, 0), (1, 3, 0, 0), (0, 4, 0, 0), 
(1, 4, 0, 0), \nn\\ (0, 0, 1, 0), (1, 0, 1, 0), (2, 0, 1, 0), (3, 0, 1, 0),
(0, 1, 1, 0), (0, 0, 0, 1), (1, 0, 0, 1), (2, 0, 0, 1),\nn\\ (3, 0, 0, 1), 
(4, 0, 0, 1), (0, 1, 0, 1), (1, 1, 0, 1), (0, 2, 0, 1), (1, 2, 0, 1), 
(0, 3, 0, 1), (1, 3, 0, 1).
\label{dbox_integrand_basis_GB}
\end{gather}
Note that the number of terms in $\Delta_\text{dbox}$ matches
the number of independent relations from unitarity
cuts. \eqref{dbox_integrand_basis_GB} is the integrand basis of the $4D$
double box. Of these $32$ terms, the last $16$ terms integrated to
zero by Lorentz symmetry, so they are spurious terms. 

In Example \ref{dbox_primary_decomposition}, we already used primary
decomposition to find all unitarity-cut solutions. Note that there is
shortcut'' it is enough to consider the primary decomposition of
$J$. On a laptop computer, it takes only $0.22$ seconds to finish. Using
\eqref{dbox_integrand_basis_GB}  and $4D$ tree amplitudes, we can easily
determine the double box integrand for (super)-Yang-Mills
theory \cite{Badger:2012dp, Zhang:2012ce}.

For $D=4-2\epsilon$, we need to introduce $\mu$ variables,
\begin{gather}
  \label{eq:88}
  l_i=l_i^{[4]}+l_i^\perp , \quad i=1,\ldots ,L,\nn \\
  \mu_{ij}=-l_i ^\perp \cdot   l_j ^\perp, \quad 1\leq i \leq j\leq L.
\end{gather}
In this case, we have
further simplification: $I=\la D_1,\ldots D_k \ra$ must be a prime
ideal, hence it is not necessary to consider the primary decomposition
of $I$ \cite{Zhang:2012ce, Badger:2013gxa}.

\begin{example}
  Consider two-loop five-gluon pure Yang-Mills planar amplitude, with
  helicity $(+++++)$. Note that tree-level all-plus-helicity $5$-gluon
  amplitude
  in Yang-Mills theory is zero, while the one-loop-level is
  finite. The two-loop amplitude is much more challenging. We
  used algebraic geometry method to compute this amplitude.  

For the integrand, we use both IR-D and IR-U methods \cite{Badger:2013gxa}. Note that this
amplitude is well-define only with $D=4-2\epsilon$. Repeat the
integrand reduction process, we get all the diagrams with
non-vanishing integrands in Figure \ref{2l5g_planar} (and their
permutations). For example, the box-pentagon diagram for this
amplitude has a simple integrand,
\begin{align}
  &\Delta_{431}(1^+,2^+,3^+,4^+,5^+) =
  \nonumber\\&
  -\frac{i \, s_{12}s_{23}s_{45} \, F_1(D_s,\mu_{11},\mu_{22},\mu_{12})}{\A12\A23\A34\A45\A51
  \trfive}
  \left( \tr_+(1345) (l_1 + k_5)^2 + s_{15}s_{34}s_{45}
  \right)
  \label{delta431}
\end{align}
where
\begin{align}
  F_1(D_s,\mu_{11},\mu_{22},\mu_{12})
  = (D_s-2)\left( \mu_{11}\mu_{22} + \mu_{11}\mu_{33} + \mu_{22}\mu_{33} \right) + 16\left(
  \mu_{12}^2 - \mu_{11}\mu_{22} \right),
  \label{eq:F1def}
\end{align}
and $\mu_{33}=\mu_{11}+\mu_{22}+2\mu_{12}$ and $D_s$ is the dimension
for internal states. \cite{Badger:2013gxa}. 
\begin{align}
\trfive &= \tr \! \left( \gamma_5 \feyn{k}_1 \feyn{k}_2 \feyn{k}_3
          \feyn{k}_4 \right) \; = \; \B12 \A23 \B34 \A41 - \A12 \B23
          \A34 \B41 .\nn \\
\tr_{\pm}(abcd) &= \frac{1}{2} \tr \! \big( (1 \pm \gamma_5) \feyn{k}_a \feyn{k}_b \feyn{k}_c \feyn{k}_d \big),
\end{align}
Results from IR-D and IR-U
match each other.  After getting these simple integrand, the complete
integrals and final analytic result for this amplitude was obtained by
differential equation method \cite{Gehrmann:2015bfy}.

All-plus two-loop five-gluon non-planar integrand and all-plus
two-loop six-gluon integrand were also obtained by integrand reduction
method \cite{Badger:2015lda,Badger:2016ozq}. 

\begin{figure}[h]
  \centering
\subfloat{\includegraphics[width=3cm]{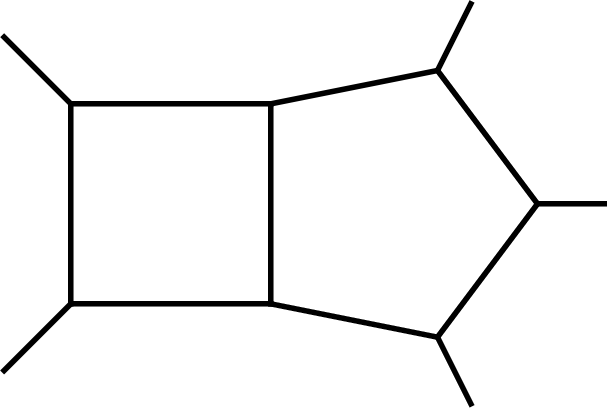}}\quad
\subfloat{\includegraphics[width=3cm]{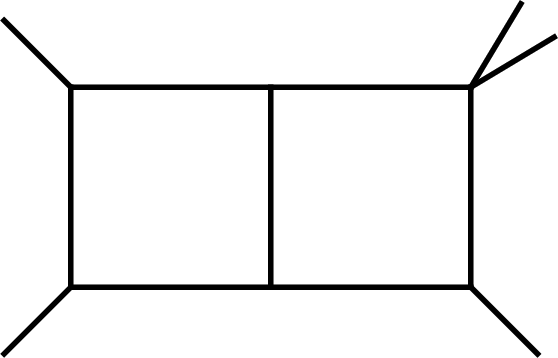}}\quad
\subfloat{\includegraphics[width=3cm]{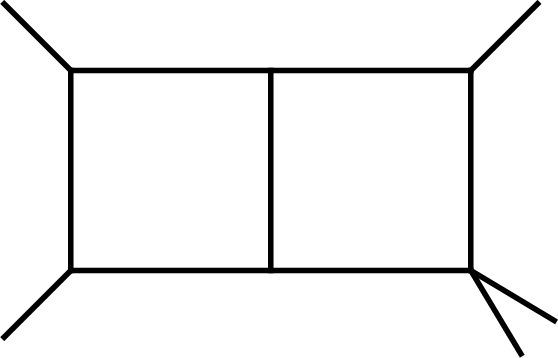}}\quad
\subfloat{ \includegraphics[width=3cm]{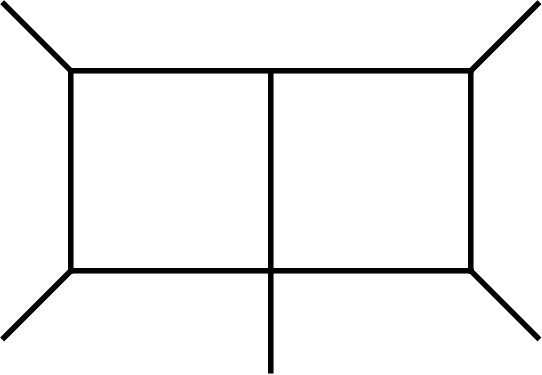}}\\
\subfloat{\includegraphics[width=3cm]{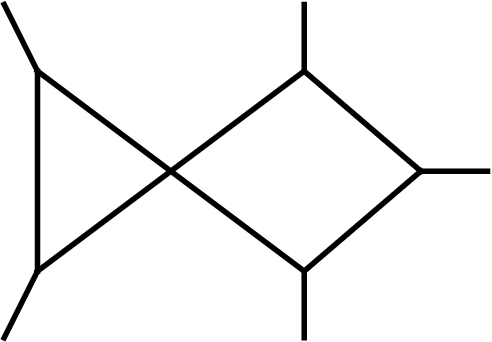}}\quad
\subfloat{\includegraphics[width=3cm]{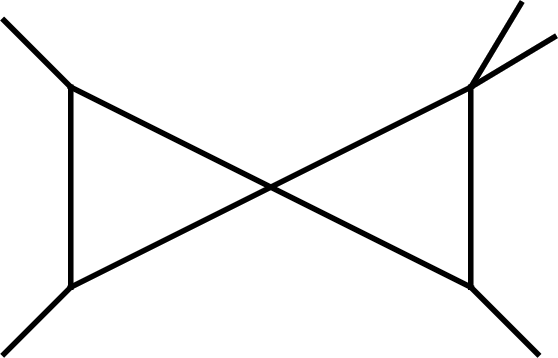}}\quad
\subfloat{ \includegraphics[width=3cm]{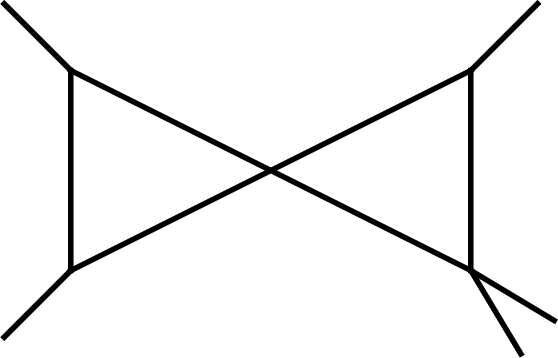}}\quad
\subfloat{ \includegraphics[width=3cm]{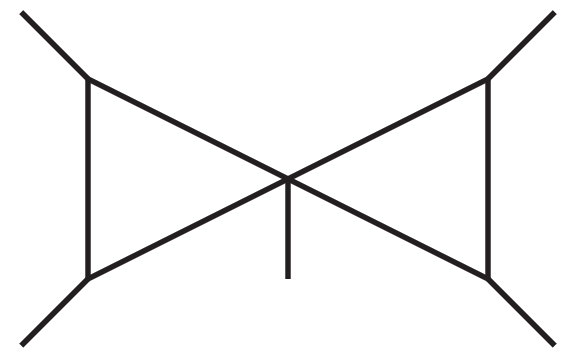}}
\caption{Nonzero diagrams from the integrand reduction, for
  $(+++++)$-helicity two-loop five-gluon planar amplitude.}
       \label{2l5g_planar}
\end{figure}
\end{example}

See \cite{Feng:2012bm, Mastrolia:2013kca,Mastrolia:2016dhn} for more
example of CAG based integrand reductions.

\section{Exercises}
\begin{ex}
  Derive the integrand basis of a box diagram with
  $k_1^2=k_2^2=k_3^2=k_4^2=m^2$ and inverse propagators 
$D_1=l_1^2$, $D_2=(l_1-k_1)^2$,  $D_3=(l_1-k_1-k_2)^2$ and 
 $D_4=(l_1+k_4)^2$, via OPP approach \cite{Ossola:2006us,Ossola:2007ax}.
\end{ex}

\begin{ex}[Basic operations of ideals] $I$ and $J$ are two ideals in $\Fpoly$. Define
  $I+J=\{f+g|f\in I, \ g\in J\}$ and $IJ$
      as the ideal generated by the set $\{f g|f\in I, \ g\in
      J\}$. 
  \begin{enumerate}
  \item Prove that  $\sqrt I$, $I+ J$, $I\cap J$ are ideals.
\item Prove that $IJ \subset I\cap J$ and $\sqrt{I\cap J}=\sqrt{IJ}$.
\item Let $I=\la y(y-x^2)\ra$, $J=\la x y\ra$ in $\Q[x,y]$, determine
  $Z_\C(I+J)$, $Z_\C(I\cap J)$ and $Z_\C(IJ)$. Compute generating sets of $I+J$, $I\cap J$ and $IJ$. Is
   $IJ$ the same as $I\cap J$ in this case? Compute generating sets
   of $\sqrt{I\cap J}$ and $\sqrt{IJ}$. 
\end{enumerate}
 \end{ex}

\begin{ex}[Hilbert's weak Nullstellensatz]
Let $f_1(x)=2 x - 4 x^2 + x^3$ and $f_2(x)=x^2-1$. Prove that as an
ideal in $\Q[x]$, $\la f_1,f_2\ra=\la 1 \ra$. Explicitly find two
polynomials $h_1(x)$ and $h_2(x)$ in $Q[x]$ such that,
\begin{equation}
  \label{eq:39}
  h_1(x) f_1(x) + h_2(x) f_2(x) =1\,.
\end{equation}
(Hint: use Euclid's algorithm, Algorithm \ref{Euclid}.)
\end{ex}

\begin{ex}[Zariski topology]
  Prove \eqref{algebraic_set_intersection} and
  \eqref{algebraic_set_union}. 
\begin{eqnarray}
    \bigcap_i \mathcal Z(I_i) &=&\mathcal Z(\bigcup_i I_i )\,.
\nn \\
    \mathcal Z(I_1)\bigcup \mathcal Z(I_2) &=&\mathcal Z( I_1 I_2
    )=\mathcal Z( I_1 \cap I_2 )\,.
  \end{eqnarray}
\end{ex}

\begin{ex}[Elimination theory] \

  \begin{enumerate}
\item Use computer software like \mm, {\sc Maple}, \Singular\  or 
\Macaulay,  to eliminate $y$ and $z$ from 
\begin{equation}
  \label{eq:81}
  I=\la -x^3-x z+y^2-1,x^2+x z+y^2,x y+x z+y\ra\,,
\end{equation}
to get a equation in $x$ only. How many common zeros are there for
the three polynomials over $\C$? 

\item Use computer software to find the projection of the curve
  $\mathcal C$,
  \begin{equation}
    \label{eq:84}
    \mathcal C:\quad x^2 + x y + z^2= x^2 - z y - z^3+1=0\,,
  \end{equation}
on $x$-$y$ plane. 

\end{enumerate}
\end{ex}

\begin{ex}[Polynomial division via \GB]
Let $f_1=y^2 - x^3 - 1$, $f_2=x y + y^2 + 1$ and $f_3=y^2 + x - y$.
Use {\sc Maple} or {\sc Macaulay} to find the \GB\ $G=\{g_1, \ldots
 , g_m\}$ and the conversion,
\begin{equation}
  \label{eq:85}
  g_j =\sum_{i=1}^3 f_i a_{ij} .
\end{equation}
Reduce the fraction $1/(f_1 f_2 f_3)$ as,
\begin{equation}
  \label{eq:87}
  \frac{1}{f_1 f_2 f_3}=\frac{q_1}{f_2 f_3}+\frac{q_2}{f_1 f_3}+\frac{q_3}{f_1 f_2}
\end{equation}
where $q_1$, $q_2$ and $q_3$ are polynomials in $x$ and $y$. 
\end{ex}

\begin{ex}[Primary decomposition]
Use  {\sc Macaulay} or \Singular\ to find the primary decomposition of
$I=\la x z - y^2, x^3 - y z\ra$. Then parameterize each 
irreducible closed set. 
\end{ex}

\begin{ex}[Galois group and primary decomposition]
Use the method in Example \ref{Galois_group}, to determine the Galois group of
$x^4-10 x^2+1$.
\end{ex}

\begin{ex}[Integrand basis via Gr\"obner basis]
Massless crossed box diagram is the two-loop diagram with 
$k_1^2=k_2^2=k_3^2=k_4^2=0$ and inverse propagators
$D_1=l_1^2$, $D_2=(l_1 - k_1)^2$, $D_3=(l_1 - k_1 - k_2)^2$, $D_4=l_2^2$, $D_5=(l_2 - k_4)^2$, $D_6=(l_1 + l_2 - 
   k_1 - k_2 - k_4)^2$, $D_7=(l_1 + l_2)^2$.
   \begin{figure}[ht]
     \centering
\includegraphics[scale=0.9]{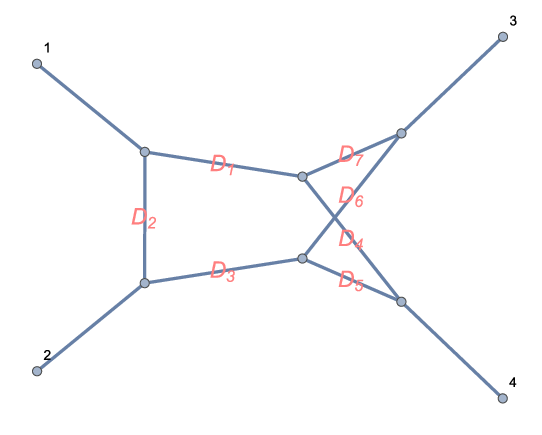}
     \caption{crossed box diagram}
   \end{figure}
Find the $4D$ integrand basis via  Gr\"obner basis.
\end{ex}

\begin{ex}[Fit integrand basis from unitarity cuts]
Use the $4D$ double box integrand basis  \eqref{dbox_integrand_basis_GB} to determine
the  double box integrand form of the $4D$ $(--++)$ and $(-+-+)$ helicity
color-ordered amplitude in pure-Yang-Mills theory. (Hint: see \cite{Badger:2012dp}.)
\end{ex}

\chapter{ Unitarity Cuts and Several Complex Variables}
\section{Maximal unitarity}
Besides integrand reduction method for loop amplitudes, we can also
consider (generalized) unitarity with residue approach \cite{Bern:1994zx,Bern:1994cg,Bern:1995db,Britto:2004nc,Britto:2005fq}
\begin{equation}
A_n^\text{$L$-loop}=\sum_{i} c_i I_i + \text{rational terms}\,. 
\label{MI}
\end{equation}
The set $\{I_k\}$ is the master integral (MI) basis, i.e., minimal
linear basis of Feynman integrals. For example, for one-loop order, we
have {\it scalar} box, triangle, bubble (and tadpole) integrals
. The MI basis is usually a proper subset of the integrand basis like
\eqref{box_integrand_basis}, since spurious terms are removed and
integration-by-parts (IBP) identities are used.

{\it Maximal unitarity} method gets coefficients $c_i$'s for a scattering
process, from contour integrals. (Usually contour integrals are 
simpler than Euclidean Feynman integrals. )
Let $k$ be
the largest number of propagators for all integrals in MI basis. Suppose that there are $d(k)$ diagrams with
exactly $k$ propagators in the master integral list, $\mathscr
D_1,\ldots, \mathscr D_{d(k)}$. 

Maximal
unitarity method first separate (\ref{MI}) as,
\begin{gather}
A_n^\text{$L$-loop}=\sum_{\alpha=1}^{d(k)}\sum_j c_{\alpha,j} I_{\alpha,j} + \big(\text{simpler integrals}\big)+\text{rational terms}
\label{maximal_unitarity}
\end{gather}
where for fixed $\alpha$, $I_{\alpha,j}$'s stand for all master integrals
associated with the diagram $\mathscr D_\alpha$.  ``Simpler integrals'' stands for integrals with
fewer-than-$k$ propagators.

The coefficients $c_{\alpha,j}$'s can be obtained
by maximal unitarity as follows: Let the propagators of $\mathscr
D_\alpha$ be $D_1,
\ldots, D_k$. For simplicity, we drop the index $\alpha$. The cut
equation is,
\begin{equation}
  \label{eq:90}
  D_1=\ldots =D_k=0
\end{equation}
which has $m$ independent solutions. In algebraic geometry language,
the ideal $I=\la D_1 \ldots D_k \ra$ has the primary decomposition,
\begin{equation}
  \label{eq:89}
  I=I_1 \cap \ldots \cap I_m\,.
\end{equation}
Each independent solution is an (irreducible) variety, $V_i=\mathcal
Z(I_i)$. For an {\it integer value } of the spacetime dimension $D$, we replace a generic
Feynman integral as a contour integral,
\begin{eqnarray}
  \label{contour_integral}
  \int \frac{d^D l_1}{i\pi^{D/2}} \ldots \frac{d^D l_L}{i\pi^{D/2}}
\frac{N(l_1,\ldots l_L)}{D_1^{}\ldots
  D_k^{}}
&\rightarrow &\oint \frac{d^D l_1}{(2\pi i)^{D}} \ldots \frac{d^D
               l_L}{(2\pi i )^{D}}
\frac{N(l_1,\ldots l_L)}{D_1^{} \ldots
              D_k^{}}\nonumber \\
&=& \sum_{i=1}^m \sum_b w_b^{(i)} \oint_{\mathcal C_b^{(i)}} \Omega^{(i)}(N)\,.
\end{eqnarray}
In the first line, we have a $DL$-fold contour integral. Part of the
contour integrals serve as ``holomorphic'' Dirac delta functions in
$D_1,\ldots ,D_k$, and the original integral becomes $(\dim V_i)$-fold
contour integrals on each $V_i$.  $\mathcal C_b^{[i]}$'s are 
non-trivial contours on $V_i$ for this integrand, which consists of poles
in the integrand and fundamental cycles of $V_i$ 
On each cut
solution, the original numerator $N(l_1,\ldots,l_L)$ becomes,
\begin{equation}
  \label{eq:91}
  N(l_1,\ldots,l_L) \big|_{V_i}= S^{(i)}.
\end{equation}
where $S^{(i)}$ is the sum of products of tree amplitudes obtained
from the maximal cut. In general, there may be several nontrivial contours
on $V_i$, so for each one we set up a weight $w_b^{(i)}$ to be determined
later. 

We demand that if the original integral is zero, or can be reduce to
integrals with fewer propagators by IBPs, the corresponding contour
integral is zero. If
\begin{eqnarray}
  \label{eq:92}
  \int \frac{d^D l_1}{i\pi^{D/2}} \ldots \frac{d^D l_L}{i\pi^{D/2}}
\frac{F(l_1,\ldots l_L)}{D_1^{}\ldots
  D_k^{}}=0\,,\quad \text{($F$ is spurious)},
\end{eqnarray}
or
\begin{eqnarray}
\int \frac{d^D l_1}{i\pi^{D/2}} \ldots \frac{d^D l_L}{i\pi^{D/2}}
\frac{F(l_1,\ldots l_L)}{D_1^{}\ldots
  D_k^{}}= (\text{simpler integrals})\,,\quad \text{(IBP relation)},
\end{eqnarray}
Then 
\begin{equation}
  \label{eq:93}
  \sum_{i=1}^m \sum_b w_b^{(i)} \oint_{\mathcal C_b^{(i)}}
\Omega^{(i)}(F)=0\,.
\end{equation}

Spurious terms and IBPs fix $w_b^{(i)}$'s up to the normalization
of master integrals. To extract the
coefficients $c_i$'s in (\ref{MI}), we can find a special set of weights
$w_{b,j}^{(i)}$ such that,
\begin{equation}
  c_j =   \sum_{i=1}^m \sum_b w_{b,j}^{(i)} \oint_{\mathcal C_b^{(i)}} \Omega^{(i)}(S^{(i)}).
\label{linear_fitting_c}
\end{equation}
After getting all $k$-propagator master integrals' coefficients, we repeat this
process for $(k-1)$ propagator integrals. We need spurious terms
, IBPs and parent integral exclusion conditions, to fix the
contour weights.

For example, consider the $4D$ massless four-point
amplitude. \eqref{maximal_unitarity} reads.
\begin{gather}
A_4^\text{$1$-loop}=c_{box} I_{box} + \ldots \,.
\end{gather}
From \eqref{box_solution}, $D_1=D_2=D_3=D_4=0$ has two
solutions. Change the original integral to contour integrals,
\begin{eqnarray}
  \label{eq:94}
  \int \frac{d^4 l_1}{i \pi^2} \frac{1}{D_1 D_2 D_3 D_4} \to \frac{2}{t(s+t)}\oint
  \frac{dx_1dx_2 dx_3 dx_4}{(2 \pi i)^4} \frac{1}{D_1
    D_2 D_3 D_4}
=
    \left\{\begin{array}{cc}
         \frac{1}{4st} & \text{on }V_1\\
 -\frac{1}{4st} & \text{on }V_2\\
    \end{array}
\right.
\end{eqnarray}
and 
\begin{eqnarray}
  \label{spurious_contour}
  \int \frac{d^4 l_1}{i \pi^2} \frac{l\cdot \omega}{D_1 D_2 D_3 D_4} \to\frac{2}{t(s+t)}\oint
  \frac{dx_1dx_2 dx_3 dx_4}{(2 \pi i)^4} \frac{x_4}{D_1
    D_2 D_3 D_4}
=
    \left\{\begin{array}{cc}
         \frac{1}{8s} & \text{on }V_1\\
 \frac{1}{8s} & \text{on }V_2\\
    \end{array}
\right.
\end{eqnarray}
We have two weights $\omega^{(1)}$ and $\omega^{(2)}$. From the
contour integral of spurious term \eqref{spurious_contour},
\begin{equation}
  \label{eq:95}
  \omega^{(1)} \frac{1}{8 s}+\omega^{(2)} \frac{1}{8 s} =0 \,.
\end{equation}
Hence $\omega^{(2)}=-\omega^{(1)}$. Normalize the weights for the scalar box
integral, 
\begin{equation}
  \label{eq:13}
  \omega^{(1)}=2 st,\quad  \omega^{(2)}=-2 st\,.
\end{equation}

Hence
\begin{align}
  \label{eq:96}
  c_{box}&=2 st \cdot \frac{2}{t(s+t)} \big(\oint_{V_1} \frac{dx_1dx_2 dx_3
    dx_4}{(2\pi i)^4}  \frac{N}{D_1 D_2 D_3 D_4} -\oint_{V_2}\frac{dx_1dx_2 dx_3
    dx_4}{(2\pi i)^4}  \frac{N}{D_1 D_2 D_3 D_4} \big)\nn\\
&= \frac{1}{2} S^{(1)}+\frac{1}{2} S^{(2)}\,.
\end{align}
which is the same as \eqref{box_unitarity}.

 Two-loop maximal unitarity method was first invented in
\cite{Kosower:2011ty} for the $4D$ massless double box diagram, in an
elegant way of determining all contours and corresponding contour
weights. Afterwards, this method was generalized for  double box
diagram with external massive legs \cite{Johansson:2012sf, Johansson:2012zv, Johansson:2013sda, CaronHuot:2012ab}. 

In general, for multi-loop cases, the contour
integrals are multivariate, and can be complicated in some cases. There are
complicated issues with \eqref{contour_integral}.
\begin{enumerate}
  \item The solution set $V_i$ is not a rational variety. For example,
    $V_i$ can be an elliptic curve or a hyper-elliptic curve. Then 
    contour integrals are then not only residue computation, but also
    integrals over the
    fundamental cycles. Some of these cases are treated by maximal
    unitarity with complete elliptic integrals or hyper-ellitpic integrals
    \cite{Sogaard:2014jla,Georgoudis:2015hca}. There is a rich algebraic geometry
    structure in this direction and these integrals are important
    for LHC physics. But we are not going to cover this
    direction in these notes, since the background knowledge of algebraic
    curves need to be introduced. 
\item The residue is multivariate and Cauchy's formula does not work
  since the Jacobian at the pole is zero. For example, the $4D$ slashed
  box diagram and the $4D$ triple box diagram both have complicated
  multivariate residues.
We discuss this direction in the rest of this chapter.
\end{enumerate}

Note that, in a different context,
\cite{Henn:2013pwa,Henn:2014qga,Caron-Huot:2014lda,Remiddi:2016gno,Primo:2016ebd,Meyer:2016slj}
contour integrals like \eqref{contour_integral}
from Feynman integrals are also important for determining the {\it canonical}
MI basis, for which the differential equation makes simple.

\subsection{A multivariate residue example}

Consider $4D$ three-loop massless triple box diagram
(Figure. \ref{tribox}). There are $10$ inverse
propagators ,
\begin{align}
D_1 = {} & l_1^2\;, &
D_2 = {} & l_2^2\;, &
D_3 = {} & l_3^2\;, &
D_4 = {} & (l_1+k_1)^2\;, \nn \\
D_5 = {} & (l_1-k_2)^2\;, &
D_6 = {} & (l_2+k_3)^2\;, &
D_7 = {} & (l_2-k_4)^2\;, &
D_8 = {} & (l_3+k_1+k_2)^2\;, \nn \\
D_9 = {} & (l_1-l_3-k_2)^2\;, &
D_{10} = {} & (l_3-l_2-k_3)^2\;.
\end{align}
with $k_1^2=k_2^2=k_3^2=k_4^2=0$. 
\begin{figure}
  \centering
  \includegraphics[scale=1]{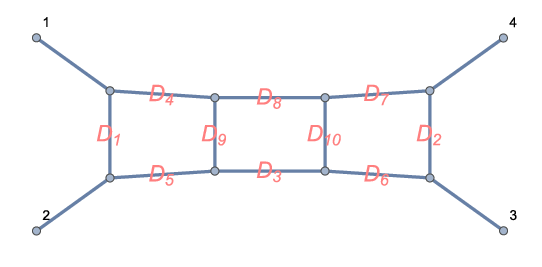}
  \caption{Three loop triple box diagram}
  \label{tribox}
\end{figure}
We parameterize loop momenta with spinor helicity formalism \cite{Dixon:1996wi}, 
\begin{align}
\ell_1(\alpha_1,\dots,\alpha_4) = {} &
\alpha_1 k_1+\alpha_2 k_2+
\alpha_3\frac{\avg{23}}{\avg{13}}1\tilde 2+
\alpha_4\frac{\avg{13}}{\avg{23}}2\tilde 1\;,
\nn \\[1mm]
\ell_2(\beta_1,\dots,\beta_4)\, = {} &
\beta_1 k_3+\beta_2 k_4+
\beta_3
\frac{\avg{14}}{\avg{13}} 3\tilde 4+
\beta_4
\frac{\avg{13}}{\avg{14}}
4\tilde 3\;.
\nn \\[1mm]
\ell_3(\gamma_1,\dots,\gamma_4)\, = {} &
\gamma_1 k_2+\gamma_2 k_3+
\gamma_3
\frac{\avg{34}}{\avg{24}}
2\tilde 3+
\gamma_4
\frac{\avg{24}}{\avg{34}}
3\tilde 2\;,
\end{align}
The cut solution for $D_1=D_2=\ldots D_{10}=0$ can be found by primary
decomposition \cite{Badger:2012dv,Sogaard:2013fpa},
\begin{equation}
  I=I_1\cap \ldots \cap I_{14}.
\end{equation}
There are $14$ independent solutions, each of which can be
parameterized rationally. For example, on $V_1=\mathcal Z(I_1)$ the
triple box Feynman integral with numerator $N(l_1,l_2,l_3)$
becomes a contour integral,
\begin{align}
 \frac{1}{(2\pi i)^2t^2s^{8}}\oint
\frac{dz_1\wedge dz_2  N(l_1,l_2,l_3)|_{V_1}}{(1+z_1)(1+z_2)(1+z_1-\chi z_2)}\;.
\end{align}
where the denominators come from the Jacobian of evaluating holomorphic
delat functions in $D_1,\ldots D_{10}$. $z_1$, $z_2$ are free
variables parametrizing this solution. The difficulty is that on this
cut solution,  loop momenta $l_i$ are not polynomials in $z_1$ and
$z_2$, but rational functions in $z_1$ and $z_2$ \cite{Sogaard:2013fpa}. Hence we get contour integrals like,
\begin{equation}
\label{xbox_residue}
  \frac{1}{(2\pi i)^2}\oint
\frac{dz_1\wedge dz_2 P(z_1,z_2)}{(1+z_1)(1+z_2)(1+z_1-\frac{t}{s} z_2)z_2}\;.
\end{equation}
where $P(z_1,z_2)$ is a polynomial in $z_1$ and $z_2$. $(z_1,z_2)\to
(-1,0)$ is a multivariate residue. Note that at this point, $3$
factors in denominators vanish,
\begin{equation}
  1+z_1,\quad z_2 \quad 1+z_1-\frac{t}{s} z_2\,.
\end{equation}
Hence, the Jacobian of denominators must be vanishing at $(-1,0)$, so
the residue cannot be calculated by inverse Jacobian (Cauchy's
theorem). Note that we cannot directly use polynomial division to
simplify the integrand, since $I=\la 1+z_1 , z_2, 1+z_1-\frac{t}{s}
z_2\ra =\la 1+z_1,z_2\ra\not=\la 1\ra$. So Hilbert's
weak Nullstellensatz (Theorem \ref{weak_Nullstellensatz}) cannot be
used here to reduce the number of
denominators. 

Difficult multivariate residues also arises from the maximal cut of integrals
with doubled propagators from two-loop integrals. In the rest of this
chapter, we use algebraic geometry techniques to compute these
residues efficiently.


\section{Basic facts of several complex variables}
\subsection{Multivariate holomorphic functions}
We first review some properties of several complex variables \cite{MR1288523,MR1045639,MR2176976}.
\begin{definition}
  Complex variables for $\C^n$ are $z_i=x_i+i y_i$ and the basis for the tangent space is
  \begin{equation}
    \frac{\partial}{\partial
      z_i}=\frac{1}{2}\big(\frac{\partial}{\partial x_i}-i
    \frac{\partial}{\partial y_i}\big),\quad
\frac{\partial}{\partial
      \bz_i}=\frac{1}{2}\big(\frac{\partial}{\partial x_i}+i \frac{\partial}{\partial y_i}\big).
  \end{equation}
For a point $\xi=\xis$ in $\C^n$, the (open) polydisc with radius $r$ is
\begin{equation}
  \Delta(\xi,r)=\{\zs\big||z_i-\xi_i|<r, \ i=1,\ldots n\}
\end{equation}
\end{definition}
\begin{definition}
  A differentiable function $f$ on $U$, an open set of $\C^n$, is {\it
    holomorphic } if,
  \begin{equation}
    \frac{\partial f}{\partial
      \bz_i}=0,\quad i=1,\ldots n.
  \end{equation}
\end{definition}
\begin{thm}[Cauchy's formula]
Let $f$ a function holomorphic in $\Delta(\xi,r)$ and continuous on
$\bar \Delta(\xi,r)$. Then for $z\in \Delta(\xi,r)$,
\begin{equation}
  f(z)=\frac{1}{(2\pi i)^n}\int_{|w_1-\xi_1|=r}\ldots \int_{|w_n-\xi_n|=r} \frac{ f(w_1,\ldots,w_n)dw_1
    \ldots dw_n}{(w_1-z_1)\ldots (w_n-z_n)}.
\end{equation}
\end{thm}
\begin{proof}
  Apply one-variable Cauchy's formula $n$ times \cite{MR1045639}. 
\end{proof}
From the Taylor expansion of $1/(w_i-z_i)$ in $(z_i-\xi_i)$, $f(z)$ has a multivariate Taylor expansion in  $\bar
\Delta(\xi,r)$. Hence like univariate case, a holomorphic function is
an analytic function. Similarly, for two holomorphic functions $f$ and $g$
on a connected open set $U\subset \C^n$, if $f=g$ on an open subset of
$U$ then $f=g$ on $U$. 

However, the pole structure of a multivariate function is very
different from that in the univariate case.  

\begin{thm}[Hartog's extension]
\label{Hartog}
  Let $U$ be an open set of $\C^n$, $n>1$. $K$ is a compact subset of
  $U$ and $U-K$ is connected. Then any holomorphic function on $U-K$
  extends to a holomorphic function of $U$.
 \end{thm}
 \begin{proof}
   See H\"ormander \cite[Chapter 2]{MR1045639}.
 \end{proof}

 \begin{example}
   Consider $n=2$, $U$ is the polydisc $\Delta(0,r)$ and
   $K=O=\{(0,0)\}$ in Theorem \ref{Hartog}. Suppose that $f(z_1,z_2)$
     is holomorphic in $U-K$. Define the function
     \begin{equation}
       g(z_1,z_2)=\frac{1}{2 \pi i}\int_{|w_2|=r'} \frac{f(z_1,w_2)dw_2}{w_2-z_2}\,,
     \end{equation}
where $0<r'<r$. Clearly $g$ is well defined in the smaller polydisc
$\Delta(0,r')$. $g$ is holomorphic in both $z_1$
and $z_2$. If $z_1\not =0$, then by one-variable Cauchy's formula,
$g(z_1,z_2)=f(z_1,z_2)$. $f=g$ in $\Delta(0,r')\cap
\{z|z_1\not=0\}$, hence $f=g$ in $\Delta(0,r')-O$.
Define a new function
\begin{equation}
  F(z)=\left\{
    \begin{array}{cc}
      g(z) &  z\in \Delta(0,r')\\
      f(z) & z\not \in \Delta(0,r')\ \text{but } z \in \Delta(0,r)
    \end{array}
\right. \,
\end{equation}
Clearly $F$ is holomorphic in $U$ and $F=f$ in $U-K$, so $F$ is the extension.
 \end{example}

Hartog's extension means the pole of a multivariate holomorphic $f$
has complicated structure, say, cannot be a point. It also implies
that we should not study the space of holomorphic functions on an open
set like $U-K$ in Theorem \ref{Hartog}, since these functions can
always be extended. 

Laurent expansion of a multivariate holomorphic function is also
subtle. 
\begin{definition}
  A subset $\Omega$ of $\C^n$ is called Reinhardt domain, if $\Omega$
  is open, connected and for any $\zs\in \Omega$, $(e^{i\theta_1}
  z_1,\ldots , e^{i\theta_n}
  z_n)\in \Omega$, $\forall \theta_1 \in \mathbb R,\ldots, \theta_n
  \in \mathbb R$. This is a generalization of an annulus on complex plane.
\end{definition}

\begin{proposition}
  Let $f$ be a holomorphic function on a Reinhardt domain
  $\Omega$. Then there exists a Laurent series, 
  \begin{equation}
    \sum_{(\alpha_1,\ldots,\alpha_n) \in Z^n}  c_{\alpha_1\ldots \alpha_n}z_1^{a_1} \ldots z_n^{a_n}\,,
  \end{equation}
which is uniformly convergent to $f$ on any compact subset of $\Omega$.
\end{proposition}
\begin{proof}
See Scheidemann \cite{MR2176976}.
\end{proof}

A multivariate function $f$ may be defined over a domain which is not a Reinhardt domain. For a simple example, the function
$f=1/(z_1-z_2)$ is defined on $U=\{(z_1,z_2)|z_1\not =z_2,\ (z_1,z_2)\in
\C^2\}$, where $U$ is not a Reinhardt domain. It is hard to define
Laurent series for $f$ in $U$. Instead, consider $\Omega=\{(z_1,z_2)||z_1|>|z_2|,\ (z_1,z_2)\in
\C^2\}$. Then $\Omega$ is a Reinhardt domain, and on $\Omega$,
\begin{equation}
  \frac{1}{z_1-z_2}=\sum_{n=0}^\infty z_2^n z_1^{-n-1}, \quad
  (z_1,z_2)\in \Omega
\end{equation}
This Laurent series does not converge outside $\Omega$.

We turn to complex manifolds. 
\begin{definition}
  A complex manifold $M$ is a differentiable manifold, with an open
  cover $\{U_\alpha\}$ and coordinate maps $\phi_\alpha: U_\alpha \to
   \C^n$, such that 
   all $\phi_\alpha\phi_\beta^{-1}$'s components are holomorphic on 
   $\phi_\beta(U_\alpha \cap U_\beta)$ for 
   $U_\alpha \cap U_\beta \not = \emptyset$.
\end{definition}

\begin{example}[Complex projective space]
\label{CPn}
  Define $\CP^n$ as the quotient space $\C^{n+1}-\{0,\ldots 0\}$ over 
  \begin{equation}
    (Z_0,\ldots ,Z_{n})\sim  (\lambda Z_0,\ldots \lambda Z_{n}),\quad
    \lambda \in\C^*
\label{rescaling}
  \end{equation}
The equivalence class of $(Z_0,\ldots,Z_n)$ in $\CP^n$ is denoted as
$[Z_0, \ldots, Z_n]$, which is the {\it homogeneous
  coordinate}. Define an open cover of $\CP^n$,
$\{U_0,U_1,\ldots,U_n\}$, where,
\begin{equation}
  U_i=\{[Z_0,\ldots,Z_n]|Z_i\not=0\},\quad i=0,\ldots, n
\end{equation}
For each $U_i$, the coordinate map $\phi_i: \ U_i \to \C^n$ is
\begin{equation}
  \phi_i([Z_0,\ldots,Z_n])=(\frac{Z_0}{Z_i},\ldots,
  \widehat{\frac{Z_i}{Z_i}},\ldots, \frac{Z_n}{Z_i})\equiv
  (z_0^{(i)},\ldots,\widehat{z_i^{(i)}}\ldots z_n^{(i)}).
\end{equation}
Hence, for $i<j$, 
\begin{equation}
\label{projective_atlas}
  \phi_i \phi_j^{-1}(z_0^{(j)},\ldots,\widehat{z_j^{(j)}}\ldots
  z_n^{(j)})=(\frac{z_0^{(j)}}{z_i^{(j)}},\ldots,\widehat{\frac{z_i^{(j)}}{z_i^{(j)}}}\ldots,\frac{1}{z_i^{(j)}},\ldots
  \frac{z_n^{(j)}}{z_i^{(j)}}),
\end{equation}
Since on $\phi_j (U_i\cap U_j)$, $z_i^{(j)}\not =0$, the transformation
\eqref{projective_atlas} is holomorphic. Hence $\CP^n$ is a compact complex
space.  In particular, we may identify $U_0$ as  $\C^n$.

For a homogeneous polynomial $F(Z_0,\ldots Z_n)$, the equation $F(Z_0,\ldots
Z_n)=0$ is well defined, since the rescaling \eqref{rescaling}
does not affects the value $0$.
\end{example}

Like real manifold case, we can also study sub-manifolds of a
manifold. In particular, the codimension-1 case is very important for
our discussion of residues in this chapter. 
\begin{definition}
  An analytic hypersurface $V$ of a complex manifold $M$ is a subset
  of $M$ such that $\forall p\in V$, there exists a neighborhood of $p$ in
  $M$, such that locally $V$ is the set of zeros of a
  holomorphic function $f$, defined in this neighborhood. 
\end{definition}
Like the algebraic variety case (Theorem \ref{variety_decomposition}),
any analytic hypersurface uniquely decomposes as the union of irreducible
analytic hypersurfaces. \cite{MR1288523}.
\begin{definition}
  For a complex manifold $M$, a divisor $D$ is a locally finite formal linear
  combination,
  \begin{equation}
    D=\sum_i c_i V_i,
  \end{equation}
where each $V_i$ is an irreducible analytic hypersurface in $M$. 
\end{definition}

\subsection{Multivariate residues}
Recall the for one-variable case, the residue of a meromorphic
function $h(z)/f(z)$ at the point $\xi$, is defined as,
\begin{equation}
  \Res{}_\xi\bigg(\frac{h(z)}{f(z)}dz\bigg)=\frac{1}{2\pi i}\oint_{|z-\xi|=\epsilon} \frac{h(z)dz}{f(z)}\,.
\end{equation}
where $f$ and $h$ are holomorphic near $\xi$. 

To define a multivariate residue in $\C^n$, we need $n$ vanishing
denominators $f_1, \ldots f_n$ such that $f_1(z)=\ldots =f_n(z)=0$
defines isolated points. 

\begin{definition}
\label{Grothendieck_residue}
   Let $U$ be a ball in $\C^n$ centered at $\xi$,
   i.e. $||z-\xi||<\epsilon$ for $z\in U$. $f_1(z), \dots, f_n(z)$ are
   holomorphic function in $U$, and have only one isolated
common zero, $\xi$ in $U$. Let $h(z)$ be a holomorphic function in
a neighborhood of $\bar U$. Then for the differential form,
\begin{equation}
  \label{omega}
  \omega=\frac{h(z) dz_1 \wedge \cdots \wedge dz_n}{f_1(z) \cdots f_n(z)}\;,
\end{equation}
the (Grothendieck) residue at $\xi$ is defined to be \cite{MR1288523} ,
\begin{equation}
    \label{local_residue}
    \Res{}_{\{f_1,
  \dots, f_n\},\xi}(\omega)=\bigg(\frac{1}{2\pi i}\bigg)^n\oint_{\Gamma}
\frac{h(z) dz_1 \wedge \cdots \wedge dz_n}{f_1(z) \cdots f_n(z)}\;,
  \end{equation}
where the contour $\Gamma$ is define by the real $n$-cycle
$\Gamma=\{z: z\in $U$, ~|f_i(z)|=\epsilon\}$ with the orientation specified by
$d(\arg f_1) \wedge \cdots \wedge d(\arg f_n)$. 
\end{definition}
Note that $\Gamma$ in this definition ensures that $f_i$'s are
nonzero for this contour integral. A naive contour choice
$\Gamma'=\{z: z\in $U$, |z_i-\xi_i|=\epsilon,\ \forall i\}$ in general does not
work. For instance,
\begin{equation}
  \frac{1}{(2\pi i)^2} \oint_{\Gamma'} \frac{dz_1\wedge dz_2}{(z_1+z_2) (z_1-z_2)},
\end{equation}
with $\Gamma'=\{z: z\in $U$, |z_1|=\epsilon,|z_2|=\epsilon 
\}$ is ill-defined. On this contour, both $(z_1+z_2)$ and $(z_1-z_2)$
has zeros. 

Note that if we permute functions $\{f_1 ,\ldots , f_n\}$, the
differential form is invariant but the contour orientation will be
reversed if the permutation is odd. This is a new feature of
multivariate residue, hence in Definition \ref{Grothendieck_residue}, we
keep $\{f_1,\ldots, f_n\}$ in the subscript.

Clearly, if $f_1(z)=f_1(z_1),\ldots, f_n(z)=f_n(z_n)$, then,
\begin{equation}
\label{factorable_residue}
   \Res{}_{\{f_1,
  \dots, f_n\},\xi}(\omega)=\bigg(\frac{1}{2\pi i}\bigg)^n\oint
\frac{dz_1}{f_{z_1}} \oint\frac{dz_2}{f_{z_2}} \ldots \oint \frac{h(z)dz_n}{f_{z_n}}  \,,
\end{equation}
the multivariate residue becomes iterated univariate residues.

\begin{definition}
  In Definition \ref{Grothendieck_residue}, if the Jacobian of
$f_1,
  \ldots, f_n$ in $z_1,\ldots z_n$ at $\xi$ is nonzero, we call this
  residue non-degenerate. Otherwise it is called degenerate.
\end{definition}

\begin{proposition}[Cauchy]
If the residue in Definition \ref{Grothendieck_residue} is non-degenerate, then
  \begin{equation}
    \label{non_degenerate_residue}
    \Res{}_{\{f_1,
  \dots, f_n\},\xi}(\omega)=\frac{h(\xi)}{J(\xi)}\;.
  \end{equation}
where $J(\xi)$ is the Jacobian of
$f_1,
  \ldots, f_n$ in $z_1,\ldots z_n$ at $\xi$ .
\end{proposition}
\begin{proof}
  In this case, we can use implicit function theorem to treat $f_i$'s
  as coordinates and compute the residue directly \cite{MR1288523}.
\end{proof}

\begin{proposition}
\label{vanishing_residue}
If $h$ in Definition \ref{Grothendieck_residue} satisfies,
\begin{equation}
  h(z)=q_1(z) f_1(z) + \ldots + q_n(z) f_n(z), \quad z\in U
\end{equation}
where $q_i$'s are holomorphic functions in $U$. Then 
\begin{equation}
    \Res{}_{\{f_1,
  \dots, f_n\},\xi}(\omega)=0\;.
  \end{equation}
\end{proposition}
\begin{proof}
  This is from Stokes' theorem \cite{MR1288523}.
\end{proof}

In general a multivariate residue is not of the form like 
\eqref{factorable_residue} or non-degenerate. Unlike the univariate
case, Laurent expansion, even if it is defined in a subset, in general
does not help the evaluation of multivariate residues. Hence we need a
sophisticated method to compute residues, like \eqref{xbox_residue}.

\begin{thm}[Global residue]
   Let $M$ be a compact complex manifold. $D_1 \ldots D_n$ are
   divisors of $M$, such that $D_1\cap \ldots \cap D_n=S$ is a finite set. If
   $\omega$ is a holomorphic $n$-form defined in $M-D_1\cup \ldots
   \cup D_n$ whose polar divisor is $D=D_1+\ldots D_n$, then
   \begin{equation}
     \sum_{\xi \in S}\Res{}_{\{D_1,
  \dots, D_n\},\xi}(\omega) =0 .
   \end{equation}
\end{thm}
\begin{proof}
  This is from Stokes' theorem of a complex manifold. See Griffiths and Harris \cite{MR1288523}.
\end{proof}
Note that to consider global residue theorem, we need a compact
complex manifold but $\C^n$ is not. So residues on a complex manifold
has to be defined. $\omega$ has the polar divisor  $D=D_1+\ldots D_n$
means, near a point $\xi\in S$, locally each $D_i$ is a divisor of a
holomorphic function $f_i$ and $\omega$ has the local form
\eqref{omega}. Again, the script $\{D_1,
  \dots, D_n\}$ indicates the ordering of denominators. 

  \begin{example}
    Consider the meromorphic differential form in $\C^n$,
    \begin{equation}
      \omega=\frac{dz_1 \wedge dz_2}{(z_1+z_2)(z_1-z_2)}\,.
    \end{equation}
Extend $\omega$ to a  meromorphic differential form  in $\CP^2$
(Example \ref{CPn}). Let $[Z_0,Z_1,Z_2]$ be the homogeneous
coordinate. In the patch $U_0$, define $z_1=Z_1/Z_0$,
$z_2=Z_2/Z_0$. For the patch $U_1$, let $u_0=Z_0/Z_1$,
$u_2=Z_2/Z_1$. Then on $U_0\cap U_1$,
\begin{equation}
  z_1=\frac{1}{u_0},\quad z_2=\frac{u_2}{u_0}\,.
\end{equation}
After the change of variables, on $U_0\cap U_1$,
\begin{equation}
      \omega=\frac{du_0 \wedge du_2}{u_0(u_2-1)(u_2+1)}\,,
    \end{equation}
Similarly, For the patch $U_2$, let $v_0=Z_0/Z_2$,
$v_1=Z_1/Z_2$. On $U_0\cap U_2$,
\begin{equation}
      \omega=\frac{dv_0 \wedge dv_1}{v_0(v_1-1)(v_1+1)}\,.
    \end{equation}
Then in $\CP^2$, $\omega$ is defined except on $3$ irreducible
hypersurfaces $V_1=\{Z_0=0\}$, $V_2=\{Z_1+Z_2=0\}$ and
$V_3=\{Z_1-Z_2=0\}$. To apply global residue theorem, consider
\begin{equation}
  D_1=V_0+V_1,\quad D_2=V_2 \,.
\end{equation}
Then $D=D_1+D_2$ is the polar divisor of $\omega$. $D_1\cap
D_2=\{P_1,P_2\}$, where $P_1=[1,0,0]$ and $P_2=[0,1,1]$. 
Global residue theorem reads,
\begin{equation}
\Res{}_{\{D_1, D_2\},P_1} (\omega)+\Res{}_{\{D_1, D_2\},P_2} (\omega)=0 \,.
\end{equation}
Explicitly by
\eqref{non_degenerate_residue},
\begin{equation}
  \Res{}_{\{D_1, D_2\},P_1} (\omega)=-\half,\quad \Res{}_{\{D_1, D_2\},P_2} (\omega)=\half \,.
\end{equation}
Note that if we consider a different
set of divisors, say, $D_1'=V_1$,\quad $D_2'=V_0+V_2$, then $D_1'\cap
D_2'=\{P_1,P_3\}$, where $P_3=[0,1,-1]$. So there is another relation $\Res{}_{\{D_1', D_2'\},P_1}
(\omega)+\Res{}_{\{D_1', D_2'\},P_3} (\omega)=0$, and,
\begin{equation}
  \Res{}_{\{D_1', D_2'\},P_3} (\omega)=\half.
\end{equation}
We see that for a multivariate case, there can be several global
residues relations for one meromorphic form.
\end{example}

 \section{Multivariate residues via computational algebraic geometry }
There are several algorithms of calculating multivariate
residues in algebraic geometry. We mainly use two methods,
{\it transformation law} and {\it Bezoutian}.

\subsection{Transformation law}

\begin{thm}
  For the residue in Definition \ref{Grothendieck_residue}, $g_i=\sum_j a_{ij} f_j$, where
  $a_{ij}$ are locally holomorphic functions near $\xi$. Then 
\begin{equation}
    \label{transformation}
    \Res{}_{\{f_1,
  \dots, f_n\},\xi}(\omega)= \Res{}_{\{g_1,
  \dots, g_n\},\xi}(\det A~\omega)
  \end{equation}
where $A$ is the matrix $(a_{ij})$.
\end{thm}
\begin{proof}
  See Griffiths and Harris \cite{MR1288523}.
\end{proof}
Note that this is a transformation of denominators, not the complex
variables. In particular, if $f_1,\ldots f_n$
are polynomials, we can calculate Gr\"obner basis for $I=\langle
f_1,\ldots f_n\rangle$ in \lex\  to get a set of polynomial $g_i$'s, such that each $g_i$ is
univariate. $g_i(z)=g_i(z_i)$ (Theorem \ref{Elimination}). Then the r.h.s of \eqref{transformation}
can be calculated as  univariate residues.
\begin{example} Consider the residue of
  \begin{equation}
    \omega=\frac{dx \wedge dy }{f_1 f_2},\,
  \end{equation}
at $(0,0)$, where $f_1=a y^3 + x^2 + y^2$, $f_2=x^3 + x y - y^2$. This
is a degenerate residue. By
\GB\ computations,
\begin{equation}
  A=\left(
\begin{array}{cc}
 -\frac{2 a x^2+a x-a y x-a y+1}{a^2} & \frac{a x^4-a y x^2+a y^2 x-x+a y^2-y}{a} \\
 \frac{a^2 y^5-2 a y^3-a x^2 y^2+x y+y+x^2}{a^3} & \frac{a x y^2-y-x}{a^3} \\
\end{array}
\right),
\end{equation}
and,
\begin{equation}
  \{g_1,g_2\}=\big\{\frac{x^2 (a^2 x^5-3 a x^2-a x-1)}{a^2},\frac{y^3
    (a^3 y^5-2 a^2 y^3+a y+1)}{a^3}\big\}. \,
\end{equation}
Note that $g_1$, $g_2$ are univariate polynomials. Hence by
\eqref{transformation},
\begin{equation}
  \Res{}_{\{f_1,f_2\},(0,0)}(\omega)=a(1-a)\,.
\end{equation}

\end{example}

\begin{example}
Consider the $4D$ triple box's maximal cut \eqref{xbox_residue}, near
$z_1=-1$ and $z_2=0$,
  \begin{equation}
    \omega=\frac{dz_1\wedge dz_2 P(z_1,z_2)}{(1+z_1)(1+z_2)(1+z_1-\frac{t}{s}
  z_2)z_2}\;.
\end{equation}
Define $V_1=\{1+z_1=0\}$, $V_2=\{z_2=0\}$ and $V_3=\{1+z_1-\chi z_2\}$, which
  are irreducible hypersurfaces. So locally the polar divisor of
  $\omega$ is,
  \begin{equation}
    D=V_1+V_2+V_3.
  \end{equation}
To define multivariate residues, we may consider two divisors
$D_1=V_1+V_2$ and $D_2=V_3$. This corresponds to the denominator
definitions, $f_1=(1+z_1)z_2$ and $f_2=(1+z_1-t/s z_2)$. Using
\eqref{transformation} to change denominators, we find that,
for example if $P=1$,
\begin{equation}
  \Res{}_{\{f_1,f_2\},(0,0)}(\omega)=s/t.
\end{equation}
Note that there are different ways to define the divisors for
$\omega$, for instance, $D_1'=V_1+V_3$ and $D_2'=V_2$,
i.e. $f_1'=(1+z_1)(1+z_1-\chi z_2)$ and $f_2'=z_2$. Multivariate
residue dependence on the definition of divisors, for example if $P=1$,
 \begin{equation}
  \Res{}_{\{f_1',f_2'\},(0,0)}(\omega)=0\not=\Res{}_{\{f_1,f_2\},(0,0)}(\omega).
\end{equation}
Hence we need to consider all possible divisor definitions.

We calculated all $64$ residues from the maximal
unitarity cut of a three-loop triple box diagram
\cite{Sogaard:2013fpa}, by Cauchy's theorem and transformation
law. Then the contours weights are determined by spurious integrals
and IBPs. We used contour weights to derive the triple box master
integrals part of $4$-gluon $3$-loop pure-Yang-Mills amplitude, which agrees
with that from integrand reduction method \cite{Badger:2012dv}.

For integral with doubled propagators, we can also use transformation
law to compute residues for contour integrals \cite{Sogaard:2014ila,Sogaard:2014oka}.
\end{example}

\begin{remark}\ 

  \begin{enumerate}
    \item Usually, \GB\ computation in \lex\ is heavy. It is
      better to first compute \GB\ in \grevlex\ order, $G(I)=\{F_1,\ldots
      F_k\}$ and find the relations $F_i=b_{ij} f_j$. Then compute \GB
      in a block order to get  univariate polynomials
      $g_i(z_i)$. Divide $g_i(z_i)$ towards $G(I)$ and use $b_{ij}$'s,
      we get the matrix $A$.
      \item This method also works if $f_1,\ldots f_n$ are holomorphic functions
        but not polynomials. Replace $f_i$'s by their Taylor series,
        we can apply Gr\"obner basis method.  
  \end{enumerate}
\end{remark}

\subsection{Bezoutian} 
Multivariate residue computation via \GB, may be quite heavy since the transformation matrix
$A$ may contain high-degree polynomials. Bezoutian method provide a
different approach. 

\begin{definition}
  With the convention of Definition \ref{Grothendieck_residue}, for $\xi\in \C^n$,
  define the local symmetric form, for locally holomorphic functions
  $N_1$ and $N_2$,
 \begin{gather}
    \langle N_1, N_2\rangle_\xi\equiv\Res{}_{\{f_1,
  \dots, f_n\},\xi}\bigg(\frac{N_1 N_2 \dnz}{f_1\ldots f_n}\bigg)\,,
  \end{gather}
If $f_1,\ldots f_n$, $N_1, N_2$ are globally holomorphic in $\C^n$ and $\mathcal
Z(\{f_1,\ldots f_n\})$ is a finite set, then the
global symmetric form is 
\begin{gather}
  \langle N_1, N_2\rangle\equiv \sum_{\xi\in \mathcal
Z(\{f_1,\ldots f_n\})} \Res{}_{\{f_1,
  \dots, f_n\},\xi}\bigg(\frac{N_1 N_2\dnz}{f_1\ldots f_n}\bigg)\,.
  \end{gather}
\end{definition}

For the rest of discussion, we assume $f_1,\ldots f_n$, $N_1, N_2$ are
polynomials. In the previous Chapter, we used the ring
$R=\mathbb F[x_1,\ldots,x_n]$ and ideals to study algebraic
varieties. Here to discuss local properties of a variety, we need the
concept of local ring.
\begin{definition}
\label{local_ring}
  Consider $R=\mathbb C[x_1,\ldots,x_n]$, for a point $\xi \in \C^n$,
  $R_\xi$ is the set of rational functions,
  \begin{equation}
    R_\xi\equiv\bigg\{\frac{f(z)}{g(z)}\big|g(\xi)\not=0,\quad f,g\in R\bigg\}\,.
  \end{equation}
For an ideal $I$ in $R$, we denote $I_\xi$ as the ideal in $R_\xi$
generated by $I$. If $\xi\in \mathcal Z(I)$, and $\dim_{\C}R_\xi/I_\xi<\infty$, we define the multiplicity of $I$ at $\xi$ as $\dim_{\C}R_\xi/I_\xi$. 
\end{definition}

Let $I$ be $\la f_1,\ldots f_n\ra$. From Proposition \ref{vanishing_residue}, it is clear that $\la ,\ra$
is defined in $R/I$ and $\la ,\ra_\xi$ is defined in $R_\xi/I_\xi$,
because any polynomial in the ideal $I$ or the localized ideal $I_\xi$
must provide zero residue.

\begin{thm}[Local and Global Dualities]
  Let $I=\la f_1,\ldots, f_n\ra$ be an ideal in $\C[x_1,\ldots x_n]$,
  $\mathcal Z(I)$ is a finite set. Then $\la ,\ra$ is a non-degenerate
  in $R/I$ and $\la ,\ra_\xi$ is non-degenerate
  in $R_\xi/I_\xi$.
\end{thm}
\begin{proof}
  See Griffiths and Harris \cite{MR1288523}, Dickenstein et al. \cite{Dickenstein:2010:SPE:1965470}.
\end{proof}
Non-degeneracy of $\la, \ra$ implies that given a linear basis $\{e_1,
\ldots e_k\}$ of $R/I$, there is a dual basis $\{\Delta_1,
\ldots \Delta_n\}$, such that,
\begin{equation}
  \label{eq:5}
  \langle e_i, \Delta_j\rangle =\delta_{ij}\,.
\end{equation}
If these two bases are explicitly found, then we can compute any $\la
N_1, N_2\ra$. In particular, the sum of residues (in affine space)
of $\omega=N\dnz /(f_1\ldots f_n)$ is obtained algebraically,
\begin{equation}
  \sum_{\xi\in \mathcal Z(I)}\Res{}_{\{f_1,\ldots f_n\},\xi}(\omega)=\la N, 1\ra=\la \sum_{i=1}^k c_i
  e_i, \sum_{j=1}^k\mu_j \Delta_j\ra=\sum_{i=1}^k c_i \mu_i\,,
\end{equation}
where in the second equality, we expand $N=\sum_i c_i e_i $ and
$1=\sum_i \mu_i \Delta_i$. $c_i$'s and $\Delta_i$'s are complex
numbers.

Explicitly, $\{e_i\}$'s are found by using \GB\ of $I$ in \grevlex,
$G(I)$. They are monomials which are not divisible by any leading term
in $G(I)$. The dual basis can be found via Bezoutian matrix \cite{Dickenstein:2010:SPE:1965470}. First, calculate
the Bezoutian matrix $B=(b_{ij})$,
\begin{gather}
  \label{eq:4}
  b_{ij}\equiv \frac{f_i(y_1,\ldots y_{j-1},z_j,\ldots,z_n)}{z_j-y_j}
-\frac{f_i(y_1,\ldots y_{j},z_{j+1},\ldots,z_n)}{z_j-y_j}\,,
\end{gather}
where $y_i$'s are auxiliary variables. Let $\tilde I$ be the ideal in
$\C[y_1,\ldots y_n]$ which is $I$ after the replacement $z_1 \to
y_1,\ldots, z_n \to y_n$.

Then we divide the determinant $\det B$ over the double copy of the
Gr\"obner bases, $G(I)\otimes G(\tilde I)$. The remainder can be expand as,
\begin{equation}
  \label{eq:6}
  \sum_{i=1}^k \Delta_i(y) e_i(z),
\end{equation}
here $\Delta_i(y)$'s, after the backwards replacement $y_1 \to
z_1,\ldots, y_n \to z_n$ become the elements of the dual basis
\cite{Dickenstein:2010:SPE:1965470}.

\begin{example}
  Let $f_1=z_1+9 z_2+14 z_3+6$, $f_2=11 z_2 z_1+12 z_3 z_1+3 z_1+4
  z_2+16 z_2 z_3+14 z_3$ and $f_3=2 z_1 z_2+15 z_1 z_3 z_2+5 z_3 z_2+8
  z_1 z_3$ be polynomials in $\C[z_1,z_2,z_3]$. Define
  \begin{equation}
    \omega=\frac{z_1^3 dz_1\wedge dz_2 \wedge dz_3}{f_1 f_2 f_3}\,.
  \end{equation}
The Bezoutian determinant in $z_1,z_2,z_3$ and auxiliary
variables $y_1, y_2, y_3$ is,
\begin{gather}
  \det B=-180 y_1^2 z_3+2520 y_1 z_3^2-1485 y_2 y_1 z_2-576 y_1
  z_2-1620 y_2 y_1 z_3\nn\\
+1620 y_1 z_2 z_3-408 y_1 z_3-207 y_2 z_2+612 y_2 z_3+2160 y_2 z_2 z_3+165 y_2 y_1^2+64 y_1^2-322 y_2 y_1\nn\\-128 y_1-115 y_2-3360 z_2 z_3^2-952 z_3^2+140 z_2+1372 z_2 z_3+700 z_3\,.
\end{gather}
Let $I=\la f_1,f_2,f_3\ra$. Divide $\det B$ towards $G(I)\otimes G(\tilde I)$, we get the basis
$\{e_i\}$,
\begin{equation}
  e_1=z_3^3,\quad e_2=z_2 z_3, \quad e_3=z_3^2, \quad e_4=z_2, \quad e_5=z_3,\quad e_6=1\,,
\end{equation}
and the dual basis $\{\Delta_i\}$,
\begin{gather}
  \Delta_1=\frac{141120}{23},\quad \Delta_2= 2 (-12420 z_2-22680
  z_3-\frac{203652}{23})\,,\nn\\\Delta_3=-22680 z_2-35280
  z_3-\frac{335832}{23}\,,\nn\\ \Delta_4= 2 (-22680 z_3^2-12420 z_2 z_3-5436 z_3+1872
  z_2+1278)\,,\nn\\ \Delta_5=-35280 z_3^2-22680 z_2 z_3-24528 z_3-5436
z_2-\frac{79884}{23}\,,
\nn\\ \Delta_6=\frac{141120 z_3^3}{23}-\frac{335832 z_3^2}{23}-\frac{203652 z_2 z_3}{23}-\frac{79884 z_3}{23}+1278 z_2+\frac{21282}{23}\,.
\end{gather}
From the dual basis, we find the linear relation,
\begin{gather}
  1=\frac{23}{141120} \Delta_1\,.
\end{gather}
By polynomial division, we find 
\begin{gather}
  z_1^3=\frac{1568}{11} e_1 + c_2 e_2 +\ldots c_6 e_6 \mod I\,.
\end{gather}
Hence the sum of residues,
\begin{align}
  \sum_{\xi\in \mathcal Z(I)}\Res{}_{\{f_1,f_2,f_3\},\xi}(\omega)&=\la z_1^3,
  1\ra\nn\\
&=\frac{23}{141120} \la\frac{1568}{11} e_1 + c_2 e_2 +\ldots + c_6
  e_6, \Delta_1\ra
=\frac{23}{990}\,.
\end{align}
Note that all points in $\mathcal Z(I)$ and all local residues are
irrational, but the sum is rational.

This example is from CHY formalism of scattering equation for
$6$-point tree amplitudes. In CHY formalism, scattering amplitude is
expressed as the sum of residues of CHY integrand. Here we calculate
the amplitude without solving scattering equations
\cite{Sogaard:2015dba}. See alternative algebraic approaches in \cite{Baadsgaard:2015voa,Baadsgaard:2015ifa,Huang:2015yka}.
\end{example}

\begin{remark}\
  
  \begin{enumerate}
  \item Note that by this method, we get the sum of residues (in affine space)
purely by \GB\, and matrix determinant 
computations. It is not needed to consider algebraic extension or 
explicit solutions of $f_1=\ldots = f_n=0$. 

\item The Bezoutian matrix is
just a $n\times n$ matrix, i.e., the size of matrix is independent of
the dimension $\dim_\C R/I$. Hence it is an efficient method for
computing the sum of residues. 
\item If $f_i$'s coefficients are parameters, this method proved
that the sum of residues is a rational function of these parameters. 
\item In some cases, the sum of residues can also be evaluated by
  global residue theorem (GRT). However, in general, there are many poles at
  infinity so the GRT computation can be messy. 
  \end{enumerate}
\end{remark}

We can also use Bezoutian matrix to find local residues. One approach
is {\it partition of unity } for an affine variety: For each
$\xi\in\mathcal Z(I)$, we can find a polynomial $s_\xi$ \cite{opac-b1094391}, such that,
\begin{gather}
  \label{eq:7}
  \sum_{ \xi\in \mathcal Z(I)} s_\xi=1\mod I, \quad s_\xi^2=s_\xi \mod I\nn\,,\\
 s_{\xi_i} s_{\xi_j}=0 \mod I,\quad \text{if } i\not =j\,.
\end{gather}
 Then the individual residue is extracted from the sum of residues,
 \begin{equation}
   \label{eq:8}
    \Res{}_{\{f_1,
  \dots, f_n\},\xi}(\omega) =\sum_{ u\in \mathcal Z(I)} \Res{}_{\{f_1,
  \dots, f_n\},u}(s_{\xi} \omega)\,,
 \end{equation}
where the r.h.s is again obtained by Bezoutian matrix computation \cite{Dickenstein:2010:SPE:1965470}.

For the implement of the transformation law and Bezoutian matrix
computation, we refer to the public package {\sc MultivariateResidues} \cite{Larsen:2017aqb}.

\section{Exercises}
  
\begin{ex}
  Consider the maximal unitarity cut of $D=2$ massless sunset diagram with
  $k_1^2=M^2$ and inverse propagators,
  \begin{equation}
    D_1=l_1^2,\quad D_2=l_2^2,\quad D_3=(l_1+l_2-k_1)^2\,.
  \end{equation}
  \begin{enumerate}
  \item Define an auxiliary vector $\omega$, $k_1\cdot \omega=0$,
    $\omega^2=-M^2$. Let $e_1=(k_1+\omega)/2$ and
    $e_2=k_1-\omega$. Parameterize loop momenta as,
    \begin{equation}
      l_1 =a_1 e_1+ a_2 e_2,\quad l_2=b_1 e_1+b_2 e_2\,.
    \end{equation}
Rewrite $D_i$'s as polynomials in $a_1,a_2,b_1,b_2$. Define $I=\la
D_1, D_2, D_3\ra$, use \Singular\ or \Macaulay\ to find independent
solutions via primary decomposition,
\begin{equation}
  I=I_1 \cap \ldots \cap I_m\,.
\end{equation}
\item Formally define,
  \begin{equation}
    I[s_1,s_2,s_3;N]=\int \frac{d^2 l_1}{(2\pi)^2}\frac{d^2
      l_2}{(2\pi)^2} \frac{N}{D_1^{s_1} D_2^{s_2} D_3^{s_3}}\,.
  \end{equation}
Consider the maximal cut of the scalar integral $I[1,1,1;1]$ on each
of the cut solutions
$\mathcal Z(I_i)$. From resulting contour integrals,
determine all the poles on maximal cut. How many of them are
redundant?
\item Denote independent poles as $\{P_1,\ldots P_k\}$ and denote
  $I[s_1,s_2,s_3;N]_{P_i}$ as the residue of its corresponding contour 
  integral at $P_i$. Compute $I[1,1,1;1]_{P_i}$ for all $P_i$.
\item Denote 
  \begin{equation}
    I[s_1,s_2,s_3;N]|_\text{cut}=\sum_i^k w_i  I[s_1,s_2,s_3;N]|_{P_i}\,,
  \end{equation}
where $w_i$'s are weights of contours. Require that
$I[1,1,1;N]|_\text{cut}=0$ for spurious terms $N$,
\begin{equation}
  l_1\cdot \omega,\quad  l_2\cdot \omega,\quad (l_1\cdot
  \omega)(l_2\cdot k_1),\quad l_2\cdot k_1 -l_1\cdot k_1,\quad (l_2\cdot k_1)^2 -(l_1\cdot k_1)^2\,.
\end{equation}
What are the linear constraints of $w_j$'s?
\item Determine the ratio, $
  I[2,1,1;1]|_\text{cut}/I[1,1,1;1]|_\text{cut}$. Derive the on-shell
  integral relation (by determining $c$)
  \begin{equation}
     I[2,1,1;1]= c I[1,1,1;1] +(\text{simpler integrals})\,.
  \end{equation}
Similarly, determine $c'$ in
\begin{equation}
  \label{eq:9}
     I[3,1,1;1]= c' I[1,1,1;1] +(\text{simpler integrals})\,.
\end{equation}
  \end{enumerate}
\end{ex}

\begin{ex}
  Consider the meromorphic form, 
  \begin{equation}
      \omega=\frac{z_1 dz_1 \wedge dz_2}{(z_1+z_2)(z_1-z_2+z_1 z_2)}\,.
\end{equation}
Extend $\omega$ to a meromorphic form in $\CP^2$. Find all residues of
$\omega$ in $\CP^2$ and verify global residue theorem
explicitly. 
\end{ex}

\begin{ex}
  Consider the meromorphic form, 
  \begin{equation}
     \omega=\frac{N(z_1,z_2)dz_1\wedge dz_2}{(z_1+a z_2)(z_1^3+z_2^2+b
       z_1 z_2)}\,.
  \end{equation}

  \begin{enumerate}
  \item Use transformation law and \GB\ computation in {\sc Maple} or
    \Macaulay2, to compute the residue at $(0,0)$ with $N(z_1,z_2)=1$
    and $N(z_1,z_2)=z_1$.
    \item Without computation, argue that if $N(z_1,z_2)=z_1^2$ then the
      residue at $(0,0)$ is zero by Proposition \ref{vanishing_residue}.
  \end{enumerate}
\end{ex}

\begin{ex}
  Consider the meromorphic form, 
  \begin{equation}
     \omega=\frac{N(z_1,z_2)dz_1\wedge dz_2}{(z_1+ z_2)(z_1-z_2)(z_1^2+z_2^2+z_1)}\,.
  \end{equation}
Define $f_1=(z_1+ z_2)$, $f_2=(z_1-z_2)$ and
$f_3=(z_1^2+z_2^2+z_1)$. Use the transformation law to compute,
\begin{equation}
  \Res{}_{\{f_1, f_2 f_3\},(0,0)}(\omega),\quad \Res{}_{\{f_1 f_2, f_3\},(0,0)}(\omega)\,.
\end{equation}
\end{ex}

\begin{ex}[Sum of residues]
   Consider $f_1=z_1^2+ z_1 z_2+a
       z_2$, $f_2=z_1^3+z_2^2+b z_1 z_2$ and $I=\la f_1,f_2\ra$.
       \begin{enumerate}
       \item Use \GB\ in \grevlex, determine the basis $\{e_i\}$ for
         $\C[z_1,z_2]/I$. 
         \item Use Bezoutain matrix, find the dual basis
           $\{\Delta_i\}$.
           \item Compute the sum of residues in $\C^n$ for
             \begin{equation}
               \omega=\frac{z_1 z_2^2 dz_1\wedge dz_2}{f_1 f_2}\,.
             \end{equation}
             \item Compute $\la e_i, e_j\ra$ for all elements in
               $\{e_i\}$. Define $s_{ij}=\la e_i, e_j\ra$ and check
               that $S=(s_{ij})$ is a symmetric non-degenerate
               matrix. 
       \end{enumerate}
  
\end{ex}

\chapter{Integration-by-parts Reduction and Syzygies} 
\label{cha:integr-parts-reduct}

Integration-by-parts (IBP) identities \cite{Tkachov:1981wb,Chetyrkin:1981qh} arise from
the vanishing integration of total derivatives. Combined with symmetry
relations, IBPs reduce integrals to master integrals (MIs), i.e., the
linearly independent integrals. 

An $L$-loop $D$-dimensional \footnote{In general, we need to consider
  IBP in $D$-dimension. Otherwise for a specific integer-valued $D$,
  IBP relations may contain non-vanishing boundary terms.} IBP in general has the form,
\begin{equation}
\int \frac{d^D l_1}{i\pi^{D/2}} \ldots \int \frac{d^D l_L}{i\pi^{D/2}}
\sum_{j=1}^L \frac{\partial}{\partial l_j^\mu}
\bigg(\frac{v_j^\mu \hspace{0.5mm} }{D_1^{a_1} \cdots D_k^{a_k}}\bigg)
\hspace{1mm}=\hspace{1mm} 0 \,, \label{eq:IBP_schematic}
\end{equation}
where  vectors components $v_j^\mu$'s are polynomials in the internal and
external momenta, the $D_k$'s denote inverse propagators, and
$a_i$'s are integers.

For many multi-loop scattering amplitudes, IBP reduction is
a necessary step. After using unitarity and integrand
reduction to obtain the integrand basis, we may carry out IBP
reduction to get the minimal basis of integrals. For
differential equations of Feynman integrals, after differentiating of
master integrals, we get a large number of integrals in general. Then IBP
reduction is required to convert them to a linear combination of MIs, so
that the differential equation system is closed \cite{Kotikov:1990kg,Kotikov:1991pm,Henn:2014qga}. 

Multi-loop IBP reduction in general is very difficult. The difficulty comes from the large
number of choices of $v_i^\mu$ in \eqref{eq:IBP_schematic} : there are many IBP relations and
integrals involved. After obtaining IBP relations, we need to
apply linear reduction to find the independent set of IBPs. This process
usually takes a lot of computing time and RAM. The current standard
IBP generating algorithm is {\rm Laporta} \cite{Laporta:2001dd,Laporta:2000dc}. There are several publicly available implementations
of automated IBP reduction: AIR~\cite{Anastasiou:2004vj},
FIRE~\cite{Smirnov:2008iw,Smirnov:2014hma},
Reduze~\cite{Studerus:2009ye,vonManteuffel:2012np},
LiteRed~\cite{Lee:2012cn}, Kira~\cite{Maierhoefer:2017hyi,
  Maierhofer:2018gpa} along with private implementations. IBP
computation can be sped up by using finite-field methods \cite{vonManteuffel:2014ixa,vonManteuffel:2016xki}.

One sophisticated way to improve the IBP generating efficiency is to pick up suitable
$v_i^\mu$'s such that (\ref{eq:IBP_schematic}) contains no doubled
propagator \cite{Gluza:2010ws}. Since from Feynman rules, usually we
only have integrals
without doubled
propagator. Hence if we can work with integrals without doubled
propagator during the whole IBP reduction procedure, the computation will be
significantly simplified. Specifically, when $a_i= 1$, $\forall i=1,\ldots
,k$ in (\ref{eq:IBP_schematic}), if 
\begin{equation}
\label{no_double_propagator}
  \sum_j \frac{\partial D_i}{\partial l_j^\mu} v_i^\mu =\beta_i
  D_i,\quad i=1\ldots k\,,
\end{equation}
where $\beta_i$ is a polynomial in loop momenta, then all double-propagator
integrals are removed from the IBP relation \eqref{eq:IBP_schematic}.

Note \eqref{no_double_propagator} appears to be a linear equation system for
$v_i^\mu$'s and $\beta_i$. However, 
$v_i^\mu$'s must be polynomials in loop momenta, otherwise the
doubled propagators reappear. If we solve \eqref{no_double_propagator}
by standard linear algebra method, then the solutions are in general rational
functions which do not help the IBP reduction. To
distinguish with linear equations, \eqref{no_double_propagator} is
called a
{\it syzygy} equation. It is not surprising that the form of
\eqref{no_double_propagator} is closely related to S-polynomials and
polynomial division (Definition \ref{S-polynomial} and Algorithm
\ref{multivariate_polynomial_division}), so syzygy can be solved by
\GB\ .

We believe that the syzygy approach \eqref{no_double_propagator} 
maximaizes its power, when combined
with Baikov representation \cite{Baikov:1996rk} and unitarity cuts. Baikov
representation linearizes inverse propagators $D_i$'s so the syzygy
equation becomes simpler. Furthermore, It is more efficient to compute IBPs with
unitarity cuts, in a divide-and-conquer fashion, than to get complete IBPs at once.

In this chapter, we first introduce Baikov representation and then
review syzygy and the geometric meaning of
\eqref{no_double_propagator}. We will see that it defines  {\it polynomial tangent
fields of a hypersurface}, or  formally {\it derivations} in
algebraic geometry.  Finally we sketch some recent IBP 
algorithm development based on computational algebraic geometry \cite{Ita:2015tya, Larsen:2015ped, Boehm:2017wjc,
  Boehm:2018fpv, Abreu:2018rcw}. On the other hand, see
\cite{Bitoun:2017nre} for the recent IBP reduction method based on
D-module theory (differential algebra).

\section{Baikov representation}
The basic idea of Baikov representation \cite{Baikov:1996rk} is to define inverse
propagators and irreducible scalar products (ISP), except
$\mu_{ij}$'s, as free variables. In this section, we give an intuitive
approach to derive the Baikov representation. Once you get familiar
with it, you can directly use the Baikov representation formula given in
\cite{Lee:2012cn}. 

For a simple example, consider $D=4-2\epsilon$ one-loop box diagram \eqref{box_D}. 
\begin{equation}
  I_\text{box}^D[N]=\int \frac{d^D l}{i \pi^{D/2}} \frac{N}{D_1
    D_2 D_3 D_4}\,.
\end{equation}
Use van Neerven-Vermaseren variables, there are $5$ variables
$x_1,x_2,x_3,x_4$ and $\mu_{11}$. Hence it is a $5$-variable
system. The solid angel of $(-2\epsilon)$ directions in this integral is
irrelevant, hence,
\begin{align}
  I_\text{box}^D[N]&=\frac{1}{i \pi^{D/2}} \int d^{-2\epsilon}l^\perp
                       \int d^4 l^{[4]}\frac{N}{D_1
    D_2 D_3 D_4}\nn\\
&= \frac{1}{i \pi^{D/2}} \frac{ \pi ^{\frac{D-4}{2}}}{\Gamma
    (\frac{D-4}{2})} \int_0^\infty \mu_{11}^{\frac{D-6}{2}} d\mu_{11}  \int d^4 l^{[4]}\frac{N}{D_1
    D_2 D_3 D_4}\nn\\
&= \frac{1}{i\pi^2 \Gamma
    (\frac{D-4}{2})} \frac{2}{s(t+s)}\int_0^\infty \mu_{11}^{\frac{D-6}{2}} d\mu_{11}
    \int dx_1 dx_2 dx_3 dx_4\frac{N}{D_1
    D_2 D_3 D_4} \,,
\end{align}
where the factor $2/(s(t+s))$ is the Jacobian of changing variables
$l^{[4]} \to x_1,\ldots , x_4$.  Note that in this form, dimension shift
identities \eqref{box_dimension_shift} are manifest.

Since in this case ISPs are $x_4$ and
$\mu_{11}$, we define Baikov variables $z_1,\ldots z_5$ as,
\begin{equation}
\label{box_baikov_Direct}
  z_1 \equiv D_1,\quad  z_2 \equiv D_2,\quad  z_3 \equiv D_3,\quad  z_4 \equiv D_4,\quad
  z_5 \equiv l_1 \cdot \omega\,,
\end{equation}
Note that the Jacobian
\begin{equation}
  \frac{\partial(z_1,z_2,z_3,z_4,z_5)}{\partial(x_1,x_2,x_3,x_4,\mu_{11})}=-8
\end{equation}
is a constant. This is not surprising since by \eqref{box_D_RSP},
\begin{eqnarray}
x_1 = \half (z_1-z_2),\quad x_2 =\half (z_2-z_3)+\frac{s}{2},\quad 
 x_3 = \half (z_4-z_1),
\end{eqnarray}
$z_5=x_4$ and $D_1$ is linear in $\mu_{11}$. The inverse map,
$(z_1,z_2,z_3,z_4,z_5)\mapsto (x_1,x_2,x_3,x_4,\mu_{11})$ uniquely exists
and has polynomial form,
\begin{align}
  \mu_{11}&=\frac{1}{4 s t (s+t)}\big( s^2 t^2-2 s^2 t z_2-2 s^2 t
  z_4+s^2 z_2^2+s^2 z_4^2-4 s^2 z_5^2-2 s^2 z_2 z_4 \nn\\-2 s t^2 z_1
&-2 s t^2 z_3+2 s t z_1 z_2-4 s t z_1 z_3+2 s t z_2 z_3+2 s t z_1
  z_4-4 s t z_2 z_4+2 s t z_3 z_4+t^2 z_1^2 \nn\\ & +t^2 z_3^2-2 t^2
                                                    z_1 z_3\big)
                                                    \equiv
                                                    F(z_1,z_2,z_3,z_4,z_5) \,,
\end{align}
Then, we get the Baikov representation,
\begin{equation}
  I_\text{box}[N]=\frac{1}{i\pi^2 \Gamma
    (\frac{D-4}{2})} \frac{1}{4s(t+s)}
    \int_\Omega dz_1 dz_2 dz_3 dz_4 dz_5 F(z_1,z_2,z_3,z_4,z_5)^{\frac{D-6}{2}} \frac{N}{z_1
    z_2 z_3 z_4}.
\end{equation}
where $F(z_1,z_2,z_3,z_4,z_5)$ is called Baikov polynomial. $N$ is a
polynomial of $z_1,\ldots z_5$. The integral region $\Omega$ is
defined by $F(z_1,z_2,z_3,z_4,z_5) \geq 0$. In general, the integral
region of Baikov representation is complicated. However, for the
purpose of deriving IBPs, the region is not important. 

In practice, after OPP integrand reduction
\cite{Ossola:2006us,Ossola:2007ax}, $N$ is a polynomial of $\mu_{11}$
and at most linear in $(l\cdot \omega)$
\eqref{box_integrand_basis_D}. The terms with $\mu_{11}$ lead to scalar
integrals in higher dimension \eqref{box_dimension_shift}, while terms linearly in  $(l\cdot
\omega)$ are spurious. Hence we assume that $N$ is independent of
$(l\cdot \omega)$ and $\mu_{11}$. That implies that we can {\it integrate out} $\omega$
direction.

Define $V=\sp\{k_1,k_2,k_4\}$ and $V^\sharp$ is the direct sum of
$\sp\{\omega\}$ and $(-2\epsilon)$-dimensional spacetime. We decompose
$l=l^{[3]}+l^\sharp$ according to $V\oplus V^\sharp$. Then 
\begin{equation}
  (l^\sharp)^2=-\mu_{11}-\frac{s}{t(s+t)} x_4^2\equiv -\lambda_{11}.
\end{equation}
It is clearly that $D_1,\ldots, D_4$ are functions in $x_1,x_2,x_3$
and $\lambda_{11}$ only. We may redefine Baikov variables,
\begin{equation}
 z_1 = D_1,\quad  z_2 = D_2,\quad  z_3 = D_3,\quad  z_4 = D_4.
\end{equation}
Only $4$ variables are needed. Repeat the previous process,
\begin{equation}
  I_\text{box}[N]=\frac{1}{i\pi^{3/2} \Gamma
    (\frac{D-3}{2})} \frac{1}{4\sqrt{-st(t+s)}}
    \int dz_1 dz_2 dz_3 dz_4 \tilde F(z_1,z_2,z_3,z_4)^{\frac{D-5}{2}} \frac{N}{z_1
    z_2 z_3 z_4}\,,
\end{equation}
if $N$ has no $l_1\cdot \omega$ dependence. $\tilde  F(z_1,z_2,z_3,z_4)=F(z_1,z_2,z_3,z_4,0)$.

Baikov representation also works for higher-loop and both planar and
nonplanar diagrams. For example, in a scheme of which all external
particles are in $4D$, a two-loop integral with $n\geq 5$
points becomes
\begin{equation}
I^{(2)}_{n\geq 5} [N]= \frac{2^{D-6}}{\pi^{5}\Gamma(D-5) J} \int \prod_{i=1}^{11} d z_i \hspace{0.6mm}
  F(z)^{\frac{D-7}{2}} \frac{N}{z_1 \cdots z_k}\,,
\label{two-loop-integral-z}
\end{equation}
where $J$ is a Jacobian without $D$ dependence. Here $F(z)$ is the
determinant $\mu_{11}\mu_{22} -\mu_{12}^2$ in Baikov representation. In the same scheme, for a two-loop
amplitude with $n< 5$ point, we can integrate out $5-n$ spurious
directions and get,
\begin{equation}
I^{(2)}_{n<5} [N]= \frac{2^{D-n-1}}{\pi^{n}\Gamma(D-n) J} \int \prod_{i=1}^{2n+1} d z_i \hspace{0.6mm}
  F(z)^{\frac{D-n-2}{2}} \frac{N}{z_1 \cdots z_{k}} \,.
\label{two-loop-integral-4z}
\end{equation}
We leave the Baikov representation of massless double box
diagram as an exercise (Exercise \ref{dbox_Baikov}).

For deriving IBP relations, the overall prefactors are irrelevant. In the
rest of this chapter, we neglect these factors in Baikov
representation. In general for an $L$-loop integral in a 
scheme of which external particles are in $4D$,
\begin{equation}
\label{Baikov}
  I^{(L)}_{n} [N]\propto \int \prod_{i=1}^{\phi(n) L+\frac{L(L-1)}{2}}
  dz_i \ F(z)^\frac{D-L-\phi(n)}{2} \frac{N}{z_1 \ldots z_k}\,,
\end{equation}
where
\begin{equation}
  \phi(n)=\left\{
      \begin{array}{cc}
        n, & n<5\\
        5,& n\geq 5
      \end{array}
\right . .
\end{equation}
The Baikov polynomial $F(z)$ is the determinant $\det(\mu_{ij})$ if
$n\geq 5$, or   the determinant $\det(\lambda_{ij})$ is $n<5$.

\subsection{Unitarity cuts in Baikov representation}

We see that in Baikov representation, inverse propagators are simply linear
monomials. Another feature 
is that unitarity cut structure is clear.

Note that now all inverse propagators are linear, so a unitarity cut
$D_i^{-1} \to \delta(D_i)$ just means to set certain $z_i$ as zero in
\eqref{Baikov}. For a given $c$-fold cut ($0\leq c\leq k$), let
$\mathcal{S}_\mathrm{cut}$, $\mathcal{S}_\mathrm{uncut}$ and
$\mathcal{S}_\mathrm{ISP}$ be the sets of indices labelling
cut propagators, uncut propagators and ISPs, respectively.
$\mathcal{S}_\mathrm{cut}$ thus contains $c$ elements. Furthermore, we
denote $m$ as 
the total number of $z_j$ variables, 
\begin{equation}
  m=\phi(n) L+\frac{L(L-1)}{2}\,,
\end{equation}
and
set $\mathcal{S}_\mathrm{uncut}=\{r_1,\ldots, r_{k-c}\}$
and $\mathcal{S}_\mathrm{ISP}=\{r_{k-c+1},\ldots, r_{m-c}\}$.
Then, by cutting the propagators, $z_i^{-1} \to \delta(z_i),
i \in \mathcal{S}_\mathrm{cut}$, the integrals (\ref{two-loop-integral-z}) and
(\ref{two-loop-integral-4z}) reduce to,
\begin{equation}
I^{(L)}_\mathrm{cut} [N]= \int \frac{d z_{r_1} \cdots d z_{r_{m-c}} }
{z_{r_1}\cdots z_{r_{k-c}}} N F(z)^{\frac{D - L-\phi(n)}{2}} \bigg|_{z_i=0\,, \forall i\in \mathcal{S}_\mathrm{cut}} \,,
\label{cut-z}
\end{equation}

\begin{example}
\label{dbox_Baikov_cut}
  Consider the quintuple cut for $D$-dimensional 
massless double box. (See Exercise \ref{dbox_Baikov}), $D_2=D_3=D_5=D_6=D_7=0$. In this cases, $m=9$.
$\mathcal{S}_\mathrm{uncut}=\{1,4\}$,
$\mathcal{S}_\mathrm{cut}=\{2,3,5,6,7\}$,
$\mathcal{S}_\mathrm{ISP}=\{8,9\}$. Baikov representation \eqref{cut-z}
with this cut reads,
\begin{equation}
  I^{(2)}_\mathrm{penta-cut} [N]= \int \frac{d z_1 d z_4 
    d z_8 d z_9}{z_1 z_4 } F_{[5]}(z)^{\frac{D - 6}{2}}
  N \bigg|_{z_2=z_3= z_5= z_6= z_7=0}\,,
\end{equation}
where,
\begin{gather}
  F_{[5]}(z)=\frac{(s t\c -s z_1\c -2 s z_8\c -2 s z_9\c -t z_1\c-t
    z_4\c+2 z_4 z_8\c -4 z_8 z_9) (2 s z_1 z_9\c +4 s z_8 z_9\c +t z_1 z_4)}{4 s t (s+t)}.
\end{gather}

If we consider the maximal cut $D_1=D_2=\ldots =D_7=0$, then $\mathcal{S}_\mathrm{uncut}=\emptyset$,
$\mathcal{S}_\mathrm{cut}=\{1,2,3,4,5,6,7\}$,
$\mathcal{S}_\mathrm{ISP}=\{8,9\}$. Baikov representation \eqref{cut-z}
on this cut reads,
\begin{equation}
  I^{(2)}_\mathrm{hepta-cut} [N]= \int 
    d z_8 d z_9 F_{[7]}(z)^{\frac{D - 6}{2}}
  N \bigg|_{z_i=0,\ 1\leq i\leq 7}\,,
\end{equation}
and Baikov polynomial on maximal cut is simply,
\begin{equation}
  F_{[7]}(z)=\frac{z_8 z_9 (s t-2 s z_8-2 s z_9-4 z_8 z_9)}{t (s+t)}\,.
\end{equation}
 \end{example}

\subsection{IBPs in Baikov representation}
Note that the higher the unitarity cut is, the simpler the Baikov polynomial
becomes. So We try to use cuts as much as possible to reconstruct the full IBP,
instead of solving the full IBP at once.  Suppose that we consider a $c$-fold cut and make an IBP ansatz as, 
\begin{align}
  \label{ansatz}
0\hspace{-0.5mm}&=\hspace{-0.5mm}\int \hspace{-0.8mm}
d \bigg( \hspace{-0.4mm} \sum_{i=1}^{m-c} \hspace{-0.8mm}
     \frac{(-1)^{i+1} a_{r_i} F(z)^{\frac{D-h}{2}}}{z_{r_1}\cdots z_{r_{k-c}}}
     d z_{r_1} \hspace{-0.5mm} \wedge \hspace{-0.5mm} \cdots \widehat{d z_{r_i}} \cdots
     \hspace{-0.5mm} \wedge \hspace{-0.5mm} d z_{r_{m-c}}
                    \hspace{-0.5mm} \bigg)\nn\\
\c&= \c \int  \sum_{i=1}^{m-c} \c \bigg(\c \frac{\partial a_{r_i}}{\partial z_{r_i}}
   \c \bigg)F(z)^{\frac{D-h}{2}} \omega\c +\c \frac{D\c -\c h}{2}\c \sum_{i=1}^{m-c}
    \bigg(a_{r_i}\frac{\partial F}{\partial z_{r_i}}\bigg) F^\frac{D-h-2}{2}
    \omega\c - \sum_{i=1}^{k-c}  \frac{a_{r_i}}{z_{r_i}}F(z)^{\frac{D-h}{2}} \omega\,.
\end{align} 
where $\omega$ is the measure $dz_{r_1}\wedge \ldots
\wedge dz_{r_{m-c}}/(z_{r_1}\cdots z_{r_{k-c}})$ and $h=D-L-\phi(n)$. The second sum contains
integrals in $D-2$ dimension while the third sum contains doubled
propagators. 

If it is required that resulting IBP has no dimensional shift or doubled
poles \cite{Ita:2015tya,Larsen:2015ped}, we have the {\it syzygy equations}, 
\begin{align}
b F + \sum_{i=1}^{m-c} a_{r_i}\frac{\partial F}{\partial
z_{r_i}} &=0\,, \label{eq:syzygy_1} \\
a_{r_i} + b_{r_i} z_{r_i}&=0\,, \hspace{4mm} i=1,\ldots, k-c \,,
\label{eq:syzygy_2}
\end{align}
where $a_{r_i}$, $b$ and $b_{r_i}$ must be polynomials in $z_j$. Note
that the last $(k-c)$
equations in Eq.~(\ref{eq:syzygy_2}) are trivial since they are solved as
$a_{r_i}=-b_{r_i} z_{r_i}$. So alternatively, we have only one syzygy equation,
\begin{equation}
   b F -\sum_{i=1}^{k-c} b_{r_i}\bigg(z_{r_i}\frac{\partial F }{\partial z_{r_i}}\bigg) \hspace{0.8mm}+\hspace{0.8mm}
    \sum_{j=k-c+1}^{m-c} a_{r_j}\frac{\partial F}{\partial z_{r_j}}
    \hspace{0.8mm}=\hspace{0.8mm} 0
\label{syz}
\end{equation}
for polynomials $b_{r_i}$, $a_{r_i}$ and $b$. 
These equations are similar to the tangent condition of a
hypersurface in differential geometry, however, we require polynomial
solutions. So we apply algebraic geometry to study these equations. 

We find that it is trivial to solve the two equations
\eqref{eq:syzygy_1} and \eqref{eq:syzygy_2} separately, and it is easy
to combine the two individual solutions for \eqref{eq:syzygy_1} and
\eqref{eq:syzygy_2}  together via module intersection, to get the
simultaneous solution. This {\it module intersection} approach is much
more efficient than to solve \eqref{syz} directly.

\section{Syzygies}
Syzygy can be understood as relations of polynomials. Consider the
ring $R=\Fpoly$ and $R^m$, the set of all $m$-tuple of $R$. $R^m$ in
general is not a ring but $R\times R^m \to R$ is well-defined as,
\begin{gather}
  f \cdot (f_1,\ldots ,f_m) \mapsto (f f_1, \ldots, f f_m).
\end{gather}
This leads to the definition of modules.
\begin{definition}
  A module $M$ over the ring $R$ is an Abelian group ($+$) with a map
  $R\times M\to M$ such that,
  \begin{enumerate}
  \item $r\cdot (m_1+m_2)=r \cdot m_1+r \cdot m_2$, $\forall r\in R$, $m_1,m_2\in M$. 
   \item $(r_1+r_2)\cdot m=r_1 \cdot m+r_2 \cdot m$, $\forall
     r_1,r_2\in R$, $m\in M$. 
     \item $(r_1 r_2) \cdot m=r_1 \cdot (r_2 \cdot m)$, $\forall
     r_1,r_2\in R$, $m\in M$. 
     \item $1\cdot m=m$. $1\in R$, $\forall m\in M$. 
  \end{enumerate}
\end{definition}
For example $R^m$, $I$ and $R/I$ are all $R$-modules, where $I$ is an
ideal of $R$. To simplify notations, we formally write an element
$(f_1,\ldots f_m)\in R^m$ as $f_1\e_1+\ldots f_m \e_m$.  

\begin{proposition}
  Any submodule of $R^m$ is finitely generated.
\end{proposition}
\begin{proof}
  This is a generalization of Theorem \ref{thm_Noether}. See Cox,
  Little and O'Shea \cite{opac-b1094391}.
\end{proof}

\begin{definition}
  Given a $R$ module $M$, the syzygy module of $m_1, \ldots m_k\in M$,
  $\syz(m_1\ldots m_k)$, is the submodule of $R^k$ which consists of all $(a_1, \ldots a_k)$
  such that
  \begin{equation}
    a_1 \cdot m_1 + a_2 \cdot m_2 + \ldots a_k \cdot m_k=0\,.
  \end{equation}
\end{definition}
So \eqref{syz} defines a syzygy module with $M=R$,
i.e., ``relations'' between polynomials. Naively, given $f_1,\ldots, f_k$,
it is clearly that $f_j \e_i -f_i \e_j\in R^k$, $i\not=j$ is a syzygy
for $f_1,\ldots, f_k$. Such a syzygy is called a {\it principal
  syzygy} which is denoted as $P_{ij}$.

In some cases, principal syzygies generate the whole syzygy module of
given polynomials. For example,
\begin{proposition}
\label{principle_syzygy}
  Given $f_1,\ldots, f_k$ in $R=\Fpoly$, if $\la f_1,\ldots,f_k\ra=\la
  1\ra$, then $\syz(f_1,\ldots f_k)$ is generated by principal
  syzygies $P_{ij}=f_j \e_i -f_i \e_j$, $1\leq i\not =j\leq k$.
\end{proposition}
\begin{proof}
  We have $q_1 f_1 +\ldots q_k f_k=1$, where $q_i$'s are in $R$. For
  any element in $\syz(f_1,\ldots f_k)$,
  \begin{equation}
    a_1 f_1 +\ldots a_k f_k=0\,,
  \end{equation}
we can rewrite $a_i$ as,
\begin{align}
  a_i=\sum_{j=1}^k a_i q_j f_j=\big(\sum_{\substack{j=1\\j\not=i}}^ka_i q_j f_j
       \big)+a_i q_i f_i \nn =\sum_{\substack{j=1\\j\not=i}}^ka_i q_j f_j -\sum_{\substack{j=1\\j\not=i}}^ka_j q_i f_j \equiv \sum_{\substack{j=1\\j\not=i}}^ks_{ij} f_j,
\end{align}
where $s_{ij}=a_i q_j-a_j q_i$ is a polynomial and antisymmetric in indices. Hence, this syzygy
\begin{align}
  \sum_{i=1}^k a_i \e_i= \sum_{i=1}^k\sum_{\substack{j=1\\j\not=i}}^k
  s_{ij} f_j \e_i=\sum_{\substack{i=1,j=1\\j\not=i}}^k
  \frac{s_{ij}}{2}(f_j \e_i -f_i \e_j)= \sum_{\substack{i=1,j=1\\j\not=i}}^k
  \frac{s_{ij}}{2} P_{ij}\,,
\end{align}
is generated by principal syzygies. 
\end{proof}

\begin{example}
\label{circle_syzygy}
  Consider the polynomial $F=x^2+y^2-1$ in $\Q[x,y]$. Define
  $f_1=\partial F/\partial x=2x$, $f_2=\partial F/\partial y=2y$ and
  $f_3=F$. It is clear that $\la 2x,2y,x^2+y^2-1\ra=\la 1\ra$. Hence
  $\syz(f_1,f_2,f_3)$ is generated by,
  \begin{gather}
    (y,-x,0),\quad (x^2+y^2-1,0,-2x),\quad (0,x^2+y^2-1,-2y)\,.
\end{gather}
Note that we see that $F=0$ defines the unit circle. The tangent
vector at any point on the circle is,
\begin{gather}
  y\frac{\partial }{\partial x}-x\frac{\partial }{\partial y}\,,
\end{gather}
which corresponds to the first generator in \eqref{circle_syzygy}. 
\end{example}

In general, syzygy module for given polynomials can be found by \GB\
computation. For a \GB\ $G=\{g_1,\ldots g_m\}$ in a certain monomial
order, consider two elements $g_i$, $g_j$, $i<j$. Let $S(g_i,g_j)=a_i g_i
+a_j g_j$ be the
S-polynomial (Definition \ref{S-polynomial}). $S(g_i,g_j)$ must be
divisible by $G$, hence, by polynomial division (Algorithm \ref{multivariate_polynomial_division}), 
\begin{equation}
  a_i g_i +a_j g_j = \sum_{l=1}^m q_l g_l\,.  
\end{equation}
Clearly, this is a syzygy of $g_1,\ldots g_m$, which explicitly reads $ q_1 \e_1 +\ldots
(q_i-a_i) \e_i +\ldots + (q_j -a_j) \e_j +\ldots q_m \e_m$. We call
this syzygy, {\it reduction of an S-polynomial}  and denote it as
$s_{ij}\equiv \sum _{l=1}^m (s_{ij})_l \e_l$. 
\begin{thm}[Schreyer]\
\label{Schreyer}
  \begin{enumerate}
  \item  For a \GB\  $G=\{g_1,\ldots g_m\}$ in $R=\Fpoly$, $\syz(g_1,\ldots g_m)$
  is generated by  reductions of  S-polynomials, $s_{ij}$.   
  \item For generic polynomials $\{f_1,\ldots,f_k\}$ in  $R=\Fpoly$,
    let $G=\{g_1,\ldots g_m\}$ be their \GB\ in a certain monomial
    order. Suppose that the conversion relations are,
    \begin{equation}
      g_i=\sum_{j=1}^k a_{ij} f_j\quad f_i=\sum_{j=1}^m b_{ij} g_j\,.
    \end{equation}
The $\syz(f_1,\ldots f_k)$ is generated by,
\begin{gather}
  \sum_{i=1}^m\sum_{j=1}^k (s_{\alpha\beta})_i a_{ij} \e_j, \quad
  1\leq \alpha<\beta \leq m\nn\\
  \e_i-\sum_{l=1}^k \sum_{j=1}^m b_{ij} a_{jl} \e_l, \quad
  1\leq i \leq k
\label{syzygy_GB}
\end{gather}
  \end{enumerate}
 \end{thm}
 \begin{proof}
   See Cox, Little and O'Shea \cite{opac-b1094391}. Note that in
   second line of \eqref{syzygy_GB}, the relations are coming from the
   map from $f_i$'s to $G$ and the inverse map.
 \end{proof}
This theorem also generalizes to modules. Given several elements
$m_1,\ldots m_k$ in
$R^m$, we can define a module order which is an extension of monomial
order. Then we can compute \GB\ and the syzygy module of $m_1,\ldots
m_k$ \cite{opac-b1094391}.

In practice, we may use \pmb{syz} in \Singular or \pmb{syz} in
\Macaulay , to find the syzygy module of polynomials or elements in
$R^m$. See
alternative ways of finding syzygies with the linear algebra method
\cite{Schabinger:2011dz} or by F5 algorithm \cite{MR2772175}. 

\section{Polynomial tangent vector field}
In this section, we use the tool of syzygy to study {\it polynomial tangent
vector field} \cite{Hauser1993}. See \cite{Bern:2017gdk} for the
physical meaning of these tangent vectors for IBPs.

Let $F(z)$ be a polynomial in $R=\mathbb
C[z_1,\ldots,z_n]$. $F=0$ defines a hypersurface (reducible or irreducible). The set of
polynomial tangent fields, $\TF_F$, is the
submodule in $R^n$ which consists of all $(a_1,\ldots, a_n)\in R^n$ such that,
\begin{equation}
  \label{polynomial_tangent_field}
  \sum_{i=1}^n a_i \frac{\partial F}{\partial z_i} = b F\,,
\end{equation}
for some polynomial $b$ in $z_i$'s. \eqref{polynomial_tangent_field} is a
syzygy equation which can be solved by the algorithm in Theorem \ref{Schreyer}. (We
drop the factor $b$ in the definition of $\TF_F$, since this factor can be easily recovered
later.) Mathematically, $\TF_F$ is called the set of {\it derivations}, from
$R/\la F \ra$ to $R/\la F \ra$. 

Geometrically, if a point $\xi=(\xi_1,\ldots ,\xi_n)\in \mathbb C^n$ is on
the hypersurface $\mathcal Z(F)$, then,
\begin{equation}
  \sum_{i=1}^n a_i (\xi_1,\ldots ,\xi_n) \frac{\partial F}{\partial z_i}(\xi_1,\ldots ,\xi_n) = 0\,,
\end{equation}
and $(a_1(\xi),\ldots,a_n(\xi))$ is along the tangent direction of
$\mathcal Z(F)$. This is the origin of terminology, {\it polynomial tangent
vector field}.

Although syzygy computation by Theorem \ref{Schreyer} can find
$\mathbf T_F$ for any polynomial $F$, it is interesting to study
the geometric properties of $F$ and $\mathbf T_F$.

\begin{definition}
\label{singular_ideal}
  For a polynomial $F$ in $R=\mathbb
C[z_1,\ldots,z_n]$. The singular ideal $I_S$ for $F$ is defined to be,
\begin{equation}
  I_s=\la\frac{\partial F}{\partial z_1},\ldots, \frac{\partial
    F}{\partial z_n}, F\ra \,,
\end{equation}
If $I_s=\la 1 \ra$, then we call the hypersurface $\mathcal Z(F)$
smooth. Otherwise we call points in $\mathcal Z(I_s)$ singular points.  \end{definition}

Intuitively, at a singular point $\xi\in \mathcal Z(I_s)$, $F$ and all its first derivates vanish.  Hence near $\xi$,
$F=0$ does not define a complex submanifold with codimension $1$. 

If a hypersurface is smooth, then by Definition \ref{singular_ideal} and 
Proposition \ref{principle_syzygy}, we have the following statement. 
\begin{proposition}
  If  $F$ in $R=\mathbb
C[z_1,\ldots,z_n]$ defines a smooth hypersurface, then $\TF_F$ is
generated by principal syzygies of $\partial F/\partial z_1,\ldots, \partial
    F/\partial z_n, F$.
\end{proposition}

For instance, in Example \ref{circle_syzygy}, the unit circle is
clearly smooth. Hence its polynomial tangent vector fields is
generated by principal syzygies. This can be understood as an algebraic
version of implicit function theorem. 

The singular cases are more interesting and subtle. 
\begin{example}
  Let $F=y^2-x^3$. $F=0$ is not a smooth curve, since the singular variety
  is $I_s=\la -3x^2,2y,y^2-x^3\ra=\la x^2, y\ra \not=\la 1\ra$. So
  there is one singular point at $(0,0)$ which is a cusp point. We
  cannot just use principal syzygies to generate $\TF_F$, so we turn
  to Theorem \ref{Schreyer}. 

Define that $\{f_1,f_2,f_3\}=\{-3x^2,2y,y^2-x^3\}$. Note that this is
a \GB\ in \grevlex, although it is not a reduced \GB. 
\begin{itemize}
\item $S(f_1,f_2)=(2y)f_1+(3x^2)f_2=0$ hence we get a syzygy generator
  $\mathcal S_1=(2y,-3x^2,0)$.
\item $S(f_2,f_3)=(-x^3)f_2-(2y)f_3=-2y^3=-y^2 f_2$. $\mathcal S_2=(0,
  -x^3+y^2 ,-2 y)$.
\item  $S(f_3,f_1)=-3f_3+x f_1=-3y^2=-\frac{3}{2}y f_2$. $\mathcal S_3=(x,
 \frac{3}{2}y,-3)$.
\end{itemize}
$\mathcal S_3$ is not from principal syzygies. Locally it
characterizes the scaling behavior of the curve $y^2-x^3=0$ near the cusp
point $(0,0)$. It is a {\it weighted Euler vector field} \cite{Hauser1993}. 

Dropping the factor $b$ in \eqref{polynomial_tangent_field}, we find
that $\TF_F$ is generated by,
\begin{equation}
  (2y, -3x^2),\quad (0,-x^3+y^2), \quad (x,\frac{3y}{2})\,.
\end{equation}

\end{example}


\begin{proposition}
   Let $F\in R=\mathbb C[z_1,\ldots,z_n]$, $\TF_F$ is a Lie algebra with
   $[,]$ defined as that for vector fields. 
\end{proposition}
\begin{proof}
  Let $v_1=(a_1, \ldots ,a_n)$ and $v_2=(b_1,\ldots ,b_n)$ be two polynomial
  tangent vector fields, 
  \begin{equation}
    \sum_{i=1}^n a_i \frac{\partial F}{\partial z_i} = A F,\quad 
\sum_{i=1}^n b_i \frac{\partial F}{\partial z_i} = B F\,,
  \end{equation}
where $A$ and $B$ are polynomials. $[v_1,v_2]$'s i-th component is,
\begin{gather}
\sum_{j=1}^n \bigg( a_j \frac{\partial
    b_i}{\partial z_j}- b_j \frac{\partial a_i}{\partial
    z_j}\bigg) \,,
\end{gather}
Hence $[v_1,v_2]$ acts on $F$ as,
\begin{gather}
 \sum_{i=1}^n \sum_{j=1}^n \bigg( a_j \frac{\partial
    b_i}{\partial z_j}- b_j \frac{\partial a_i}{\partial
    z_j}\bigg) \frac{\partial F}{\partial z_i}=F\cdot \sum_{j=1}^n\bigg( a_j \frac{\partial B}{\partial z_i}-b_j \frac{\partial A}{\partial z_i}\bigg)\,.
\end{gather}
so $[v_1,v_2]$ is in $\TF_F$. 
\end{proof}
In general, $\TF_F$ is an infinite-dimensional Lie algebra over
$\C$. We may call $\TF_F$ a {\it tangent algebra}.

If we require a polynomial vector field $(a_1, \ldots ,a_n)$  tangent to a list of
hypersurfaces defined by $F_1,\ldots,F_k$, like
the case of \eqref{eq:syzygy_1} and \eqref{eq:syzygy_2},
\begin{gather}
     \sum_{i=1}^n a_i \frac{\partial F_1}{\partial z_i} = A_1 F_1(z)\nn\\
\ldots \nn\\
\sum_{i=1}^n a_i \frac{\partial F_k}{\partial z_i} = A_k F_k(z)\,.
\label{vector_field_hypersurfaces}
\end{gather}
Then by definition, the solution set of such $(a_1, \ldots ,a_n)$'s is
the intersection of modules $\TF_{F_1}\cap\ldots \cap\TF_{F_k}$, which
is again a submodule of $R^n$. On the other hand, 
\begin{proposition}
\label{tangent_algebra_component}
  If a polynomial $F$ in $R=\C[z_1,\ldots,z_n]$ factorizes as,
  \begin{equation}
    F=f_1^{s_1}\ldots f_k^{s_k} \,,
  \end{equation}
where $f_i$'s are irreducible polynomials in $R$ and $f_i \not | f_j$
if $i\not =j$. $s_i$'s are positive integers. Then $\TF_F=\TF_{f_1}\cap\ldots \cap\TF_{f_k}$.
\end{proposition}
\begin{proof}
  It is clear that $\TF_F\supset \TF_{f_1}\cap\ldots
  \cap\TF_{f_k}$. For $(a_1, \ldots ,a_n)\in \TF_F$,
  \begin{equation}
    \sum_{l=1}^k s_l \bigg(\sum_{i=1}^n a_i \frac{\partial
      f_l}{\partial z_i}\bigg) \frac{F}{f_l}=b F\,.
  \end{equation}
For a fixed index $t$, $1\leq t\leq k$, divide the above expression by
$f_t^{s_t-1}$, 
  \begin{equation}
  s_t \bigg(\sum_{i=1}^n a_i \frac{\partial
      f_t}{\partial z_i}\bigg) \frac{F}{f_t^{s_t}}+
    \sum_{\substack{l=1\\  l\not= t}}^k s_l \big(\sum_{i=1}^n a_i \frac{\partial
      f_l}{\partial z_i}\big) \frac{F}{f_l f_t^{s_t-1}}=b \frac{F}{ f_t^{s_t-1}}\,.
  \end{equation}
Note the second term on l.h.s and the r.h.s are polynomials
proportional to $f_t$. Hence,
\begin{equation}
   s_t \bigg(\sum_{i=1}^n a_i \frac{\partial
      f_t}{\partial z_i}\bigg) \frac{F}{f_t^{s_t}}\,,
\end{equation}
is also proportional to $f_t$. However $f_t$ does not divide
$F/f_t^{s_t}$, since $f_i$'s are distinct irreducible polynomials. So
$f_t$ divides $\sum_{i=1}^n a_i \partial
      f_t/\partial z_i$ and $(a_1, \ldots ,a_n)\in \TF_{f_t}$, and $\TF_F\subset \TF_{f_1}\cap\ldots
  \cap\TF_{f_k}$.
\end{proof}
It implies that for a reducible hypersurface, its tangent algebra is
the intersection of tangent algebras of all its irreducible
components \cite{Hauser1993}. 

In practice, give a syzygy equation system \eqref{vector_field_hypersurfaces}, we can
first determine each $\TF_{F_i}$ and then calculate
the intersection $\TF_{F_1}\cap\ldots \cap\TF_{F_k}$. (See
\cite[Chapter~5]{opac-b1094391} for the algorithm of computing
intersection of submodules.) Furthermore, for
each $\TF_{F_i}$, if $F_i$ is factorable, we can use Proposition
\ref{tangent_algebra_component} to further divide the problem. This
divide-and-conquer approach is in general much more efficient than
solving \eqref{vector_field_hypersurfaces} at once.

More specifically, it is known that the tangent condition \eqref{eq:syzygy_1} for the
Baikov polynomial $F$ can be solved directly without any computation
\cite{Ita:2015tya}. In Baikov variables, the generating tangent vectors for $F$
have a beautiful structure: they contain at most linear functions in
Baikov variables \footnote{We learnt this generating tangent system
  from Roman Lee's blog \url{
    http://mathsketches.blogspot.ru/2010/07/blog-post.html} (in
  Russian). The completeness of these generating vectors, with the
  mathematical proof, is given in \cite{Boehm:2017wjc}.}. So for our
IBP computation, all $\TF_{F_i}$'s known without computation. The
intersection of these modules can be computed efficient by the compute
algebra system {\sc Singular}.

\section{IBPs from syzygies and unitarity}
With Baikov representation, unitarity cut and syzygy computation, we
introduce some recent IBP generating algorithms
\cite{Ita:2015tya,Larsen:2015ped} with the two-loop double box as an example. 

For the massless double box, define
\begin{equation}
    I[m_1,\ldots ,m_9]=\int \frac{d^D l_1}{i \pi^{D/2}} \frac{d^D l_2}{i
      \pi^{D/2}} \frac{(l_1\cdot k_4)^{-m_8}(l_2\cdot k_1)^{-m_9}}{D_1^{m_1} \ldots D_7^{m_7}}\,.
  \end{equation}
Our target integral space is the set of all list $(m_1\ldots m_9)$
such that $m_i\leq 1$, $i=1,\ldots,7$, $m_j\leq 0$ , $j=8,\ldots,9$,
since we try to find IBPs without doubled propagators.

\begin{example}
  Consider the massless double box with maximal cut. From
  Example \ref{dbox_Baikov_cut}, we see that with maximal cut the
  Baikov polynomial is
  \begin{equation}
\label{F_7cut}
 F_{[7]}= \frac{z_8 z_9 (s t-2 s z_8-2 s z_9-4 z_8 z_9)}{t (s+t)}\,.
  \end{equation}
The syzygy equation is,
\begin{equation}
  a_8 \frac{\partial F_{[7]}}{\partial z_8}+a_9
  \frac{\partial F_{[7]}}{\partial z_9}=\beta  F_{[7]}\,.
\end{equation}
Solutions of $(a_8,a_9)$ form $\TF_{F_{[7]}}$, the
tangent algebra of $F_{[7]}$. We leave the computation of $\TF_{F_{[7]}}$ as an exercise. There are $3$ generators of
$\TF_{F_{[7]}}$,
\begin{gather}
 v_1=\big(-(t-2 z_8) z_8 , (t-2 z_9) z_9 \big),\quad 
 v_2=\big (2 (s+t) z_8 z_9, -(t-2 z_9) z_9 (s+2 z_9)\big),\nn\\ 
v_3=\big(0 , -z_9 (s t-2 s z_8-2 s z_9-4 z_8 z_9)\big)\,,
\end{gather}
Using these generators and the ansatz \ref{ansatz}, we get IBPs
without double propagators. For instance, from the first generator we
have the IBP,
\begin{equation}
  I[1,1,1,1,1,1,1,-1,0]=I[1,1,1,1,1,1,1,0,-1] +
  \ldots ,
\end{equation}
and from the second generator,
\begin{gather}
  4 (D-3) I[1,1,1,1,1,1,1,0,-2]+(3 D s-12 s-2 t)
  I[1,1,1,1,1,1,1,0,-1]\nn\\
-\frac{1}{2} (D-4) s t I[1,1,1,1,1,1,1,0,0]
  =0 + \dots \,.
\end{gather}
Note that with maximal cut, any integral with at least one $m_i<1$,
$i=1,\ldots 7$, is neglected. ``$\ldots$'' stands for these
integrals. 

To get all IBPs with maximal cut, we need to consider vector fields
$q_1 v_1+q_2 v_2+q_3 v_3$ where $q_1$, $q_2$ and $q_3$ are arbitrary
polynomials in $z_8$, $z_9$ up to a given degree. When the smoke is
clear, we find that all integrals with $m_i=1$, $i=1,\ldots 7$ and
$m_j\leq 0$ , $j=8,\ldots,9$ are reduced to
$I[1,1,1,1,1,1,1,0,0]$,
$I[1,1,1,1,1,1,1,-1,0]$ and integrals with fewer-than-$7$
propagators. 
\end{example}
 
\begin{example}
\label{slashed_box_cut}
  Consider the quintuple cut of the massless double box,
 $D_2=D_3=D_5=D_6=D_7=0$. The goal is to study integrals
 $I_\text{dbox}[m_1,m_2,\ldots m_9]$ such that
 $m_2=m_3=m_5=m_6=m_7=1$, $m_1,m_4\leq 1$, $m_8,m_9$ non-positive. 
The syzygy equations read,
\begin{eqnarray}
   a_1 \frac{\partial F_{[5]}}{\partial z_1}+a_4
  \frac{\partial F_{[5]}}{\partial z_4}+a_8 \frac{\partial F_{[5]}}{\partial z_8}+a_9
  \frac{\partial F_{[5]}}{\partial z_9}&=&\beta  F_{[5]}  \\
a_1 &=& b_1 z_1\\
a_4 &=& b_4 z_4
\label{dbox_5cut}
\end{eqnarray}
In the formal language, the solutions of last two equations form
a tangent algebra $\TF_{14}$ with generators,
\begin{eqnarray}
  (z_1,0,0,0),\quad  (0,z_4,0,0),\quad  (0,0,1,0),\quad  (0,0,0,1).
\end{eqnarray}
The first equation can be solved by \pmb{syz} in \Singular\ and
\Macaulay, which leads to a tangent algebra $\TF_{F[5]}$. Then
the solution set of \eqref{dbox_5cut}  of 
$\TF_{F[5]}\cap \TF_{14}$. This intersection of submodules can be
calculated by  \pmb{intersect} in \Singular\ and
\Macaulay. 

Again we find IBPs with this tangent algebra. All integrals with
$m_2=m_3=m_5=m_6=m_7=1$, $m_1,m_4\leq 1$, $m_8,m_9$ non-negative are reduced to
$3$ master integrals $I[1,1,1,1,1,1,1,0,0]$,
$I[1,1,1,1,1,1,1,-1,0]$, $I[0,1,1,0,1,1,1,0,0]$ and integrals with fewer-than-$5$
propagators. 
\end{example}

In general it is easy to obtain IBPs with maximal cut, since the
number of variable is small. We may use symmetries and IBPs with maximal cut,
numerically, to find all MIs \cite{AZURITE}. It takes only a few seconds to
find all master integrals for massless double box. 
\begin{itemize}
  \item double box, $I[1,1,1,1,1,1,1,-1,0]$, $I[1,1,1,1,1,1,1,0,0]$,
    \item slashed box,  $I[0,1,1,0,1,1,1,0,0]$,
\item box bubble, $I[0,1,0,1,1,1,1,0,0]$,
\item double bubble, $I[1,0,1,1,0,1,0,0,0]$,
\item bubble triangle, $I[0,1,0,1,0,1,1,0,0]$,
\item $t$-channel sunset, $I[0,1,0,0,1,0,1,0,0]$,
\item $s$-channel sunset, $I[0,0,1,0,0,1,1,0,0]$.
\end{itemize}
We define that $I[m_1,\ldots m_9]$ is lower than $I[n_1,\ldots n_9]$ if
$m_i\leq n_i$, $i=1,\ldots 7$. For example, $s$-channel sunset  is lower
than the slashed box. A triple cut $D_3=D_6=D_7=0$ contains all
information of the quintuple cut in Example
\ref{slashed_box_cut}. Since here the lowest master integrals are double
bubble, bubble triangle, $t$-channel sunset, $s$-channel sunset, we
can see that the following four cuts,
\begin{gather}
  D_1=D_3=D_4=D_5=0,\quad D_2=D_4=D_6=D_7=0\nn\\
  D_2=D_5=D_7=0,\quad D_3=D_6=D_7=0\,,
\end{gather}
 determine complete IBPs without cut.

By this method \cite{Larsen:2015ped}, a
Mathematica code with the communication with \Singular, analytically
reduces all double box integrals with numerator
rank $\leq 4$, to the $8$ master integrals in about $39$ seconds for
massless double box on a laptop. Similarly, it takes about $162$ seconds
for the analytic IBP reduction of one-massive double box. 

In \cite{Boehm:2018fpv}, we tested our module intersection IBP method
in Baikov representation, with the help of the primitive implement of
sparse linear algebra and rational function reconstruction. We
successfully reduce the $2$-loop $5$-point nonplanar hexa-box
integrals with numerators up to degree $4$.

We expect
that combined with sparse linear algebra and finite-field fitting \cite{vonManteuffel:2014ixa,vonManteuffel:2016xki ,Peraro:2016wsq}
techniques, it can solve some very difficult two-loop/three-loop IBP problems in the near future.  


\section{Exercise}
\begin{ex}[Baikov representation of massless double box]
\label{dbox_Baikov}
  Consider two-loop massless double box diagram
  (Fig. \ref{graph_dbox}) with inverse propagators
  $D_1,\ldots, D_7$ defined in \eqref{dbox_propagators}. Let
  \begin{equation}
    I_\text{dbox}[N]=\int \frac{d^D l_1}{i \pi^{D/2}} \frac{d^D l_2}{i
      \pi^{D/2}} \frac{N}{D_1 \ldots D_7}.
  \end{equation}
By integrand reduction, we see that $N$ can be a polynomial in
$\mu_{11}$, $\mu_{22}$ and $\mu_{12}$, but at most linear in
$(l_1\cdot \omega)$ and $(l_2 \cdot \omega)$. Terms linear in
$(l_1\cdot \omega)$ and $(l_2 \cdot \omega)$ are spurious so 
dropped. Terms in $\mu$'s can be converted to integrals without
$\mu$'s in higher dimension, via Schwinger parameterization
\cite{Bern:2003ck}. Or alternatively,
polynomials in $\mu$'s
or $(l_i\cdot \omega)$ can be directly integrated out by {\it adaptive
  integrand decomposition} \cite{Mastrolia:2016dhn}, using Gegenbauer polynomials
techniques.  Hence we assume $N$ contains no $\mu$'s,
$(l_1\cdot \omega)$ or $(l_2 \cdot \omega)$. 

\begin{enumerate}
\item The original Van Neerven-Vermaseren variables are define in
  \eqref{dbox_vNV} and $\mu_{ij}=-l_i^\perp \cdot l_j^\perp$. To
  integrate out $\omega$ direction, define $V_1=\sp\{k_1,k_2,k_4\}$ and $V^\sharp$ as the direct sum
  of $\sp\{\omega\}$ and $(-2\epsilon)$ extra spacetime. Decompose
  $l_i=l_i^{[3]}+l_i^\sharp$ and denote $(l_i^\sharp\cdot
  l_j^\sharp)=-\lambda_{ij}$. Prove that 
  \begin{equation}
    \lambda_{11}=\mu_{11}\c +\frac{s}{t(s+t)} x_4^2,\quad
\lambda_{22}=\mu_{22}\c +\frac{s}{t(s+t)} y_4^2,\quad
\lambda_{12}=\mu_{12}\c +\frac{s}{t(s+t)} x_4 y_4\,,
\end{equation}
and $D_1,\ldots ,D_7$ only depend on
$x_1,x_2,x_3,y_1,y_2,y_3,\lambda_{11},\lambda_{22},\lambda_{12}$.  

\item
  Integrate over the solid angle parts of $l_1^\sharp$ and
  $l_2^\sharp$ to get
  \begin{gather}
     I_\text{dbox}[N]=\frac{2^{D-5}}{\pi^4\Gamma(D-4)}\int_0^\infty d\lambda_{11}\int_0^\infty 
d\lambda_{22} \int_{-\sqrt{\lambda_{11}
    \lambda_{22}}}^{\sqrt{\lambda_{11} \lambda_{22}}} d\lambda_{12} (\lambda_{11}\lambda_{22}-\lambda_{12}^2)^{\frac{D-6}{2}}\times
\nn\\
\int d^3 l_1^{[3]} d^3 l_2^{[3]} \frac{N}{D_1 \ldots D_7}\,.
  \end{gather}
  \item
    Define $9$ Baikov variables as
    \begin{gather}
      z_i=D_i,\quad 1\leq i \leq 7,\quad z_8=l_1\cdot k_4,\quad
      z_9=l_2\cdot k_1\,.
    \end{gather}
    Find the inverse map $(z_1,\ldots z_9)\mapsto
    (x_1,x_2,x_3,y_1,y_2,y_3,\lambda_{11},\lambda_{22},\lambda_{12})$ and the
    Jacobian of the map. 
\item Derive the Baikov form of integral,
 \begin{gather}
     I_\text{dbox}[N]=\frac{2^{D-5}}{\pi^4\Gamma(D-4) J}
     \int \prod_{i=1}^9 dz_i F(z)^{\frac{D-6}{2}}\frac{N}{D_1 \ldots D_7}\,.
  \end{gather}
Calculate $J$ and $F(z)$ explicitly. Note that the Jacobian of the
changing variables $l_i^{[3]}$ to $ (x_1,x_2,x_3,y_1,y_2,y_3)$ should be
included. 
\end{enumerate}
\end{ex}

\begin{ex}
  Derive the Baikov representation for two-loop pentagon-box diagram,
  (Fig. \ref{graph_pentabox}).
  \begin{figure}[h]
    \centering
    \includegraphics[scale=1]{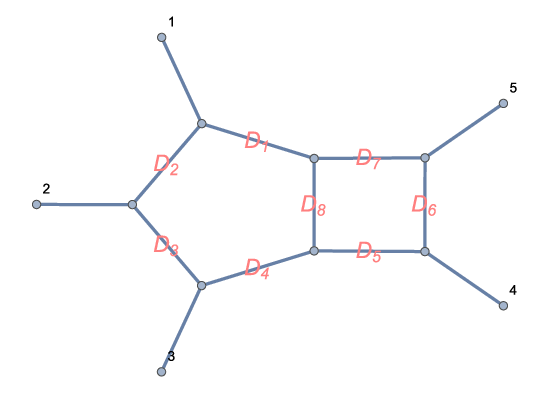}
    \caption{Pentagon box diagram}
\label{graph_pentabox}
  \end{figure}
with inverse propagators,
\begin{gather}
  D_1=l_1^2, \quad D_2=(l_1 - k_1)^2, \quad D_3= (l_1 - k_1 - k_2)^2, \quad D_4= (l_1 - k_1 - k_2 - k_3)^2, \nn\\\quad D_5= (l_2 + k_1 +
    k_2 + k_3)^2, \quad D_6= (l_2 + k_1 + k_2 + k_3 + k_4)^2, \quad
    D_7=l_2^2, \quad D_8=(l_1 + l_2)^2.
\end{gather}
(Hint: define $z_i=D_i$, $i=1,\ldots,8$. $z_{9}=l_1 \cdot k_5$,
$z_{10}=l_2 \cdot k_1$,  $z_{11}=l_2 \cdot k_2$.)
\end{ex}

\begin{ex}
   Consider $f_1=x^3-2 x y$, $f_2=x^2 y-2 y^2 +x$ as Example
   \ref{example_Buchberger}. We know that the \GB\ in \grevlex\ is
   $G=\{g_1,g_2,g_3\}=\{x^2, x y, y^2-\half x\}$. The conversion
   relations are,
   \begin{gather}
       g_1= -y f_1 +x f_2 ,\quad g_2= -\frac{ (1+x y)}{2}f_1+\half x^2 f_2
 ,\quad g_3= -\half y^2 f_1+\half (x y-1) f_2,\\
       f_1=x g_1-2 g_2,\quad f_2=y g_1-2 g_3\,.
   \end{gather}
Find the generators of $\syz(f_1,f_2)$ by Theorem \ref{Schreyer}.
\end{ex}

\begin{ex}
  Let $F=\left(x^2+y^2\right)^2+3 x^2 y-y^3$, the plot of the curve
  $F=0$ is in Figure. \ref{graph_triple}. Determine the singular
  points of this curve and find the polynomial tangent vector fields $\TF_F$.
  \begin{figure}[h]
    \centering
\includegraphics[scale=0.5]{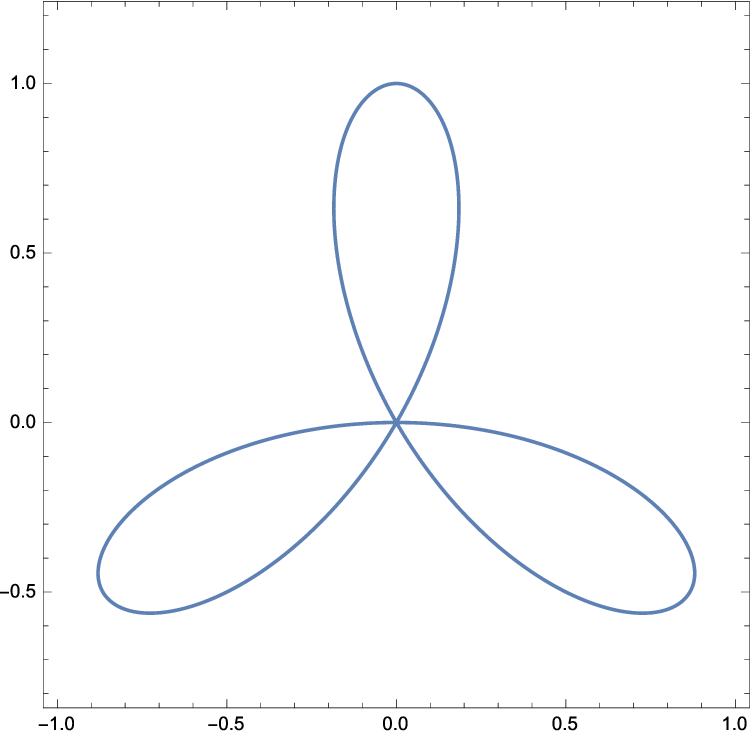}
\caption{A singular curve, $\left(x^2+y^2\right)^2+3 x^2 y-y^3=0$}
\label{graph_triple}
  \end{figure}
\end{ex}

\begin{ex}
  Computer $\TF_F$ for the double box on the maximal cut, where $F$ is the
  corresponding Baikov polynomial \eqref{F_7cut}. (We drop the
  subscript ``~${[7]}$''.)
  \begin{enumerate}
  \item Use \pmb{syz} in \Singular\  or \Macaulay, to compute $\TF_F$
    directly.
\item Note that $F$ has $3$ irreducible factors, $f_1=z_8$, $f_2=z_9$
  and $f_3=(s t-2 s z_8-2 s z_9-4 z_8 z_9)$. $f_1$ is linear so
  $\TF_{f_1}$ is generated by,
  \begin{equation}
    (z_8,0),\quad (0,1)\,.
  \end{equation}
Similarly, $\TF_{f_2}$ is generated by,
 \begin{equation}
    (1,0),\quad (0,z_9)\,.
  \end{equation}
What is $\TF_{f_1}\cap \TF_{f_2}$? Note that $f_3=0$ is smooth. Use
Proposition \ref{principle_syzygy} to find $\TF_{f_3}$.
\item Use \pmb{intersection} in \Singular\ or \Macaulay, to compute $
  \TF_{F} =\TF_{f_1}\cap \TF_{f_2}\cap \TF_{f_3}$. Compare the result
  with that from the direct computation. 

  \end{enumerate}

\end{ex}

\begin{ex}
  Consider three-loop massless triple box diagram (Figure. \ref{tribox})
  \begin{enumerate}
  \item Define $z_i=D_i$, $i=1,\ldots 10$, and 
    \begin{gather}
      z_{11}=(l_1 + k_4)^2,\quad z_{12}=(l_2 + k_1)^2,\quad z_{13}=(l_3 + k_1)^2,\nn\\ z_{14}=(l_3 + k_4)^2, \quad l_{15}=(l_1 + l_2)^2.
    \end{gather}
Determine its Baikov representation.
\item Derive IBPs with the maximal cut $D_1=\ldots=D_{10}=0$, and
  determine master integrals with $10$ propagators for this diagram. 
  \end{enumerate}
\end{ex}

\bibliographystyle{alpha}
\bibliography{AA2016.bib}

\end{document}